\documentclass[a4,12pt]{article}

\usepackage{times}
\usepackage{bm}
\usepackage{graphicx}
\usepackage{color}
\usepackage{verbatim}
\usepackage[colorlinks,citecolor=blue,urlcolor=blue,linkcolor=blue]{hyperref}

\usepackage{lipsum,changepage}

\usepackage{calc}
\newlength{\textwnew}
\usepackage[width=1.085\textwidth,font=small]{caption}

\usepackage{amsmath}
\usepackage{graphicx,psfrag,epsf}
\usepackage{enumerate}
\usepackage{natbib}
\usepackage{url} 
\usepackage{xr}
\usepackage{mathrsfs} 

\usepackage[french]{todonotes}
\RequirePackage{amssymb,amsfonts,dsfont,amsthm}
\RequirePackage[colorlinks,citecolor=blue,urlcolor=blue]{hyperref}
\RequirePackage[T1]{fontenc}

\usepackage{algorithm}
\usepackage{algorithm,algpseudocode}
\newcommand{\algorithmiccontinue}{\textbf{continue}}
\newcommand{\algorithmicbreak}{\textbf{break}}
\usepackage{caption}

\usepackage[french]{todonotes}
\usepackage{soul}
\setstcolor{red}
\usepackage[latin1]{inputenc}
\RequirePackage[T1]{fontenc}
\RequirePackage[french,english]{babel}           %
\RequirePackage{amsmath,amssymb,amsfonts,dsfont,amsthm}
\RequirePackage{natbib}
\RequirePackage[colorlinks,citecolor=blue,urlcolor=blue]{hyperref}
\usepackage{multirow}

\usepackage{array}
\usepackage{ragged2e}
\newcolumntype{R}[1]{>{\raggedright\let\newline\\\arraybackslash\hspace{0pt}}m{#1}}
\usepackage{colortbl}

\numberwithin{equation}{section}
\theoremstyle{plain}
\newtheorem{definition}{Definition}[section]
\newtheorem{thm}{Theorem}[section]

\newtheorem{lemma}{Lemma}[section]
\newtheorem{rem}{Remark}

\usepackage{pgf,tikz}
\usetikzlibrary{arrows}
\definecolor{ttccqq}{rgb}{0.2,0.8,0}
\definecolor{zzttqq}{rgb}{0.6,0.2,0}
\definecolor{qqwuqq}{rgb}{0,0.39,0}
\definecolor{qqqqff}{rgb}{0,0,1}
\definecolor{uququq}{rgb}{0.25,0.25,0.25}
\usepackage{pstricks-add}

\newcommand{\Bg}{\mathfrak{C}}
\newcommand{\Cc}{\mathcal{C}}

\newcommand{\HH}{\mathcal{H}}
\newcommand{\Hh}{\mathcal{H}}
\newcommand{\Ii}{\mathcal{I}}

\newcommand{\NN}{\mathbb{N}}

\newcommand{\Pp}{\mathcal{P}}
\newcommand{\RR}{\mathbb{R}}
\newcommand{\Rr}{\mathcal{R}}
\newcommand{\Ss}{\mathcal{S}}
\newcommand{\Uu}{\mathcal{U}}
\newcommand{\XX}{\mathbb{X}}
\newcommand{\Xx}{\mathcal{X}}
\newcommand{\YY}{\mathbb{Y}}
\newcommand{\Yy}{\mathcal{Y}}
\newcommand{\ZZ}{\mathbb{Z}}
\newcommand{\curvedepth}{curve depth}
\newcommand{\Curvedepth}{Curve depth}
\newcommand{\CurveDepth}{Curve Depth}
\newcommand{\indicatorf}{\mathds 1}

\newcommand{\spacefun}{\mathscr{C}([0,1],\mathbb{R}^d)}

\newcommand{\bmx}{\boldsymbol{x}}
\newcommand{\bmy}{\boldsymbol{y}}
\newcommand{\bmz}{\boldsymbol{z}}
\newcommand{\bma}{\boldsymbol{a}}

\begin{document}

\raggedbottom

\title{Depth for Curve Data and Applications}
\author{Pierre Lafaye de Micheaux\hspace{.2cm} \\
	School of Mathematics and Statistics, UNSW Sydney \\
	and \\
	Pavlo Mozharovskyi \\
    LTCI, T\'{e}l\'{e}com Paris, Institut Polytechnique de Paris \\
    and \\
    Myriam Vimond \\
    Univ Rennes, Ensai, CNRS, CREST - UMR 9194
}
\date{February 20, 2020}
\maketitle

\abstract{
John W. Tukey (1975) defined statistical data depth as a function that determines centrality of an arbitrary point with respect to a data cloud or to a probability measure. 
During the last decades, this seminal idea of data depth evolved into a powerful tool proving to be useful in various fields of science. Recently, extending the notion of data depth to the functional setting attracted a lot of attention among theoretical and applied statisticians. We go further and suggest a notion of data depth suitable for data represented as curves, or trajectories, which is independent of the parametrization. We show that our curve depth satisfies theoretical requirements of general depth functions that are meaningful for trajectories. We apply our methodology to diffusion tensor brain images and also to pattern recognition of hand written digits and letters. Supplementary Materials are available online.
}
\indent\\


\sloppy

\section{Introduction}\label{sec:intro}

We propose an extension of the notion of depth for curve data. 
An (unparameterized) curve datum is a set of points of $\RR^d$ which can be described by an unspecified continuous function from a sub-interval of $\RR$ to $\RR^d$. Our original motivation to study such data was to solve a neuroimaging problem involving brain fibers of elderly twins. 

Data depth was originally introduced in a seminal paper by \citet{tukey1975mathematics} to measure the degree of centrality of a
multivariate point $x$ with respect to a given data cloud. His approach consists in computing, for every halfspace $H$ containing $x$, the fraction of points from the data cloud enclosed in~$H$. He then retains the minimum of these fractions as a measure of centrality of $x$; see also \cite{donoho1992}. Since then, several other notions of data depth have been proposed. For example, using random simplices (i.e., generalizations of the notion of a triangle to arbitrary dimensions), \cite{liu1990} proposed a similar measure of ``insideness''  called simplicial depth. For a comprehensive survey on multivariate data depths the reader is referred to \cite{zuo2000general}.  

Thanks to these theoretical developments, it has become possible to extend standard univariate descriptive statistics based on ranks to analyze multivariate observations \citep[see, e.g.,][]{oja1983,liu1999}. New classical inferential statistical tools or techniques using these depth measures or some refinements have also been developed, such as $p$-values \citep{liu1997}, confidence regions \citep{yeh1997,lee2012}, regression \citep{rousseeuw1999,hallin2010}, multivariate nonparametric testing \citep{li2004,zuo2006,chenouri2012nonparametric}, classification \citep{li2012,lange2014fast,paindaveine2015,dutta2016} and estimation of extreme quantiles \citep{he2017}.
See \cite{mosler2013} for a nice introduction showing the richeness and usefulness of depth techniques.

In recent years, statisticians have been facing complex types of data that they analyze using 
a functional depth \citep{fraiman2001trimmed,lopez2009concept,narisetty2016} or even a multivariate functional depth approach \citep{claeskens2014multivariate}. These new techniques have proven to be very useful for data visualization, to estimate a measure of location or spread, to detect outliers \citep[see also][]{hubert2015}, for clustering, 
or to detect if two groups of functions come from the same population. 

However, functional depths are sensitive to parametrization of curves. 
Figures~\ref{fig:Sletters:comparison}~and~\ref{fig:simulated-hurricanes:comparison} illustrate the impact of two different parametrizations on depths rankings of curves provided by the {\it multivariate functional halfspace depth} (MFHD) developped by \cite{claeskens2014multivariate} (with weight function set to a constant) and by the {\it modified simplicial band depth} (mSBD) developped by \cite{lopezpintado2014}. 

\begin{figure}[H]
\begin{center}
	\begin{tabular}{ccccc}
	{\small ~~MFHD -- par. A} & {\small ~~MFHD -- par. B} & {\small ~~mSBD -- par. A} & {\small ~~mSBD -- par. B} & {\small ~~Curve Depth} \\
          \hspace*{-0.2cm}\includegraphics[trim = 1cm 1cm 1cm 1.5cm,clip=true,width=0.16\textwidth,page=1]{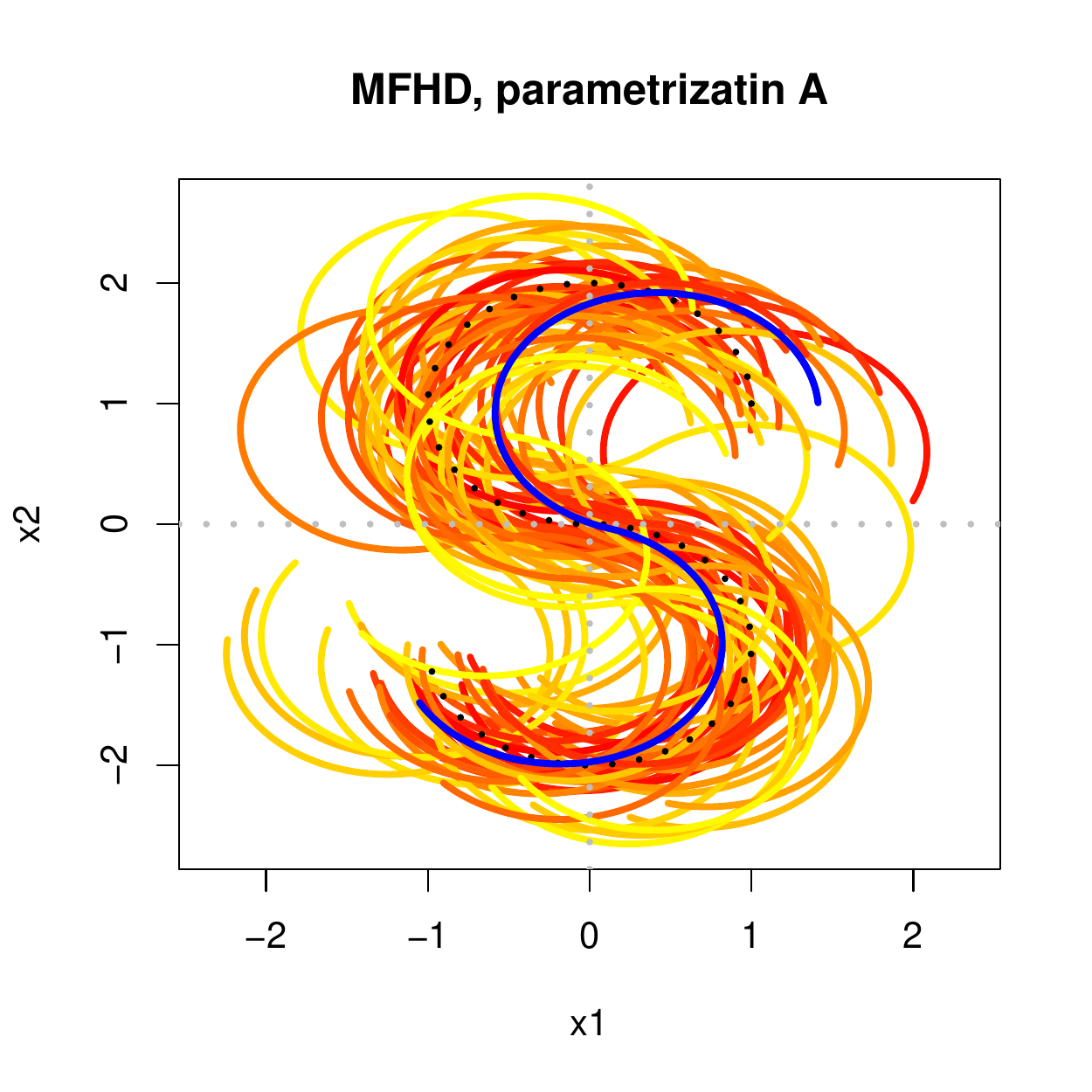} & \hspace*{-0.2cm}\includegraphics[trim = 1cm 1cm 1cm 1.5cm,clip=true,width=0.16\textwidth,page=2]{pic-S-2D.pdf} & \hspace*{-0.2cm}\includegraphics[trim = 1cm 1cm 1cm 1.5cm,clip=true,width=0.16\textwidth,page=3]{pic-S-2D.pdf} & \hspace*{-0.2cm}\includegraphics[trim = 1cm 1cm 1cm 1.5cm,clip=true,width=0.16\textwidth,page=4]{pic-S-2D.pdf} & \hspace*{-0.2cm}\includegraphics[trim = 1cm 1cm 1cm 1.5cm,clip=true,width=0.16\textwidth,page=5]{pic-S-2D.pdf}\\
          (a) & (b) & (c) & (d) & (e)
	\end{tabular}
\end{center}
	\caption{Comparison of depth based ordering for two parametrizations A and B provided respectively by MFHD (a)--(b), mSBD (c)--(d), and by our new depth for unparameterized curves (e). The depth increases from yellow to red. Each deepest curve is plotted in blue. The center of symmetry of the distribution is plotted using black dots. Source: an ensemble of 50 simulated S letters; see Section \ref{app:illustration:S} in Supplementary Materials.}\label{fig:Sletters:comparison}

\end{figure}

In Figure \ref{fig:Sletters:comparison} (a)-(d), we see that the choice of a parametrization (A or B) has a clear impact on which curve is identified as the deepest (in blue). Moreover, unlike MFHD and mSBD, our unparameterized approach finds a deepest curve which is very close to the center of symmetry (the dotted curve). Also, we observe that some curves with high depth (in red) seem to be outliers (Figure \ref{fig:Sletters:comparison} (a) and (c), upper right) and some curves with low depth (in yellow) are close to the deepest curve (Figure \ref{fig:Sletters:comparison} (b) and (d)). This problem is even more striking on Figure \ref{fig:simulated-hurricanes:comparison}. There, many simulated hurricane tracks are identified as outliers (in red, on panels (a)--(d)) by MFHD and mSBD (with two different parametrizations) even if they are close to the center of distribution of the curves (in dark blue). This is in agreement with \citep[Section 5]{mirzargar2014} who note that ``\textit{the time-parameterization is more sensitive to the velocity outlier as a parameterization-dependent feature, the arc-length and life-time percentage parameterization are more sensitive to shape and positional outliers.}'' Here again our unparameterized approach correctly identifies outliers (panel (e)).

\begin{figure}[H]
\begin{center}
	\begin{tabular}{ccccc}
	{\small ~~MFHD -- par. A} & {\small ~~MFHD - par. B} & {\small ~~mSBD -- par. A} & {\small ~~mSBD -- par. B} & {\small ~~Curve Depth} \\
          \includegraphics[width=0.155\textwidth,page=1]{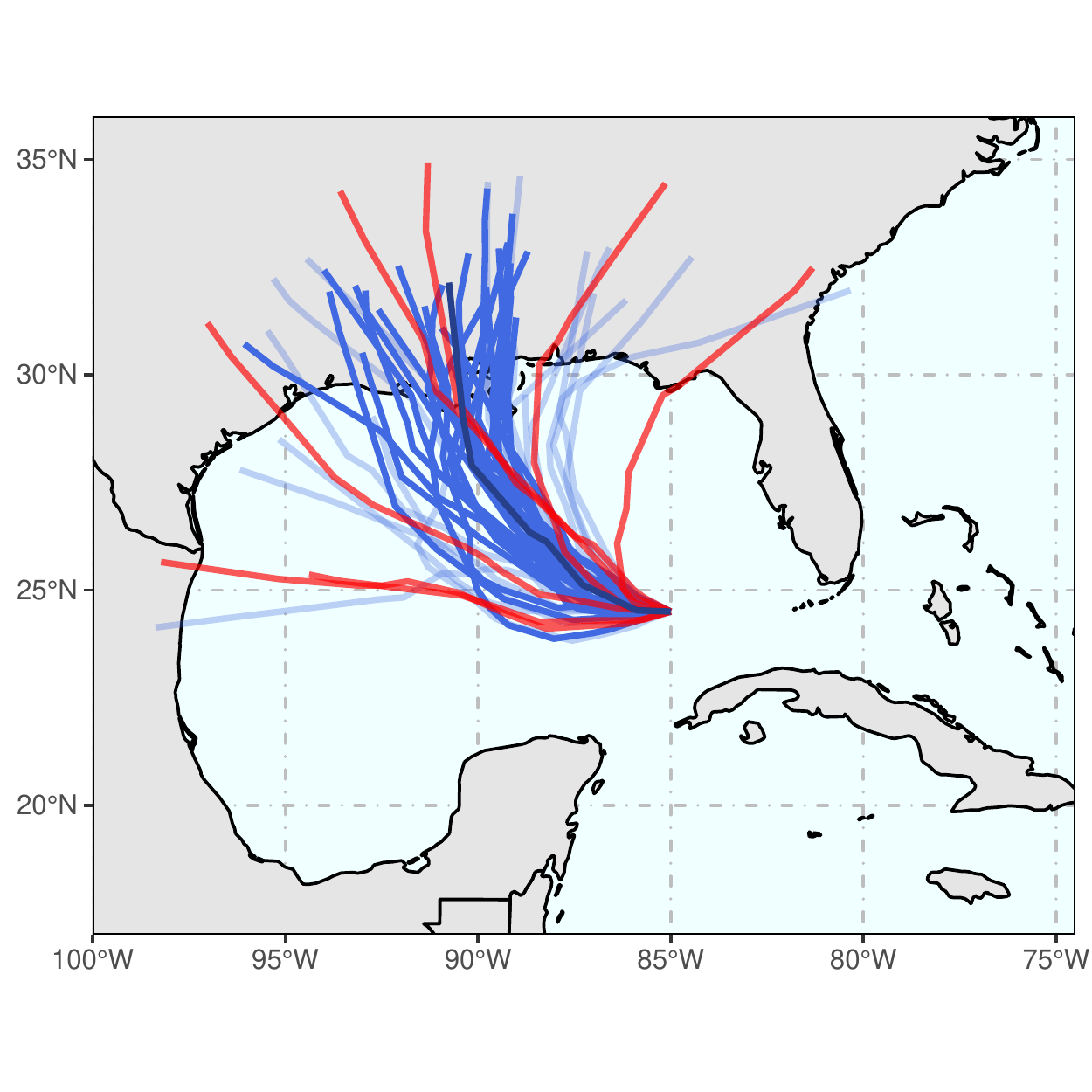}&\includegraphics[width=0.155\textwidth,page=2]{pic-uragans-2D.pdf}&\includegraphics[width=0.155\textwidth,page=3]{pic-uragans-2D.pdf}&\includegraphics[width=0.155\textwidth,page=4]{pic-uragans-2D.pdf}&\includegraphics[width=0.155\textwidth,page=5]{pic-uragans-2D.pdf}\\
                    (a) & (b) & (c) & (d) & (e)
	\end{tabular}
\end{center}
	\caption{Comparison of the depth based ordering for two parametrizations A (time) and B (arc-length) provided respectively by MFHD (a)--(b), mSBD (c)--(d), and by our new depth for unparameterized curves (e). Curves with low value of depth are plotted in red, the others in blue. Each deepest curve is plotted in dark blue. Source: an ensemble of 50 simulated hurricane tracks \citep{mirzargar2014}.}\label{fig:simulated-hurricanes:comparison}
\end{figure}

Note that MFHD and mSBD depths are computed by comparing each point on a given curve only to points (from the other curves) that ``occur at the same time''. Curves are thus compared pointwisely and not globally (this is a direct consequence of parametrization). We believe this is the cause of the aforementioned artefacts.

Of course, depending on the context, working with a proper parametrization of curves can be relevant. For instance, if available, one could use speed of writing as a meaningful parametrization in a handwriting recognition problem; see Section \ref{cursiveexample} in Supplementary Materials. 
For further discussion on the importance and possible choices of a proper parametrization when employing functional data depth, see, e.g., \cite{lopezpintado2014,mirzargar2014} and references therein.

%

In this paper we aim to define a depth which is invariant to the choice of a parametrization of the curves. This was originally motivated by the need to analyze a very large number of bundles of white matter fibers obtained through diffusion tensor imaging (an MRI-based neuroimaging technique) among a population of elderly twins. These neuronal fibers, also called axons, are nerve cell extensions that transmit electrical information between different regions of the brain. The aim of the study was to investigate if genetics plays a role in the spatial organization of these fibers.

In our setting, a mathematical curve describing a given fiber should be understood as the set of all points that describe the location in space of one of these fibers, {\it with no focus whatsoever on any parametrization}. Indeed, as outlined by \cite{kurtek2012} ``a parameterization is merely for the convenience of analysis and is not an intrinsic property of a curve'' which leads them to advocate that ``the shape analysis should be invariant not only to rigid motions and global scalings, but also to their parameterizations''. Intuitively, we want to distinguish curves solely by how they bend and twist, as well as by their lengths and relative locations in space. Consequently, the concept of functional data depth should not be used here (this is further investigated in Section~\ref{ssec:realdata-brain}).

One could think of using one of a few other existing approaches that deal specifically with curves. \cite{goldie1995} considered 2D observation records that are joined in a sequence, while \cite{sangalli2009} estimated centreline curves (and their curvature functions) of internal carotid artery vessels using three-dimensional free-knot regression splines. Unfortunately, these two methods also rely on some parametrization. \cite{mani2010} and \cite{kurtek2012} use a Riemannian framework invariant to the parametrization while \cite{zhang2015} developed a Bayesian version; see also \cite{srivastava2016} for a monograph on the statistical analysis of the shapes of curves. However, it is difficult to find a software to apply these methods on our data.


With this motivation in mind, we developed a new concept of {\it depth for curves} that is invariant to the choice of the parametrization. It will be broadly applicable, thanks to our freely available \texttt{R/C++} package \texttt{curveDepth} \citep{curveDepth}, to many other similar types of data. On can mention a few examples such as textile fibers \citep{xu2001}, blood clot fibers \citep{collet2005}, blood vessels centrelines \citep{sangalli2009}, moving objects such as birds migrating \citep{su2014,yuan2017}, multidimensional data sets obtained by constructing principal curves \citep{hastie1989}. 

The outline of the paper is as follows. In Section~\ref{sec:model}, curves are defined formally and we introduce a statistical model for sampled curves. 
Section~\ref{sec:UnparamCurves} contains a definition of the new data depth for curves.
In Section~\ref{sec:Implementation}, we discuss implementation issues. In Section~\ref{sec:realdata}, we present simulation results. We also apply our curve depth to analyze brain imaging data sets, and to classify hand-written digits. There, our curve depth is compared to other existing depths, namely MFHD \citep{claeskens2014multivariate}, the {\it modified multivariate band depth} (mMBD) of \citet{IevaP13}, the multivariate functional {\it skew-adjusted projection depth} (saPRJ) of \citet{hubert2015}, the {\it simplicial band depth} (SBD) of \citet{lopezpintado2014} and its modified version (mSBD).

Section~\ref{sec:Conclusion} gathers some concluding remarks. Supplementary Materials collect all technical proofs, along with the necessary codes and data to reproduce all our numerical and graphical results.

\section{A Statistical Model for Sampled Curves}\label{sec:model}

In what follows we introduce the space of unparameterized curves and define a statistical model on it. 
For a comprehensive reference the reader is referred to \citet[Section~2]{Kemppainen2017} which borrowed material from \citet[Section~2.1]{aizenman1999holder} and \citet[Section~2.5]{Burago2001}.
For additional details see Section~\ref{app:curves} in the Supplementary Materials. 

\subsection{The Space of Unparameterized Curves}

Let $d\geq1$ be an integer. Let $(\RR^d,|\cdot|_2)$ be the $d$-dimensional Euclidean space, $\spacefun$  be the space of continuous functions defined on the interval $[0,1]$ and taking values in $\RR^d$ and $\Gamma$ be the set of increasing continuous functions $\gamma:[0,1]\to [0,1]$ such that $\gamma(0)=0$ and $\gamma(1)=1$. 
A \textit{parameterized curve} $\beta$, also called a \textit{path}, is an element of $\spacefun.$ 
The image of $\beta,$ denoted as $S_\beta=\beta([0,1]),$ is called the \textit{locus}~of~$\beta$. 
Informally if $\beta(t)$ describes the position of a moving particle at time $t,$ then $S_\beta$ describes the physical route taken by this particle with no consideration being given to stops or goings backward occuring on its trajectory.
The function $\beta:[0,1]\mapsto\mathbb{R}^d$, a parametrization of $S_\beta$ with parameter $t$, provides an ordering along $S_\beta$. Note that there might exists an infinite number of different parametrizations describing the same locus.

\begin{rem}
The start point of $S_\beta$ is the image of $0$ by $\beta$. The end point is the image of $1$.
The \textit{locus} of a trivial curve coincides with a singleton, \textit{i.e.}, a single point of $\RR^d.$ 
\end{rem}

Formally, unparameterized curves are usually defined via an equivalence relation on the set of parameterized curves in $\RR^d$ up to the set of monotonic functions from $[0,1]$ to $[0,1]$. Roughly speaking, two curves $\beta_1$ and $\beta_2$ are said equivalent if they share the same locus and visit its points continuously and in the same order, possibly at a different speed. Hereafter, we restrict ourselves to the set of all curves equivalent to $\beta$ that start at $\beta(0)$ and stop at $\beta(1)$. More precisely, we say that two parameterized curves $\beta_1$ and $\beta_2$ are equivalent whenever there exist two reparametrizations $\gamma_1,\gamma_2\in\Gamma$ such that $\beta_1\circ\gamma_1=\beta_2\circ\gamma_2$. We then define the \emph{unparameterized curve} $\Cc_\beta$ as the set of all paths equivalent to $\beta$, that is the equivalence class of $\beta$ up to this equivalence relation. Informally, $\Cc_\beta$ describes the trajectory from $\beta(0)$ to $\beta(1)$, with no information about the location at any time.  
Note that in our context, it would be possible to consider the general definition, i.e., to walk a path $\beta$ from $\beta(1)$ to $\beta(0)$, or the other way around. But restricting all our definitions by considering the set of parameterized curves in $\RR^d$ only up to the set of reparametrizations $\Gamma$ greatly simplifies exposition; see Remark \ref{rem:order}. In the sequel, an unparameterized curve will be generically denoted $\Cc$. Notice that all parameterized curves in the same equivalence class $\Cc$ share the same locus, which enables one to talk about the locus of $\Cc$, denoted thereafter as $S_\Cc$.

The \textit{space of unparameterized curves} is then defined as
$$
\Bg = \{\Cc_{\beta}:\beta\in\spacefun\}.
$$
In other words, $\Bg$ is the quotient space of $\spacefun$ by the equivalence relation on the set of parameterized curves

Following \cite{Kemppainen2017}, we endow the space of curves $\Bg$ with the Fr\'{e}chet metric $d_\Bg$ defined as
\begin{equation}\label{eq:distance}
d_\Bg\left(\Cc_1,\Cc_2\right) = \inf\left\{\|\beta_1-\beta_2\|_\infty;~\beta_1\in\Cc_1,\ \beta_2\in\Cc_2\right\},\quad\Cc_1,\Cc_2\in\Bg,
\end{equation}
where $\|\beta\|_{\infty} = \sup_{t\in[0,1]}|\beta(t)|_2$ for $\beta\in\spacefun$.
The resulting metric space $(\Bg,d_\Bg)$ is non linear. It inherits the properties of separability and completeness from $\spacefun$; see Section~\ref{app:curves:metric} in the Supplementary Materials. This guarantees the existence of non-atomic probability measures on $(\Bg,d_\Bg)$. Moreover, according to \citet[Theorems~1.2,~3.2~and~8.1]{parthasarathy1967probability}, every probability measure defined on $\Bg$ is regular and tight.

\subsection{The Arc-Length Probability Measure of a Curve}
The length $L(\beta)$ of a parameterized curve $\beta\in\Cc$ is defined as
\begin{equation}\label{eq:length}
L(\beta)=\sup_{\tau}\left\{L_\tau(\beta);~\tau \text{ a partition of $[0,1]$} \right\},
\end{equation}
where $L_\tau(\beta)=\sum_{j=1}^J|\beta(\tau_j) - \beta(\tau_{j-1})|_2$ is the {\it chordal length}\/ of $\beta$ associated with the partition $\tau=\{\tau_0,\ldots,\tau_J;~0=\tau_0<\cdots<\tau_J=1,~J\in\mathbb{N}^*\}$.
Informally $L(\beta)$ is the total distance travelled by a particle moving from $\beta(0)$ to $\beta(1)$ along the support $\Ss_\beta$ of the curve~$\beta$ (taking into account any backward steps). 
Then all parameterized curves in $\Cc$ have the same length. Consequently, the length of $\Cc$, denoted $L(\Cc)$, is defined by $L(\Cc)=L(\beta)$, for any $\beta\in\Cc$.
Note that the function $L:\Bg\to[0,+\infty]$ is not continuous, but it is measurable (Lemma \ref{lem:SpaceCurve:mesureQ}). 
In the following we assume that all unparameterized curves belong to the measurable set $\Bg_L=\{\Cc\in \Bg\,;~ 0<L(\Cc)<\infty\}\subset\Bg$, the subset of \textit{rectifiable} (i.e., of finite length) unparameterized curves with a positive length.


According to \citet[Theorem 2.4]{vaisala2006lectures}, each curve $\Cc\in\Bg_L$ contains a unique parametrization $\beta_\Cc:[0,1]\to\RR^d$, called \emph{the arc-length parametrization}, whose restrictions to the intervals $[0,t]$, noted $\beta_\Cc^t$, satisfy $L({\beta_\Cc^t})=tL(\Cc),$ for all $t\in[0,1]$. Informally, with $\beta_\Cc$, the locus $S_\Cc$ is visited at a constant speed. Then any rectifiable curve $\Cc$ may be expressed as
$$
\Cc=\{\beta_\Cc\circ\gamma;~ \gamma\in\Gamma\}.
$$

Using the arc-length parametrization $\beta_\Cc$ of an unparameterized curve $\Cc$, one can thus define the \emph{line integral} of a non-negative Borel function $f:\RR^d\to\RR$ over $\Cc$ as
\begin{equation}\label{eq.line.integral}
\int_{\Cc} f(s) ds := \int_{0}^{1}f\left(\beta_\Cc(t)\right) L(\Cc)dt,
\end{equation}
where the integral on the right is a Riemann integral. 
Furthermore, we define the \textit{arc-length probability measure} of $\Cc$ as the probability distribution $\mu_{\Cc}$ on the Borel sets of $\RR^d$:
\begin{equation}\label{eq.def.muC}
\text{for any borel set $A$ of $\RR^d,$}\qquad \mu_\Cc(A) = \frac{1}{L(\Cc)}\int_\Cc\mathds{1}_A(s) ds\,,
\end{equation}
where the indicator function $\mathds{1}_A(x)$ takes the value 1 if $x\in A$ and 0 otherwise.

From \eqref{eq.line.integral} and \eqref{eq.def.muC}, we immediately get
\begin{equation}\label{eq.line.integral.with.mu}
  \int_\Cc f(s)d\mu_\Cc(s) = 
  \int_{0}^1 f(\beta_\Cc(t))dt.
\end{equation}
Also, note that $\mu_{\Cc}$ only contains information about the support $S_{\Cc}$ of $\Cc$ and the frequency at which its points are visited.  Roughly speaking, $\mu_\Cc(A)$ can be interpreted as a ratio: the distance travelled by a particle on the subset $S_{\Cc}\cap A$ divided by the total distance it travels on $\Ss_\Cc$. (Note that $L(\Cc)$ can be different from the length of $\Ss_\Cc$.) It is somehow a normalised measure of how much of curve $\Cc$ intersects with $A$.



\subsection{A Nonparametric Statistical Model for a Sample of Curves}

We denote by $\Pp$ the set of all probability measures defined on the Borel $\sigma$-algebra of the Borel sets of $(\Bg,d_\Bg)$ whose support is a subset of rectifiable 
curves of positive length (to exclude singletons): 
$$\Pp=\Big\{P,\,\text{a probability measure on }(\Bg,d_\Bg)\ ;~ P\left(\Bg_L\right)=1\Big\}.$$
Consider a random unparameterized curve $\Xx,$ namely a random element taking ``values'' in the space of unparameterized curves $\Bg,$ whose probability distribution $P\in\Pp$ is unknown. 
We define the probability distribution $Q_{P}$ as follows: 
\begin{equation}\label{eq.def.QP}
  \text{for all borel sets $A$ of $\RR^d,$}\qquad Q_{P}(A) = \int_\Bg \mu_\Cc(A)dP(\Cc)=E_P[\mu_\Xx(A)],
\end{equation}
a measure of how much (on average) a curve generated by $\Xx$ intersect with~$A$.

\begin{rem}\label{rem:Q}
In Section~\ref{app:curves:integral} in the Supplementary Materials, we show that for any Borel bounded function $f:\RR^d\to\RR$, the function $\Cc\in\Bg_L\mapsto \int_\Cc f d\mu_\Cc \in\RR$ is measurable. 
Consequently, $Q_P$ is well-defined.
\end{rem}

The statistical model considered in this article is to assume that the data to be observed are $n$ random unparameterized curves $\Xx_{1},\ldots,\Xx_{n}$, which are independent copies of the random element $\Xx$, that is to say
\begin{equation}\label{eq:model}
\Xx_{1},\ldots,\Xx_{n} \text{ are i.i.d.\ from $P\in\Pp$}.
\end{equation}
In the next section, we define a population data depth for unparameterized curves, and its sample version.

\section{Data Depth for Unparameterized Curves}\label{sec:UnparamCurves}

\subsection{Population and Sample Versions}

Let $y^\top$ denote the transpose of the column vector $y\in\RR^d$ and $\Ss$ be the unit-sphere in $\RR^d$. 
For a pair $(u,x)\in\Ss\times\RR^d$, let $H_{u,x}$ denote the closed halfspace $\{y\in\RR^d\ : y^\top u \geq x^\top u \}$ whose frontier is orthogonal to the vector $u$ and goes through the point $x$.
Notice that if $d=1,$ the unit-sphere is $\{-1,1\}$.

\begin{definition}[\Curvedepth, population version]
Let $\Cc\in\Bg_L$ be an unparameterized curve and let $P\in\Pp$ be a probability measure. 
We define the \emph{\curvedepth} of $\Cc$ w.r.t.\ $P$, denoted $D(\Cc|P)$, by the mapping 
\begin{eqnarray}\label{eq:depth:curve}
D: & \Bg_L\times\Pp & \to\RR\nonumber \\
   &   (\Cc,P)         & \mapsto D(\Cc|P) = \int_\Cc D(s|Q_P,\mu_\Cc) d\mu_\Cc(s),
\end{eqnarray}
where the above line integral is computed via \eqref{eq.line.integral} using, for any $d\geq1$ and any $x\in S_\Cc$,
\begin{align}\label{eq:depth:x}
D(x|Q_P,\mu_\Cc) &\!=\! \inf_{u\in\Ss}\frac{Q_P(H_{u,x})}{\mu_\Cc(H_{u,x})},
\end{align}
with the convention that $a/0=+\infty$ for all $a>0$ and $0/0=0$ in the above ratio.
\end{definition}

The term $D(x|Q_P,\mu_\Cc)$ aims to compare the two distributions $Q_P$ and $\mu_{\Cc}$ around $x\in S_\Cc$. For $u$ and $x$ fixed, recall from \eqref{eq.def.muC} and from \eqref{eq.def.QP} that $\mu_\Cc(H_{u,x})$ measures (the fraction of length of) how much the curve $\Cc$ delves into the halfspace $H_{u,x}$, whereas $Q_P(H_{u,x})$ measures (the expected fraction of length of) how much a random curve $\Xx$ (with distribution $P$) delves into $H_{u,x}$. Consequently, the ratio $Q_P(H_{u,x})/\mu_\Cc(H_{u,x})$ is small when we expect curves generated according to $P$ to enter less into $H_{u,x}$ than the curve $\Cc$. Getting a value $r>1$ (resp. $r<1$) for this ratio, indicates that $\Xx$ generates curves that enter into $H_{u,x}$, on average, $r$ times more (resp. $1/r$ times less) than $\Cc$ does; see Figure~\ref{fig:illustration} for a visual aid.

\begin{figure}[!ht]
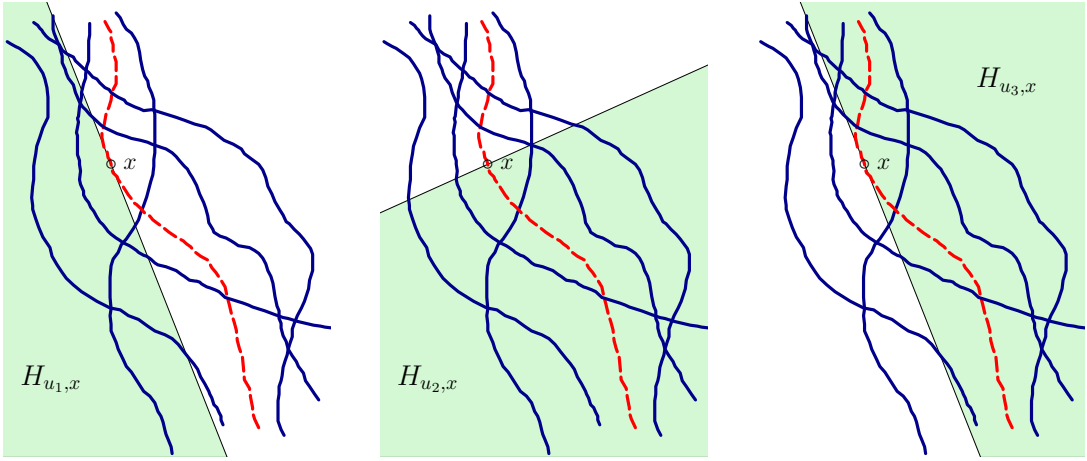

	\begin{center}
		\resizebox{0.29\linewidth}{!}{    
		\input{illustration1.tex} 
		}
		\resizebox{0.29\linewidth}{!}{    
		\input{illustration2.tex} 
		}
    	\resizebox{0.29\linewidth}{!}{    
		\input{illustration3.tex} 
		}
	\end{center}
	\caption{
          Illustrations of the statistical model and depth calculation \eqref{eq:depth:x} for three halfspaces with a sample of five curves generated by $\Xx$ in blue and the curve $\Cc$ in red. We consider all halfspaces whose frontier contains the point $x$ and pick up the smallest ratio of the probability measures between $Q_P$ and $\mu_\Cc$ : (left) $\frac{Q_P(H_{u_1,x})}{\mu_\Cc(H_{u_1,x})}=9.714$, (middle) $\frac{Q_P(H_{u_2,x})}{\mu_\Cc(H_{u_2,x})}=0.999$, (right) $\frac{Q_P(H_{u_3,x})}{\mu_\Cc(H_{u_3,x})}=0.618$}
	\label{fig:illustration}
\end{figure}

Then, similarly to the original Tukey depth, to obtain $D(x|Q_P,\mu_\Cc)$, we consider all possible rotations of the halfspace $H_{u,x}$ around $x$ to find the one that discriminates the most the curve $\Cc$ from a curve generated according to $P$.
We shall call $D(x|Q_P,\mu_\Cc)$ as the \textit{point \curvedepth{}} at $x\in S_\Cc$. 
Then \eqref{eq:depth:curve} defines the depth of $\Cc$ w.r.t.\ $P$ as the mean of the point \curvedepth{s} at all $x$ in its locus $S_\Cc$.

Notice that if there exists $u\in\Ss$ such that $Q_P(H_{u,x}) = 0,$ then $x$ is an outlier w.r.t.~$Q_P$, and thus the contribution of $x\in S_\Cc$ to the depth of $\Cc$ w.r.t.\ $P$ is set to zero, that is $D(x|Q_P,\mu_\Cc) = 0$.

 If $Q_P(H_{u,x})> 0$ for all $u\in\Ss,$ that means $x$ lies in the convex hull of the support of~$Q_P.$ 
Our aim is to calculate the depth of $x\in S_\Cc$ w.r.t.\ $Q_P$ relatively to the measure $\mu_\Cc,$ that is why we consider the ratio $Q_P(H_{u,x})/\mu_\Cc(H_u,x)$ in the definition of $D(x|Q_P,\mu_\Cc).$
In this case, we can show that there exists $u$ such that $\mu_\Cc (H_{u,x}) \geq Q_P(H_{u,x}) > 0$ (Lemma~\ref{lem:depth:bounded} in the Supplementary Materials), so that $x \mapsto D(x|Q_P,\mu_\Cc)$ is bounded by $1.$
Moreover, $x\mapsto D(x|Q_P,\mu_\Cc)$ is measurable as a limit of measurable functions (see Lemma~\ref{lem:depth:mc:conv-ps} in the Supplementary Materials).

\begin{definition}[\Curvedepth, sample version]
Let $\Xx_1,\ldots,\Xx_n$ be a random sample of unparameterized curves belonging to $\Bg_L$ a.s. and let $\Cc\in\Bg_L$ be a rectifiable unparameterized curve.
With a slight abuse of notation, and thanks to \eqref{eq.line.integral}, we define the \emph{\curvedepth} of~$\Cc$ w.r.t.\  $\Xx_1,\ldots,\Xx_n$ by the mapping 
{\small\begin{eqnarray}\label{eq:depth:emp:curve}
D: & \Bg_L\times \left\{\Bg_L\right\}^n & \to      \RR \\
   &   (\Cc,\Xx_1,\ldots,\Xx_n)         & \mapsto D(\Cc|\Xx_1,\ldots,\Xx_n) = \int_\Cc D(s|Q_n,\mu_\Cc) d\mu_\Cc(s) ,\nonumber
\end{eqnarray}}
where $Q_n = (\mu_{\Xx_1}^{}+\cdots+\mu_{\Xx_n}^{} )/n$ and $\beta_{\Cc}$ is the arc-length parametrization of  $\Cc$. 
\end{definition}


\begin{rem}
In a sense, our depth may be seen as a genaralization of the Tukey halfspace depth in $\RR^d.$
If $\Cc$ is a trivial curve, that is $L(\Cc)=0$ and $S_\Cc=\{y\}$ for some $y\in\RR^d,$ we define $\mu_\Cc$ as the dirac measure $\delta_y$ at $y$.
Then, if $\Xx_1,\ldots,\Xx_n$ are also trivial curves, that is $S_{\Xx_i}=\{x_i\}$, $i=1,\ldots,n$, we get
\begin{align*}
D(\Cc|\Xx_1,\ldots,\Xx_n) &= D(y|Q_n,\delta_y)\\ 
	& = \inf_{u\in\Ss} \frac{1}{n} \sum_{i=1}^n \mathds{1}_{x_i\in H_{u,y}}.
\end{align*}  
\end{rem}

Theorem~\ref{thm:consistency} below states that the sample version of the \curvedepth\ \eqref{eq:depth:emp:curve} converges in probability to the population version \eqref{eq:depth:curve} as $n\rightarrow\infty.$

\begin{thm}\label{thm:consistency}
Let $\Cc\in\Bg_L$ be an unparameterized curve such that $\mu_\Cc$ is non-atomic. 
Let $P$ be a probability measure in the space of unparameterized curves such that $P\in\Pp$ and $Q_P$ is non-atomic.
Then  the sample \curvedepth{} $D(\Cc|\Xx_1,\ldots,\Xx_n)$ converges in probability to $D(\Cc|P)$ as $n\to\infty.$
\end{thm}

\subsection{Properties}

The main aim of the proposed \curvedepth{} is to provide a meaningful statistical ordering of the observed curve data, which is experimentally studied and illustrated on real-data examples in Section~\ref{sec:realdata}. Theorem~\ref{thm:consistency} states the consistency of our sample depth under mild assumptions and in this subsection we discuss its properties.

Following the suggestion of \cite{liu1990} for simplicial depths, \citet{zuo2000general} have defined four properties to be satisfied by a proper multivariate depth function: affine invariance, maximality at the center of symmetry, monotonicity relative to the deepest point and vanishing at infinity. (For a slightly different version of the postulates see also \citet{Dyckerhoff2004data} and \citet{mosler2013}.) 
For a functional depth, \cite{nieto2016topologically} suggest that six properties need to be satisfied, but \citet{gibels2018general} argue that some of them could be demanding.

The situation appears to be even more challenging for the space of unparameterized curves. Indeed, (loci of) unparameterized curves can be seen as subsets of $\RR^d$ which are parameterized by paths up to the same order of visit of their points. These mathematical objects can thus be thought of as being ``between'' functional data and set data. Moreover, since no canonical mandatory postulates for a functional depth have been established yet, and since the existing postulates are mainly inherited from those for the multivariate depth function, we base the following analysis on the latter.

Since the length is an important characteristic of an unparameterized curve, similarity invariance, which is associated with a similarity group preserving orientation and ratio of lengths, seems to be more appropriate than affine invariance in our context.
Moreover, the space of unparameterized curves is not a vector space. For instance the surjection $\beta\mapsto \Cc_{\beta}$ is not linear (there is no natural way to define the addition of two unparameterized curves and thus no line segment between two unparameterized curves, a crucial point for the monotonicity property). It is thus not possible to extend the classical formulation of a depth using results from \cite{Dutta2011} or \cite{Mosler2018}, say. 
Similarly, there is no universal way to define a notion of symmetry for unparameterized curves, no symmetry center can be defined either.
 The vanishing at infinity property can be directly extended to the space of curves. Below we state the properties satisfied by our \curvedepth{} function and summarize them in Theorem~\ref{thm:properties}.



\paragraph{Boundness.}

Calculating the \curvedepth{} \eqref{eq:depth:curve} consists in integrating a non-negative function bounded by one w.r.t.\ a probability measure. This fulfills one of the basic requirements of a depth function: to take values on the unit interval.

\paragraph{Similarity invariance.}

For a multivariate depth, affine invariance is required for changelessness w.r.t.\ an affine change of the coordinate system. 
For the space of unparameterized curves, we consider affine transformations that also preserve ratios of the lengths of curves, i.e., similarities. (Note that the length of an unparameterized curve is a property of the equivalence class.) 
A similarity $f:\RR^d\to \RR^d$ is an affine transform, $f(x)=rAx+b$ such that $A$ is an orthogonal matrix, $r$ is a positive factor and $b\in\RR^d$ is a vector. 
In particular, for all $x$ and $y$ in $\RR^d$, we have $|f(x)-f(y)|_2=r |x-y|_2.$ 
We denote by $P_f$ the distribution of the image under $f$ of a stochastic process having a distribution $P$. 
A map $D$ satisfies the property of similarity invariance if for every rectifiable curve $\Cc$ and every similarity map $f:\RR^d\to\RR^d$, it holds $D(\Cc,P)=D(f\circ\Cc,P_f).$


\paragraph{Vanishing at infinity.}
The farther away an unparametrized curve is from a data cloud of curves, the smaller its depth should be. 
To formulate the vanishing at infinity property of our curve depth $D$, we consider any sequence $(\Cc_n)_n$ of curves in $\Bg_L$  such that $\mu_{\Cc_n}$ is a non-atomic measure for all $n$ and $\lim_{n\to\infty} d_{\Bg}(\Cc_n,0)=\infty$, where $0$ denotes the set of parametrized curves equivalent to the constant curve $t\mapsto\beta(t)=0$ for all $t\in[0,1]$. 
However, such a formulation involves sequences of curves whose length tends to infinity. To exclude these cases, we assume that there exists some $\ell>0$ such that $L(\Cc_n)<\ell$ for all $n$. This guarantees that only the location of these curves tends to infinity. We then prove that
$$ \lim_{n\to\infty} D(\Cc_n, P)=0.
$$

\begin{thm}\label{thm:properties}
Under the assumptions of Theorem~\ref{thm:consistency},	our \curvedepth{} is a depth function in $\Bg_L$, i.e., it takes values in $[0,1]$, is similarity-invariant and is vanishing at infinity.
\end{thm}

\section{Implementation}\label{sec:Implementation}

Even if the curves $\Xx_1,\ldots \Xx_n$ and $\Cc$ are known, it may not be possible to obtain explicit expressions of $\mu_\Cc(H)$ and $Q_n(H)$ for an arbitrary halfspace $H$. This might prevent one to compute a value for \eqref{eq:depth:emp:curve} (via \eqref{eq:depth:x}). In fact, it appears that computation of the point \curvedepth{} ${D}(x|{Q}_{n},{\mu}_{\Cc}),$ $x\in S_{\Cc}$, in \eqref{eq:depth:emp:curve} demands algorithmic elaboration. 
We describe in the Supplementary Materials (Section \ref{app:ssec:montecarlo}) a Monte Carlo scheme to approximate $D(\Cc|\Xx_1,\ldots,\Xx_n)$. This is summarized in Algorithm \ref{alg:montecarlo}.

The main idea is to generate $3$ samples. First, a sample of size $m$ is used in order to approximate $\mu_{\Cc}$ (Line~2). Next, a (stratified) sample of size $nm$ is used to approximate the $\mu_{\Xx_i}$ (Lines~3--5) and $Q_n$ (Line~6). (See Lemma~\ref{lem:SpaceCurve:approxmu} in the Supplementary Materials for the procedure to generate these samples.) These are the two ingredients involved in the approximation of ${D}(x|{Q}_{n},{\mu}_{\Cc})$.

The last sample (Line~7) consists of points generated along the curve $\beta_\Cc$. It is used to approximate the line integral of ${D}(\cdot|{Q}_{n},{\mu}_{\Cc})$ with respect to $\mu_{\Cc}$ (Line~11). 
A Monte Carlo approximation of \eqref{eq:depth:x} is obtained (Lines~8--10) by using an adaptation of a minimization algorithm from \citep{rousseeuw1996HD2} for dimension $2$ and one from \citep{dyckerhoff2016HD} for higher dimensions; see our Algorithms~1 and 2 in Section~\ref{app:ssec:calcdepth} in the Supplementary Materials. These original algorithms were developed for the computation of the multivariate Tukey depth. They need to be adapted to our context as follows. Given that we are looking to estimate a ratio whose denominator can be arbitrarily small, we introduce a threshold $\Delta$ in order to control the stochastic convergence of the proposed algorithm (see Theorem~\ref{thm:consistency0} in the Supplementary Materials). Formal algorithms for dimensions $2$ and $3$ are stated and described in the Supplementary Materials (Section~\ref{app:ssec:calcdepth}). The latter can be easily extended to higher dimensions.

\noindent\hspace*{0\parindent}%
\begin{minipage}{0.875\textwidth}
\begin{algorithm}[H]
\caption{Monte Carlo approximation of $D(\Cc|\Xx_1,\ldots,\Xx_n)$ in \eqref{eq:depth:emp:curve}; $m$ denotes the Monte Carlo sample size of points generated uniformly on each curve (taking into account their length); $\Delta_m$ is a threshold parameter such that $\Delta_m\underset{m\rightarrow\infty}{\longrightarrow}0$; $\HH_{\Delta_m}^{n,m}$ is the collection of closed halfspaces $H$ such that either $\widehat{Q}_{m,n}(H)=0$ or $\widehat{\mu}_{m}(H)>\Delta_m$.}\label{alg:montecarlo}
\begin{algorithmic}[1]
\Procedure{MCapprox}{$\Cc,\Xx_1,\ldots,\Xx_n,m,\Delta_m$}
\State{Generate $Y_1,\ldots,
  Y_m$ \text{ i.i.d.\ from $\mu_{\Cc}$} and set $\widehat{\mu}_{m}=m^{-1}\sum_{j=1}^m\delta_{Y_j}$.}
\For{$i=1:n$}  
	\State{Generate $X_{i1},\ldots, X_{im}$ \text{ i.i.d.\ from $\mu_{\Xx_i}$ and set $\widehat{\mu}_{\Xx_i}=m^{-1}\sum_{j=1}^m\delta_{X_{i,j}}$.} } 
        \EndFor
\State{Set $\widehat{Q}_{m,n}=n^{-1}\sum_{i=1}^n\widehat{\mu}_{\Xx_i} = (nm)^{-1}\sum_{i=1}^n\sum_{j=1}^m\delta_{X_{i,j}}$ to estimate \eqref{eq.def.QP}.}
\State{Generate (independently from the $Y_j$'s) $Z_1,\ldots, Z_m$ \text{ i.i.d.\ from  $\mu_{\Cc}$.}}  
\For{$k=1:m$}  \Comment{An approximation of $D(Z_k|Q_n,\mu_\Cc)$ from \eqref{eq:depth:x}}
	\State{Compute $\widehat{D}(Z_k|\widehat{Q}_{m,n},\widehat{\mu}_{m},\HH_{\Delta}^{n,m})$ as the smallest ratio of $\widehat{Q}_{m,n}(H_{u,Z_k})$ to $\widehat{\mu}_{m}(H_{u,Z_k})$ over a (random or deterministic) grid of points $u\in\Ss$ selected in such a way that $H_{u,Z_k}\in\HH_{\Delta_m}^{n,m}$. See Algorithms~1 and 2 in Section~\ref{app:ssec:calcdepth} in the Supplementary Materials for a way to build a grid achieving exactly the infimum in $\mathbb{R}^2$ or $\mathbb{R}^3$.}
\EndFor\\
\Return{$m^{-1}\sum_{k=1}^{m} \widehat{D}(Z_{k}|\widehat{Q}_{m,n},\widehat{\mu}_{m},\HH_{\Delta}^{n,m})$ as an estimate of \eqref{eq:depth:emp:curve}.}
\EndProcedure
\end{algorithmic}
\end{algorithm}
\end{minipage}

Overall, time complexity is $O(m^dn^{d-1}\log(mn))$ if ${D}(x|{Q}_{n},{\mu}_{\Cc})$ is computed exactly, where $n$ is the size of the sample of curves, and $m$ is the size of the Monte Carlo sample of points which are sampled on each curve involved in the depth computation. Time complexity is $O(km^2n)$ if ${D}(x|{Q}_{n},{\mu}_{\Cc})$ is approximated using projections on $k$ random directions (i.e., the minimum ratio in Step~9 of Algorithm~\ref{alg:montecarlo} is searched over $k$ random directions $u$ only).

In \eqref{eq:distance}, we introduced the Fréchet distance $d_\Bg(\Cc_1,\Cc_2)$ between any two curves $\Cc_1$ and $\Cc_2$ belonging to the space of curves $\Bg$. This distance will be useful for two applications of Section~\ref{sec:realdata}; namely for curve registration in the brain and also for an adaptation of the unsupervised classification method of \citet{jornsten2004clustering}. When calculating $d_\Bg(\Cc_1,\Cc_2)$, one has to search for a parameterized curve in $\Cc_1$ and a parameterized curve in $\Cc_2$ that are as close as possible, in terms of their supremum distance. Numerically, this can be done as follows. Consider a set of points on $\Cc_1$ and a corresponding relocation of each one of these points to $\Cc_2$, preserving their ordering. The goal is to minimize the largest Euclidean distance between any one point on $\Cc_1$ and any of its relocated counterpart on $\Cc_2$. The formal algorithm together with an illustrative explanation is stated in the Supplementary Materials (Section~\ref{app:ssec:calcdist}).

Numerical computation of our \curvedepth{} and of the above-mentioned distance are implemented in the \texttt{R} package \texttt{CurveDepth} \citep{curveDepth} which is available on the CRAN \citep{RSoftware2019}. 

\section{Numerical Experiments Using Simulations}\label{sec:numexperiments}


\subsection{Simulated examples with a closed-form depth formula}\label{subsec:closedform}

The particular geometrical aspects of the curves in the following examples allows one to gain a better insight in the behavior of our curve depth and its potential limitations. More details on the computations are provided in the Supplementary Materials (Section \ref{app:ssec:examples}).

\paragraph{Segments on a line.}
We observe a sample of $n$ non-overlapping segments $[a_k,b_k]$, $k=1,\ldots,n$, on a line. 
Without loss of generality, we denote by $\Xx_k$ the $k^{\text{th}}$ segment, from left to right (see Figure \ref{fig:degenerate} top).
The curve depth of $\Xx_k$ w.r.t\ $\Xx_1,\ldots, \Xx_n$ is
$$
D({\Xx_k}|\Xx_1,\ldots, \Xx_n) = \left\{\begin{array}{ll}
1/n & \text{ if } k=1 \text{ or } k=n \\ 
1/n - ((n-1)/n)\log\left((1-t_k)^{1-t_k} t_k^{t_k}\right) & \text{otherwise}, 
\end{array}
\right.
$$
where $t_k=(k-1)/(n-1).$
The deepest curve is the segment for which $t_k$ is the closest point to $1/2.$
Our curve depth induces the same ordering as when one computes the original Tukey depth of the middle points of the segments. 
It is worthwhile noting that when the sample size $n$ increases, $D({\Xx_k}|\Xx_1,\ldots ,\Xx_n)$ tends to the Shannon entropy  (in base $b=e$) of a Bernoulli($t_k$) random variable. Thus our segment depth is maximum at $1/2$ (its value being equal to $\log(2)$) and minimal (i.e., equal to 0) close to $0$ and $1$. Outliers correspond to minimal depth and minimal entropy.

\paragraph{Parallel segments on a rectangle.}
Let $\Cc_y$ be the segment of $[0,1]^2$ defined as the set $\{(x,y);\ x\in[0,1]\}.$
We define $\Xx\sim P$ as the random curve generated from the following scheme (see Figure \ref{fig:degenerate} , bottom left) :
\begin{align*}
\Xx=\Cc_{Y}\text{ with }Y\sim \Uu[0,1].
\end{align*}
The population version of our curve depth is
$$
D(\Cc_y|P) = \min(y,1-y)=1/2 - |y - 1/2| \text{ for } y\in[0,1].
$$
Our curve depth induces the same ordering as when one computes the Tukey depth of the abscissa of the segments.
Notice that due to the particular geometry of the distribution of the segments, our curve depth is unable to detect as outliers vertical segments lying in the interior of the support of the measure $Q_P$ (here it is the unit square).

\paragraph{Star segments.}
Let $\Cc_\theta$ be the segment in $\RR^2$ from $(0,0)$ to the point $(\cos(\theta),\sin(\theta))$, for $\theta\in[0,2\pi)$.
We define $\Xx\sim P$ as the random curve generated from the following scheme (see Figure \ref{fig:degenerate} , bottom right) :
\begin{align*}
\Xx=\Cc_{\theta}\text{ with }\theta\sim \Uu[0,2\pi].
\end{align*}
By symmetry, every segment has the same depth, which is equal to $0.255$.

%

\begin{figure}[H]
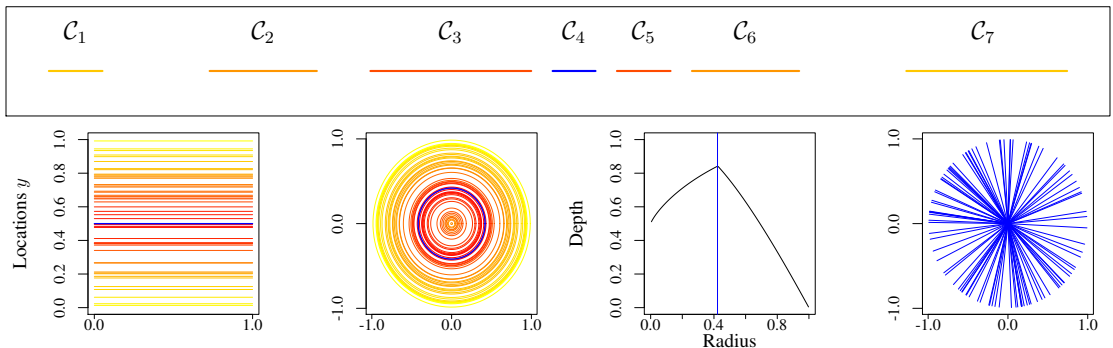

	\begin{center}
		\resizebox{0.89\linewidth}{!}{

\begin{tikzpicture}[x=5pt,y=0.5pt]
\definecolor{fillColor}{RGB}{255,255,255}
\path[use as bounding box,fill=fillColor,fill opacity=0.00] (50,100) rectangle (155,200);

\begin{scope}
\path[clip] (  50.00,  100.00) rectangle (155,211);
\definecolor{drawColor}{RGB}{0,0,0}
\path[draw=drawColor,line width= 0.4pt,line join=round,line cap=round] ( 51, 108) --
	(154, 108) --
	(154,210) --
	( 51,210) --
	( 51, 108);
\end{scope}

\begin{scope}
\path[clip] ( 50, 100) rectangle (155,200);
\definecolor{writeColor}{RGB}{0,0,0}

\definecolor{drawColor}{RGB}{255,200,0}
\path[draw=drawColor,line width= 1pt,line join=round,line cap=round] ( 55,150) -- ( 60,150);
\node[text=writeColor,anchor=base,inner sep=0pt, outer sep=0pt, scale=  1.00] at ( 57.5,180) {$\Cc_1$};

\definecolor{drawColor}{RGB}{255,150,0}
\path[draw=drawColor,line width= 1pt,line join=round,line cap=round] ( 70,150) -- ( 80,150);
\node[text=writeColor,anchor=base,inner sep=0pt, outer sep=0pt, scale=  1.00] at ( 75,180) {$\Cc_2$};

\definecolor{drawColor}{RGB}{255,75,0}
\path[draw=drawColor,line width= 1pt,line join=round,line cap=round] ( 85,150) --	( 100,150);
\node[text=writeColor,anchor=base,inner sep=0pt, outer sep=0pt, scale=  1.00] at ( 92.5,180) {$\Cc_3$};

\definecolor{drawColor}{RGB}{0,0,255}
\path[draw=drawColor,line width= 1pt,line join=round,line cap=round] (102,150) -- (106,150);
\node[text=writeColor,anchor=base,inner sep=0pt, outer sep=0pt, scale=  1.00] at (104,180) {$\Cc_4$};

\definecolor{drawColor}{RGB}{255,75,0}
\path[draw=drawColor,line width= 1pt,line join=round,line cap=round] (108,150) -- (113,150);
\node[text=writeColor,anchor=base,inner sep=0pt, outer sep=0pt, scale=  1.00] at (110.50,180) {$\Cc_5$};

\definecolor{drawColor}{RGB}{255,150,0}
\path[draw=drawColor,line width= 1pt,line join=round,line cap=round] (115,150) -- (125.00,150);
\node[text=writeColor,anchor=base,inner sep=0pt, outer sep=0pt, scale=  1.00] at (120,180) {$\Cc_6$};

\definecolor{drawColor}{RGB}{255,200,0}
\path[draw=drawColor,line width= 1pt,line join=round,line cap=round] (135,150) --	(150,150);
\node[text=writeColor,anchor=base,inner sep=0pt, outer sep=0pt, scale=  1.00] at (142.19,180) {$\Cc_7$};
\end{scope}
\end{tikzpicture}}\\
		\resizebox{0.2125\linewidth}{!}{\input{parallel.tex}}
		\resizebox{0.2125\linewidth}{!}{\input{circles.tex}}
    	\resizebox{0.2125\linewidth}{!}{

\begin{tikzpicture}[x=1pt,y=0.45pt]
\definecolor{fillColor}{RGB}{255,255,255}
\path[use as bounding box,fill=fillColor,fill opacity=0.00] (0,35) rectangle (165,320);
\begin{scope}
\path[clip] ( 49.20, 61.20) rectangle (155.47,312.15);
\definecolor{drawColor}{RGB}{0,0,0}

\path[draw=drawColor,line width= 0.4pt,line join=round,line cap=round] ( 53.63,189.26) --
	( 53.87,190.33) --
	( 54.37,192.29) --
	( 54.86,194.10) --
	( 55.35,195.81) --
	( 55.84,197.43) --
	( 56.33,198.98) --
	( 56.83,200.48) --
	( 57.32,201.92) --
	( 57.81,203.33) --
	( 58.30,204.69) --
	( 58.79,206.02) --
	( 59.29,207.32) --
	( 59.78,208.58) --
	( 60.27,209.82) --
	( 60.76,211.04) --
	( 61.25,212.23) --
	( 61.75,213.39) --
	( 62.24,214.54) --
	( 62.73,215.67) --
	( 63.22,216.77) --
	( 63.71,217.86) --
	( 64.21,218.93) --
	( 64.70,219.99) --
	( 65.19,221.03) --
	( 65.68,222.05) --
	( 66.17,223.06) --
	( 66.67,224.06) --
	( 67.16,225.04) --
	( 67.65,226.01) --
	( 68.14,226.96) --
	( 68.63,227.91) --
	( 69.13,228.84) --
	( 69.62,229.76) --
	( 70.11,230.67) --
	( 70.60,231.56) --
	( 71.09,232.45) --
	( 71.59,233.33) --
	( 72.08,234.20) --
	( 72.57,235.05) --
	( 73.06,235.90) --
	( 73.55,236.74) --
	( 74.05,237.57) --
	( 74.54,238.39) --
	( 75.03,239.20) --
	( 75.52,240.00) --
	( 76.01,240.80) --
	( 76.51,241.59) --
	( 77.00,242.37) --
	( 77.49,243.14) --
	( 77.98,243.90) --
	( 78.47,244.65) --
	( 78.97,245.40) --
	( 79.46,246.14) --
	( 79.95,246.88) --
	( 80.44,247.60) --
	( 80.93,248.32) --
	( 81.43,249.04) --
	( 81.92,249.74) --
	( 82.41,250.44) --
	( 82.90,251.13) --
	( 83.39,251.82) --
	( 83.89,252.50) --
	( 84.38,253.18) --
	( 84.87,253.84) --
	( 85.36,254.51) --
	( 85.86,255.16) --
	( 86.35,255.81) --
	( 86.84,256.45) --
	( 87.33,257.09) --
	( 87.82,257.73) --
	( 88.32,258.35) --
	( 88.81,258.97) --
	( 89.30,259.59) --
	( 89.79,260.20) --
	( 90.28,260.81) --
	( 90.78,261.41) --
	( 91.27,262.00) --
	( 91.76,262.59) --
	( 92.25,263.18) --
	( 92.74,263.76) --
	( 93.24,264.33) --
	( 93.73,264.90) --
	( 94.22,265.46) --
	( 94.71,266.02) --
	( 95.20,265.43) --
	( 95.70,264.22) --
	( 96.19,263.00) --
	( 96.68,261.76) --
	( 97.17,260.51) --
	( 97.66,259.25) --
	( 98.16,257.97) --
	( 98.65,256.68) --
	( 99.14,255.38) --
	( 99.63,254.06) --
	(100.12,252.73) --
	(100.62,251.39) --
	(101.11,250.04) --
	(101.60,248.67) --
	(102.09,247.29) --
	(102.58,245.90) --
	(103.08,244.50) --
	(103.57,243.09) --
	(104.06,241.67) --
	(104.55,240.24) --
	(105.04,238.79) --
	(105.54,237.34) --
	(106.03,235.87) --
	(106.52,234.40) --
	(107.01,232.91) --
	(107.50,231.42) --
	(108.00,229.91) --
	(108.49,228.40) --
	(108.98,226.87) --
	(109.47,225.34) --
	(109.96,223.80) --
	(110.46,222.25) --
	(110.95,220.69) --
	(111.44,219.12) --
	(111.93,217.54) --
	(112.42,215.96) --
	(112.92,214.36) --
	(113.41,212.76) --
	(113.90,211.15) --
	(114.39,209.53) --
	(114.88,207.91) --
	(115.38,206.27) --
	(115.87,204.63) --
	(116.36,202.98) --
	(116.85,201.33) --
	(117.34,199.66) --
	(117.84,197.99) --
	(118.33,196.31) --
	(118.82,194.63) --
	(119.31,192.93) --
	(119.80,191.23) --
	(120.30,189.53) --
	(120.79,187.82) --
	(121.28,186.10) --
	(121.77,184.37) --
	(122.26,182.64) --
	(122.76,180.90) --
	(123.25,179.15) --
	(123.74,177.40) --
	(124.23,175.65) --
	(124.72,173.88) --
	(125.22,172.11) --
	(125.71,170.34) --
	(126.20,168.56) --
	(126.69,166.77) --
	(127.18,164.98) --
	(127.68,163.18) --
	(128.17,161.38) --
	(128.66,159.57) --
	(129.15,157.75) --
	(129.64,155.93) --
	(130.14,154.11) --
	(130.63,152.28) --
	(131.12,150.44) --
	(131.61,148.60) --
	(132.10,146.76) --
	(132.60,144.91) --
	(133.09,143.05) --
	(133.58,141.19) --
	(134.07,139.33) --
	(134.56,137.46) --
	(135.06,135.58) --
	(135.55,133.70) --
	(136.04,131.82) --
	(136.53,129.93) --
	(137.02,128.04) --
	(137.52,126.14) --
	(138.01,124.24) --
	(138.50,122.33) --
	(138.99,120.42) --
	(139.48,118.51) --
	(139.98,116.59) --
	(140.47,114.67) --
	(140.96,112.74) --
	(141.45,110.81) --
	(141.94,108.88) --
	(142.44,106.94) --
	(142.93,105.00) --
	(143.42,103.05) --
	(143.91,101.10) --
	(144.40, 99.15) --
	(144.90, 97.19) --
	(145.39, 95.23) --
	(145.88, 93.27) --
	(146.37, 91.30) --
	(146.86, 89.33) --
	(147.36, 87.35) --
	(147.85, 85.38) --
	(148.34, 83.40) --
	(148.83, 81.41) --
	(149.32, 79.43) --
	(149.82, 77.44) --
	(150.31, 75.45) --
	(150.80, 73.46) --
	(151.29, 71.47);

\end{scope}
\begin{scope}
\path[clip] (  0.00,  0.00) rectangle (180.67,361.35);
\definecolor{drawColor}{RGB}{0,0,0}

\path[draw=drawColor,line width= 0.4pt,line join=round,line cap=round] ( 53.14, 61.20) -- (151.54, 61.20);

\path[draw=drawColor,line width= 0.4pt,line join=round,line cap=round] ( 53.14, 61.20) -- ( 53.14, 55.20);

\path[draw=drawColor,line width= 0.4pt,line join=round,line cap=round] ( 72.82, 61.20) -- ( 72.82, 55.20);

\path[draw=drawColor,line width= 0.4pt,line join=round,line cap=round] ( 92.50, 61.20) -- ( 92.50, 55.20);

\path[draw=drawColor,line width= 0.4pt,line join=round,line cap=round] (112.18, 61.20) -- (112.18, 55.20);

\path[draw=drawColor,line width= 0.4pt,line join=round,line cap=round] (131.86, 61.20) -- (131.86, 55.20);

\path[draw=drawColor,line width= 0.4pt,line join=round,line cap=round] (151.54, 61.20) -- (151.54, 55.20);

\node[text=drawColor,anchor=base,inner sep=0pt, outer sep=0pt, scale=  1.00] at ( 53.14, 39.60) {\footnotesize 0.0};

\node[text=drawColor,anchor=base,inner sep=0pt, outer sep=0pt, scale=  1.00] at ( 92.50, 39.60) {\footnotesize 0.4};

\node[text=drawColor,anchor=base,inner sep=0pt, outer sep=0pt, scale=  1.00] at (131.86, 39.60) {\footnotesize 0.8};

\path[draw=drawColor,line width= 0.4pt,line join=round,line cap=round] ( 49.20, 70.49) -- ( 49.20,302.86);

\path[draw=drawColor,line width= 0.4pt,line join=round,line cap=round] ( 49.20, 70.49) -- ( 43.20, 70.49);

\path[draw=drawColor,line width= 0.4pt,line join=round,line cap=round] ( 49.20,116.97) -- ( 43.20,116.97);

\path[draw=drawColor,line width= 0.4pt,line join=round,line cap=round] ( 49.20,163.44) -- ( 43.20,163.44);

\path[draw=drawColor,line width= 0.4pt,line join=round,line cap=round] ( 49.20,209.91) -- ( 43.20,209.91);

\path[draw=drawColor,line width= 0.4pt,line join=round,line cap=round] ( 49.20,256.38) -- ( 43.20,256.38);

\path[draw=drawColor,line width= 0.4pt,line join=round,line cap=round] ( 49.20,302.86) -- ( 43.20,302.86);

\node[text=drawColor,rotate= 90.00,anchor=base,inner sep=0pt, outer sep=0pt, scale=  1.00] at ( 34.80, 70.49) {\footnotesize 0.0};

\node[text=drawColor,rotate= 90.00,anchor=base,inner sep=0pt, outer sep=0pt, scale=  1.00] at ( 34.80,116.97) {\footnotesize 0.2};

\node[text=drawColor,rotate= 90.00,anchor=base,inner sep=0pt, outer sep=0pt, scale=  1.00] at ( 34.80,163.44) {\footnotesize 0.4};

\node[text=drawColor,rotate= 90.00,anchor=base,inner sep=0pt, outer sep=0pt, scale=  1.00] at ( 34.80,209.91) {\footnotesize 0.6};

\node[text=drawColor,rotate= 90.00,anchor=base,inner sep=0pt, outer sep=0pt, scale=  1.00] at ( 34.80,256.38) {\footnotesize 0.8};

\node[text=drawColor,rotate= 90.00,anchor=base,inner sep=0pt, outer sep=0pt, scale=  1.00] at ( 34.80,302.86) {\footnotesize 1.0};

\path[draw=drawColor,line width= 0.4pt,line join=round,line cap=round] ( 49.20, 61.20) --
	(155.47, 61.20) --
	(155.47,312.15) --
	( 49.20,312.15) --
	( 49.20, 61.20);
\end{scope}
\begin{scope}
\path[clip] (  0.00,  0.00) rectangle (180.67,361.35);
\definecolor{drawColor}{RGB}{0,0,0}

\node[text=drawColor,anchor=base,inner sep=0pt, outer sep=0pt, scale=  1.00] at (102.34, 15.60) {Radius};

\node[text=drawColor,rotate= 90.00,anchor=base,inner sep=0pt, outer sep=0pt, scale=  1.00] at ( 10.80,186.67) {Depth};
\end{scope}
\begin{scope}
\path[clip] ( 49.20, 61.20) rectangle (155.47,312.15);
\definecolor{drawColor}{RGB}{0,0,255}
\path[draw=drawColor,line width= 0.4pt,line join=round,line cap=round] ( 94.71, 61.20) -- ( 94.71,312.15);
\end{scope}
\end{tikzpicture}}
    	\resizebox{0.2125\linewidth}{!}{\input{rayon.tex}}
	\end{center}
	\caption{
			Illustration of examples : (top) Sample of $n=7$ non-overlapping segments on a line, 
          (bottom left) Sample of $n=50$ parallel segments, (bottom middle) Sample of $n=100$ concentric circles and the associated population version of our \curvedepth{} as a function of the radius, (bottom right) Sample of $n=100$ star segments. For each scenario, the deepest curves are plotted in blue while darker red indicates a higher value of depth.}
	\label{fig:degenerate}
\end{figure}

\paragraph{Concentric circles.}
Let $\Cc_r$ be the circle in $\RR^2$ of center $0$ and radius $r>0$.
We define $\Xx\sim P$ as the random curve generated from the following scheme (see Figure \ref{fig:degenerate}, bottom middle-left) :
\begin{align*}
\Xx=\Cc_{R}\text{ with }R\sim \Uu[0,1].
\end{align*}
The population version of our curve depth is plotted on Figure \ref{fig:degenerate} (bottom middle-right). 
The deepest circle is the circle with a radius $r=0.425.$
It is worthwhile noting that our approach do not incorrectly lead to the deepest curve being the circle with a null radius. This being said, one may have expected the deepest curve to be the circle with radius $r=1/2.$

\subsection{Monte Carlo Approximation of the \CurveDepth}\label{ssec:simulation}

In most cases, it is not possible to get an explicit expression of our \curvedepth{} since it requires to compute for all $x\in S_\Cc$ an infinimum of the ratio $Q_n(H_{x,u})/\mu_\Cc(H_{x,u})$ over all $u\in\Ss.$ 
Section \ref{sec:Implementation} describes a Monte Carlo estimate of $D(\Cc|\Xx_1,\ldots,\Xx_n)$; see Algorithm \ref{alg:montecarlo}.
This approximation is consistent according to Theorem \ref{thm:consistency0}.
We conducted a Monte Carlo study to assess this convergence in several scenarios in Section~\ref{app:ssec:simulation} of the Supplementary Materials. Here, we only give a brief summary of these results.

\paragraph{Scheme 1 : Concentric circles.}
We consider the population of concentric circles with radius lying in the interval $(0,1)$ described in Subsection~\ref{subsec:closedform}.
For a given sample of circles $\{\Xx_1,\ldots, \Xx_n\}$, we have an explicit expression both for $D(\Cc_r|\Xx_1,\ldots,\Xx_n)$ and $D(\Cc_r|P)$, where $\Cc_r$ is the circle of radius $r\in(0,1)$. 
This example has the particularity that the functions $x\in S_\Cc \mapsto D(x|Q_n,\mu_\Cc)$ and $x\in S_\Cc \mapsto D(x|Q_P,\mu_\Cc)$ are constant over their domain. Our main findings are the following.

\textit{1. The Monte Carlo estimator of the sample \curvedepth{} converges in probability as the Monte Carlo sample size $m$ goes to infinity}; see Figure~\ref{app:fig:circles:m:pack} in the Supplementary Materials.
Monte Carlo estimates (see Algorithm \ref{alg:montecarlo}) tend in average to underestimate the sample \curvedepth{}.
Observing such a negative bias is not surprising since we aim to compute an infimum over all directions $u\in\Ss$. 
However this bias and the standard deviation depend on the value of the radius (i.e., on the position of the curve $\Cc$ w.r.t.\ the sample of curves) and they both decrease towards zero as $m$ gets large.

\textit{2. The sample \curvedepth{} converges in probability to the population \curvedepth.}
The bias and the standard deviation of the sample \curvedepth\ computed over $5,000$ replications decrease towards zero as $n$ goes to $+\infty$; see Table~\ref{app:tab:circles:n} in the Supplementary Materials.
Moreover, the standard deviation of $D(\Cc_r|\Xx_1,\ldots,\Xx_n)$ seems to be dependent on the value of the radius $r$. As expected, the Monte Carlo estimator of the population \curvedepth\ (see Algorithm~\ref{alg:montecarlo}) also converges in probability for increasing values of both $n$ and $m$. The rate of convergence of the latter is slightly smaller, with on average a greater impact on its bias than on its standard deviation.

\paragraph{Scheme 2 and scheme 3 : functional data.}
We consider two example of simulated functional data from  \citet[see paragraph 4.2.1 in][]{claeskens2014multivariate} and from \cite{cuevas2007models}.
Here we consider as unparametrized curve the collection of points,
$$
\Xx = \{ (t,\bmx(t)) \,:\, t\in[0,1]\}
$$
where $\bmx(t)$ is a continuous function from $[0,1]$ to $\RR.$
The sample processes of these example admit a symmetry around their respective mean function.
Moreover these mean functions are known (see the black curves in Figure \ref{fig:sim}). 
Notice that for these examples, we have no explicit formula for the sample \curvedepth{} and the population \curvedepth.

\textit{1. The convergence of the Monte Carlo estimate of the sample \curvedepth.}
For a given sample of curves, $\{\Xx_1,\ldots, \Xx_n\},$ we observe that the Monte Carlo estimate of $D(\Cc|\Xx_1,\ldots,\Xx_n)$ converges in probability to a constant with $m$ goes to $\infty$ and with $\min(n,m)$ goes to $\infty.$
Moreover, we don't observe an impact of the threshold $\Delta$ in the computation of the depth.

\textit{2. The most central curves are located in a neighborhood of the mean curves.}
According to the previous simulations, we estimate the Monte Carlo error of the deepest curve (shown in red in Figure~\ref{fig:sim}) and we select the curves whose depth belongs to the $97.5\%$-confidence (shown in orange in Figure~\ref{fig:sim}).
These curves appear to be located reasonably close to the center of the stochastic process (the black mean curves).

\begin{figure}[H]
	\begin{center}
		\includegraphics[keepaspectratio=true,width=0.425\linewidth, trim  =10mm 15mm 10mm 10mm, clip]{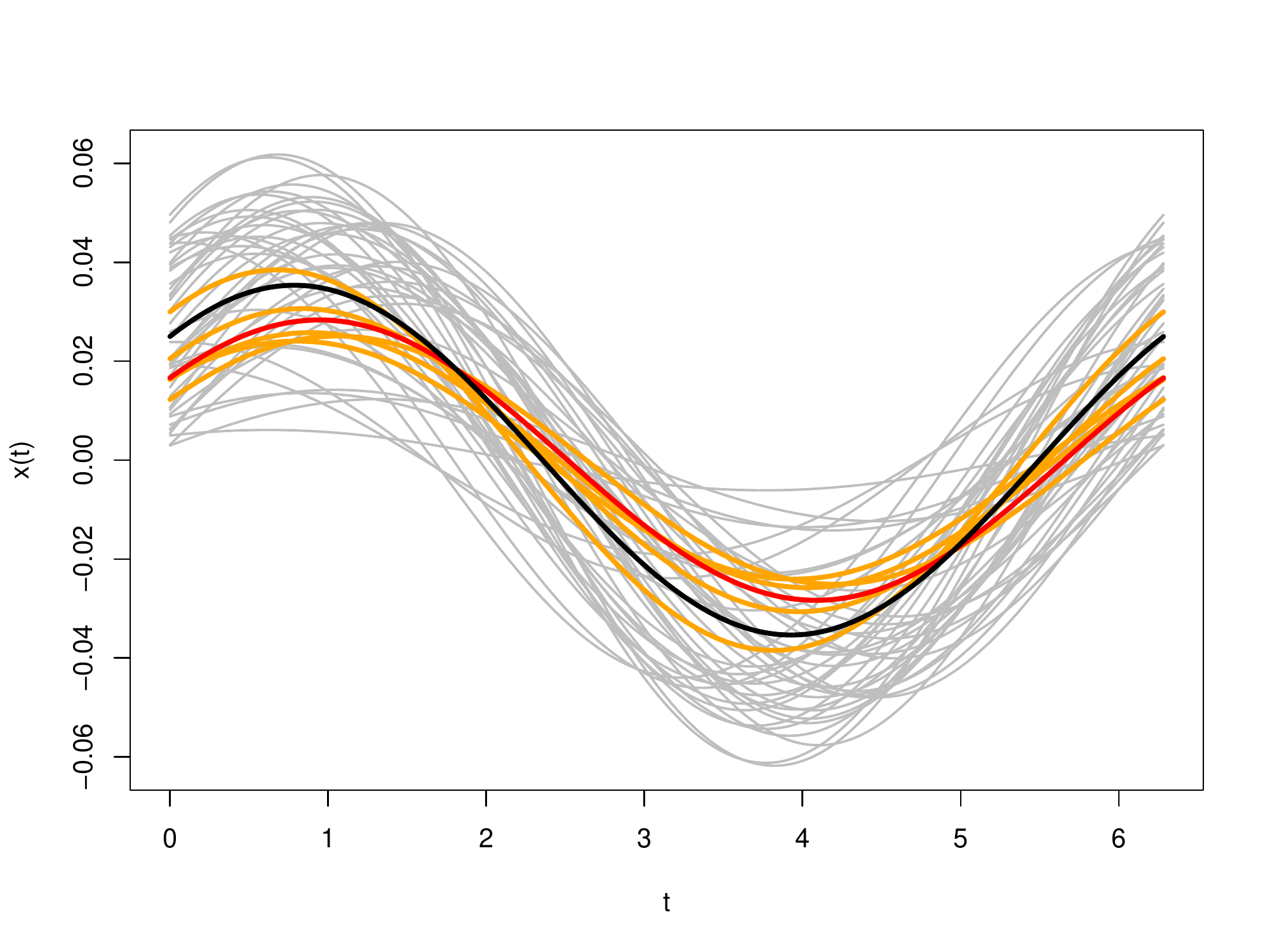} \quad
   		\includegraphics[keepaspectratio=true,width=0.425\linewidth, trim  =10mm 15mm 10mm 10mm, clip]{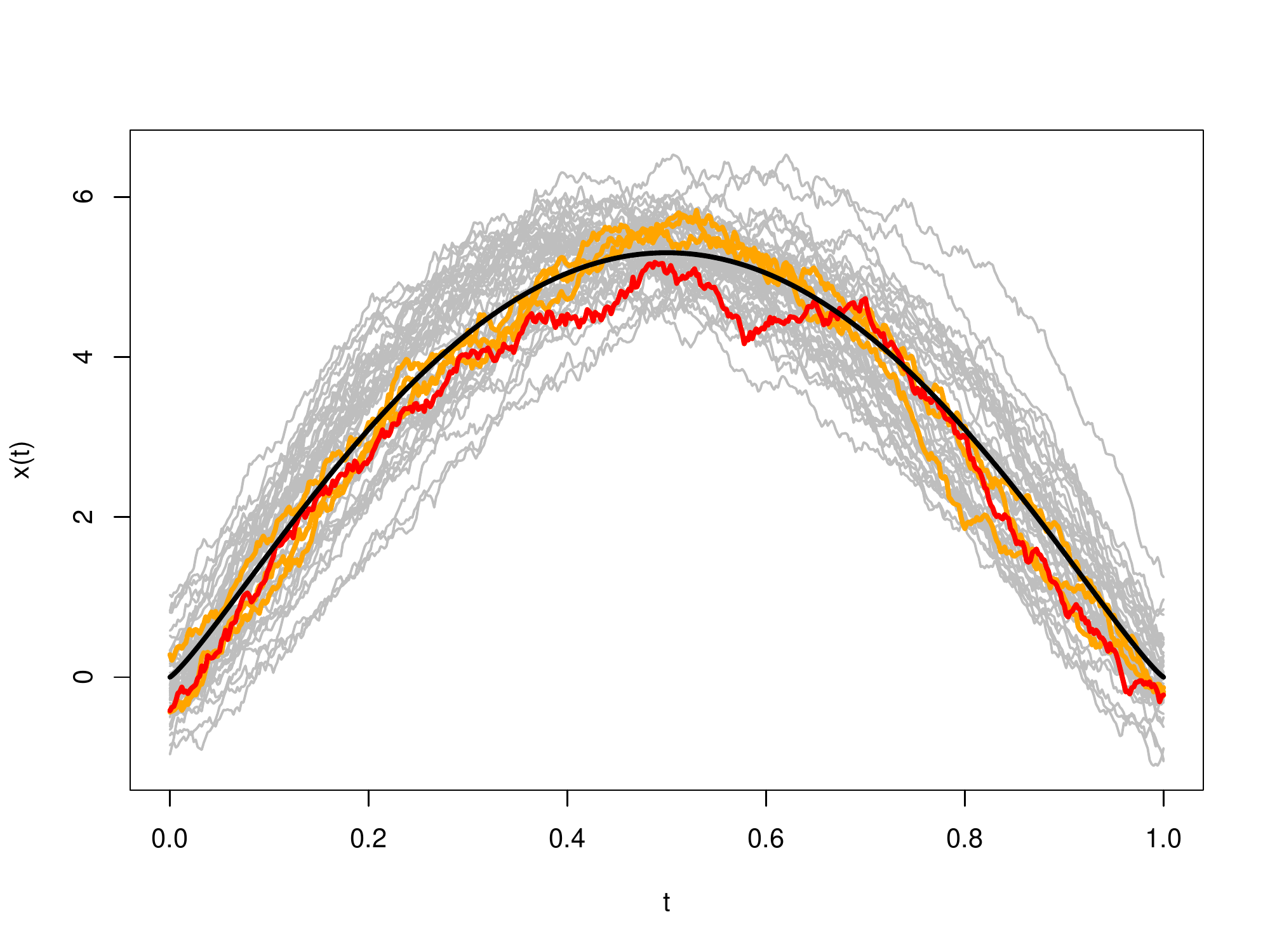}
		
	\end{center}
	\caption{Two samples of $n=50$ gray curves generated according to two simulation schemes:  (left) the one proposed by \cite{claeskens2014multivariate}, (right) the one proposed by~\cite{cuevas2007models}. The deepest curve, computed using our new method (taking $\Delta_m=1/(10m^\alpha)$ with $m=500$ and $\alpha=1/8$), is plotted in red; its depth is $0.744$ (left) and $0.752$ (right). The mean curves, plotted in black, have a depth of $0.727$ (left) and $0.571$ (right). The curves having a depth lying in the $0.975$ Monte Carlo confidence interval are plotted in orange.
}
	\label{fig:sim}
\end{figure}

\subsection{Outlier detection}

We explored the ability of our \curvedepth{} to detect outlying observations in a sample of curves, on two visual examples; see Figure~\ref{fig:outl}.

\begin{figure}[H]
	\begin{center}
		\includegraphics[keepaspectratio=true,width=0.425\linewidth, trim  =10mm 15mm 10mm 10mm, clip]{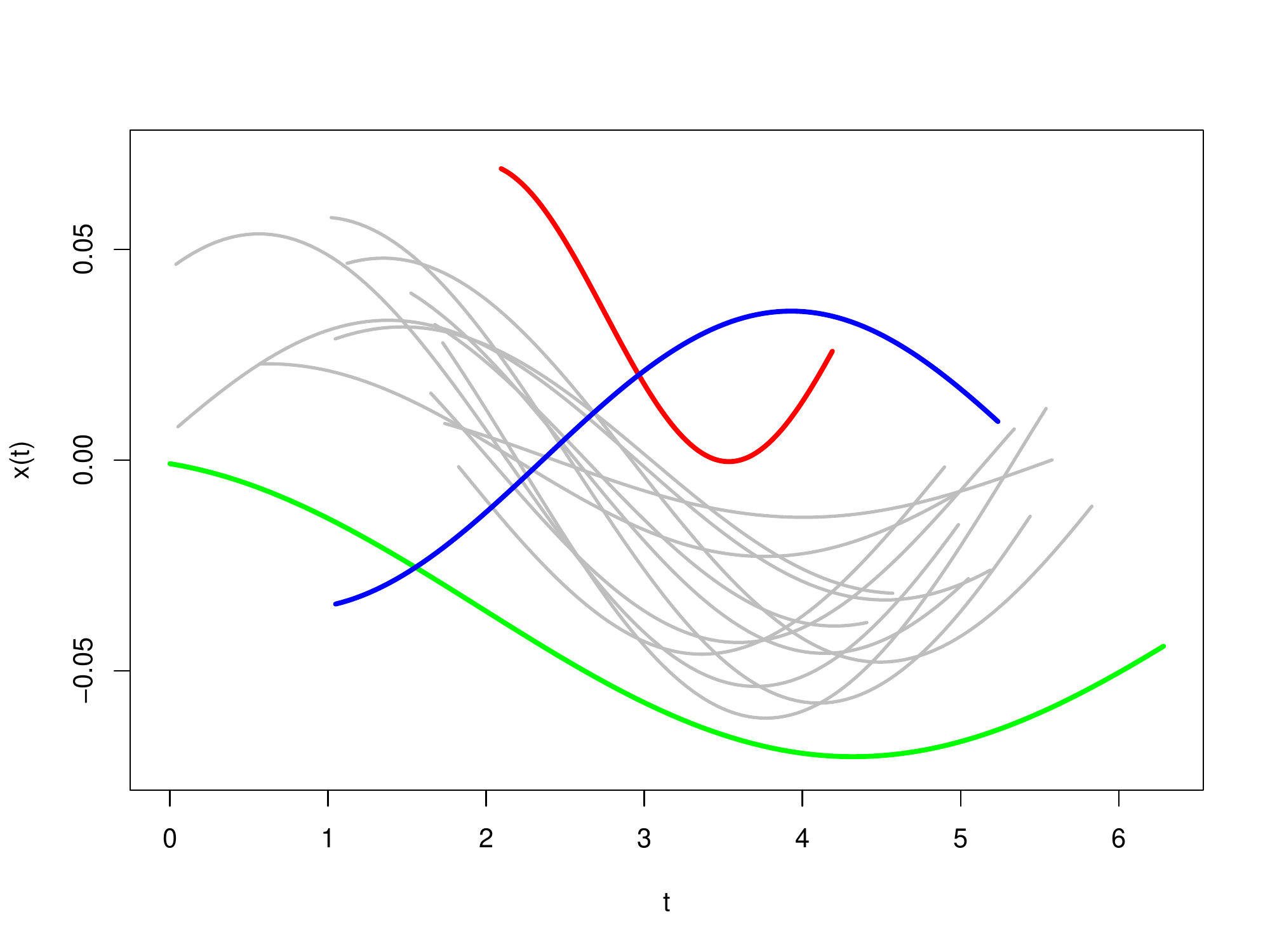} \quad
   		\includegraphics[keepaspectratio=true,width=0.425\linewidth, trim  =10mm 15mm 10mm 10mm, clip]{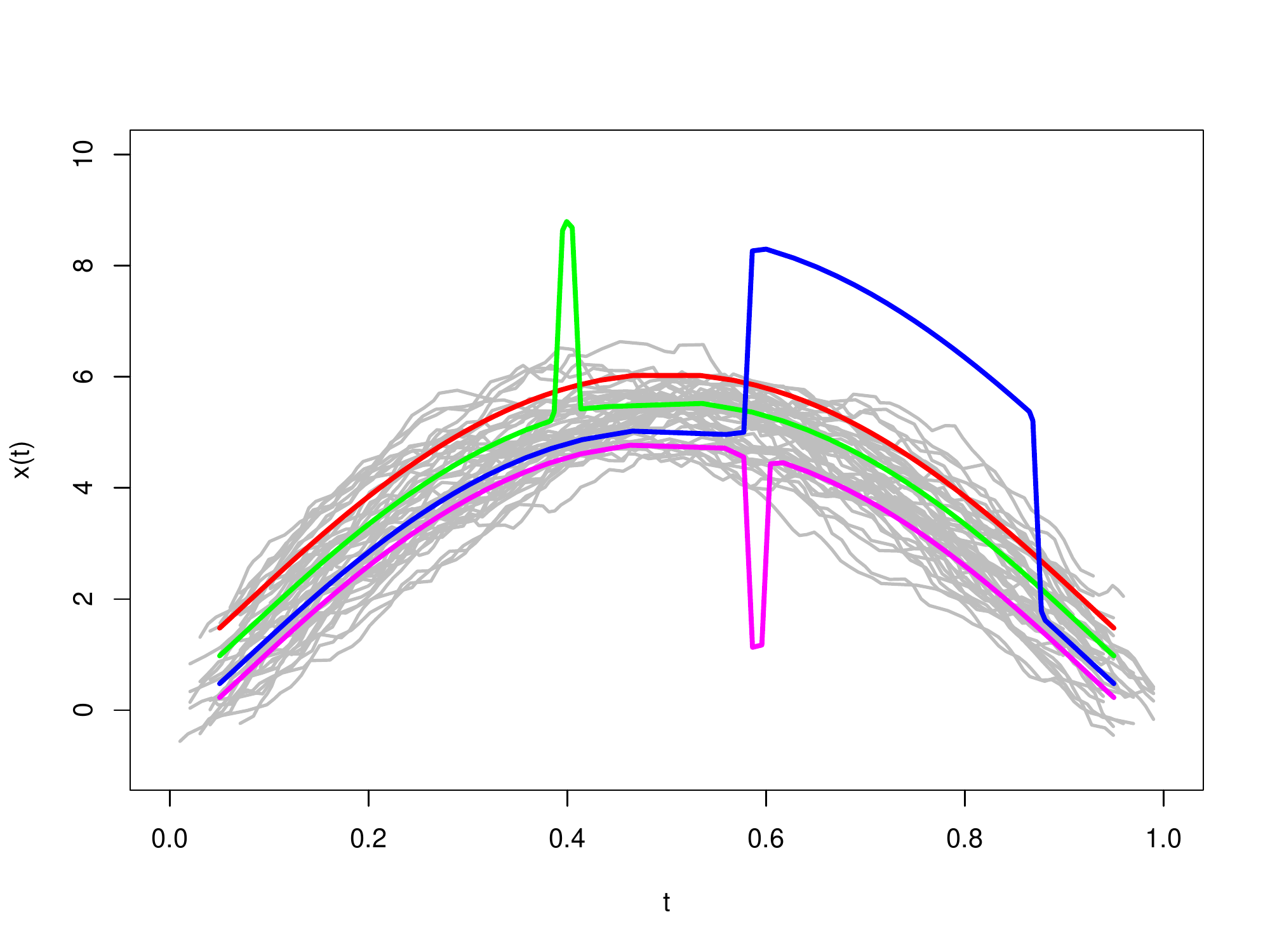}
		
	\end{center}
	\caption{Curve data with outliers. Left: $15$ smooth curves with two shift-shape (red and green) and one purely shape (blue) outliers. Right: $50$ oscillating curves with one smooth shift outlier (red), one isolated outlier (green), one persistent outlier (blue), and one isolated outlier with a negative peak (magenta).
}
	\label{fig:outl}
\end{figure}

For the first scenario, we generated a sample of $12$ 2D-curves according to the following random generating process (inspired from \citep[Section~4.2.1]{claeskens2014multivariate}):
$$
\Xx = \left\{ \left(x, A_1\sin(2\pi x) + A_2\cos(2\pi x)\right);\, x\in[L,U]\right\}\,,
$$
where $A_1,A_2\sim\Uu[0,0.05]$, $L\sim\Uu[0,\frac{2\pi}{3}]$, and $U\sim\Uu[\frac{4\pi}{3},2\pi]$, all independent. We then added three outlier curves: two (red and green) are shift-shape outliers, while the third one (blue) is a purely shape outlier \citep[for a taxonomy of multivariate functional outliers see, e.g.,][]{hubert2015}.

For the second scenario, we generated a sample of $46$ 2D-curves according to the following random generating process:
$$
\Yy = \left\{ \left(x, 30(1 - x)^{1 + W}x^{1.5 - W} + U_x\right);\,  x\in [L,U]\right\}
$$
where $\{U_t;\, t\in[0,1]\}$ is a zero mean stationary Gaussian process with covariance function $t \mapsto  0.2 e^{-\frac{1}{0.3}|t|}$, $W\sim\Uu[0,0.5]$, $L\sim\Uu[0,0.1]$, $U\sim\Uu[0.9,1]$, all independent. We then added four outliers: a shift outlier (red), an isolated outlier (green), a persistent outlier (blue), and to be fair to the other depth measures, another isolated outlier with a negative peak (magenta). All the outliers (slightly) differ in shape. 

Plots of the ordered depths of the curves in these two samples, computed using mSBD, saPRJ and MFHD (using an arc-length parametrization) as well as our \curvedepth{} are displayed on Figure~\ref{fig:outlscores}.

Our \curvedepth{} is the only one able to correctly identify the three outliers added to the $\Xx$ curves. saPRJ and MFHD only identify the two shift-shape outliers while mSBD identifies just one. None of these three other depth methods sees the pure shape outlier.

\begin{figure}[H]
\begin{center}
	\begin{tabular}{cccc}
	{\small \quad mSBD} & {\small saPRJ} & {\small MFHD} & {\small {\CurveDepth}} \\
	\includegraphics[height=2.45cm,trim = 0cm 1.5cm 1cm 1.725cm,clip=true,page=2]{pic-depthOutlA.pdf}&\includegraphics[height=2.45cm,trim = 1cm 1.5cm 1cm 1.725cm,clip=true,page=3]{pic-depthOutlA.pdf}&\includegraphics[height=2.45cm,trim = 1cm 1.5cm 1cm 1.725cm,clip=true,page=4]{pic-depthOutlA.pdf}&\includegraphics[height=2.45cm,trim = 1cm 1.5cm 1cm 1.725cm,clip=true,page=1]{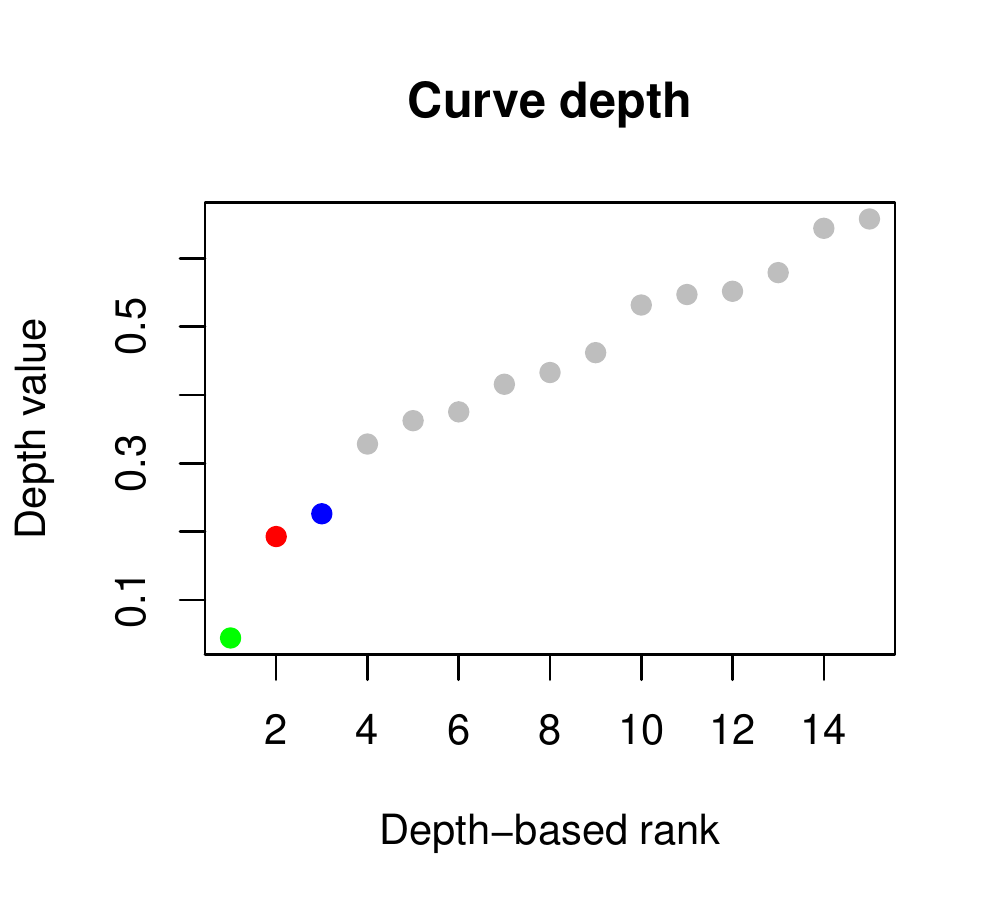}\\
	\includegraphics[height=2.83cm,trim = 0cm 0.5cm 1cm 1.725cm,clip=true,page=2]{pic-depthOutlB.pdf}&\includegraphics[height=2.83cm,trim = 1cm 0.5cm 1cm 1.725cm,clip=true,page=3]{pic-depthOutlB.pdf}&\includegraphics[height=2.83cm,trim = 1cm 0.5cm 1cm 1.725cm,clip=true,page=4]{pic-depthOutlB.pdf}&\includegraphics[height=2.83cm,trim = 1cm 0.5cm 1cm 1.725cm,clip=true,page=1]{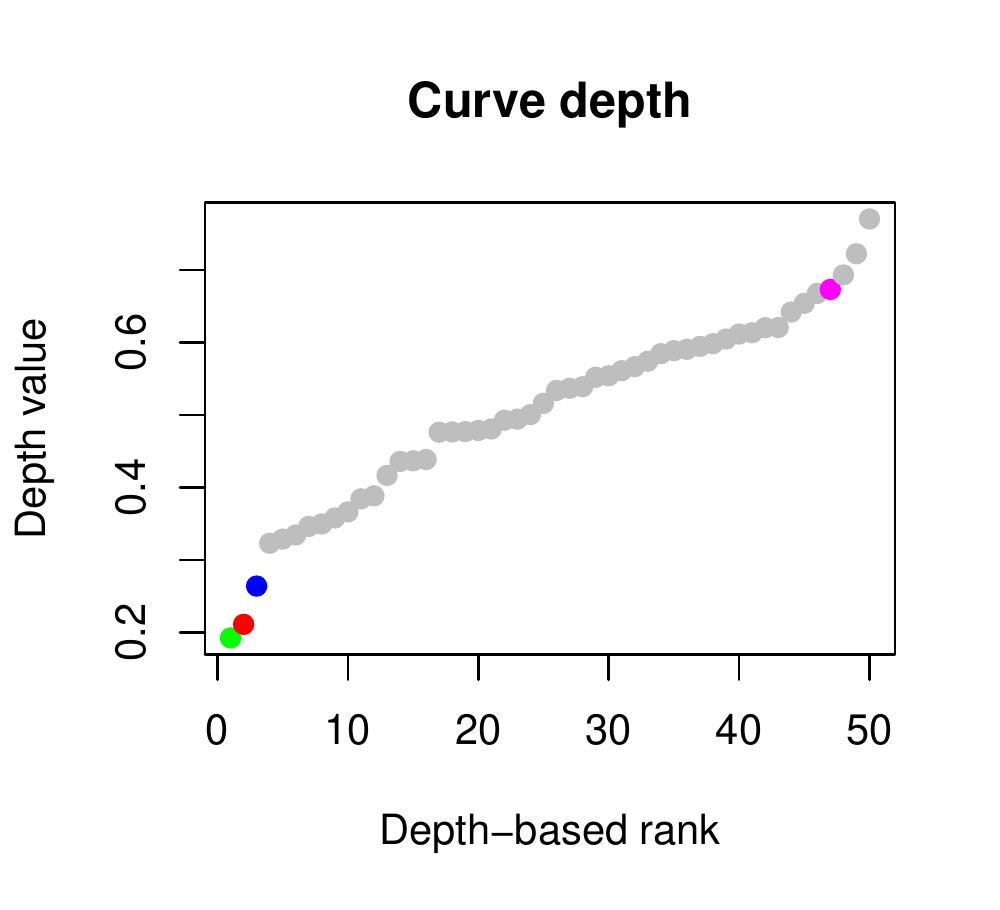}\\
	\end{tabular}
\end{center}
	\caption{Values of depth in ascending order for two samples drawn from $\Xx$ (top) and $\Yy$ (bottom) contaminated with a few outlying observations (colored as in Figure~\ref{fig:outl}).}
	\label{fig:outlscores}
\end{figure}

For the second scenario, mSBD and MFHD are not able to identify any outlier added to the $\Yy$ curves, while saPRJ only fails to find the shift outlier. Our \curvedepth\ perfectly distinguishes all outliers but the negative isolated outlier which is assigned a rather large value of depth. This failure was to be expected since a building block of our approach is to use halfspaces, as illustrated in Figure~\ref{fig:illustration}. Somehow, a similar behaviour was observed when we slightly underestimated the depth of the deepest curve in the concentric circles example of Section~\ref{subsec:closedform}.

\section{Application to Real Data}\label{sec:realdata}

\subsection{Application to the Older Australian Twins Study Data}\label{ssec:realdata-brain}


White matter (WM) in the brain is made up of long myelinated axonal fibers generally regarded as passive routes connecting several gray matter regions (the ones containing neurons) to permit flow of information across them.
In such tissue, water tends to diffuse mostly along the direction of the fibers. The ratio of axial and radial movement is called fractional anisotropy. Diffusion Tensor Magnetic Resonance Imaging (DTI) measures the motion of hydrogen atoms within water in all three dimensions.

We had access to DTI scans from the Older Australian Twins Study (OATS), an ongoing longitudinal study investigating genetic and environmental factors and their associations and interactions in healthy brain ageing and ageing-related neurocognitive disorders for people aged 65+ years \citep{Sachdev2009}.  
The DTI data considered in the current article were drawn from 34 twin pairs, aged between 67.3 and 84.2 years. Eleven of the 34 pairs were dizygotic (DZ) twin pairs (i.e., non-identical twins sharing 50\% of their genes) and 23 monozygotic (MZ) twin pairs (i.e., identical twins sharing 100\% of their genes). 
Using MRtrix software \citep{tournier2012mrtrix} to extract corticospinal fiber tracts from the DTI scans (an operation called tractography), the resulting data sets were two bundles of around $1,000$ fibers each per subject (see Figure~\ref{fig:boxplot:brain}; left). Other pre-processing steps are described in the Supplementary Materials, Section~\ref{sec:preprocess.DTI}.

It is quite a challenging task to visualize brain fibers. Consequently this information is difficult to use in a clinical environment (e.g., for surgery planning). New tools are thus needed for efficiently representing these {\it tractograms}. An interesting approach by \cite{Mercier2018} consists in progressively simplifying tractograms by grouping similar fibers into a specific geometric representation.

We believe that the depth for curves developed here can also help neuroscientists to visualize a 3D bundle of fibers. One can follow the approach adopted by \cite{Mercier2018} by grouping curves according to their depths. 
It is also possible to assign a transparency value to each curve equal or proportional to its depth value (see Figure~\ref{fig:boxplot:brain}; left) to inspect the whole bundle at once. 
Similarly, one can instead assign a low transparency value to the least deep curves in order to visualize outliers (see Figure~\ref{fig:boxplot:brain}; right). Outliers can eventually be removed before further statistical analyses are conducted.


\begin{figure}[H]
	\begin{center}
	\includegraphics[width=0.375\textwidth]{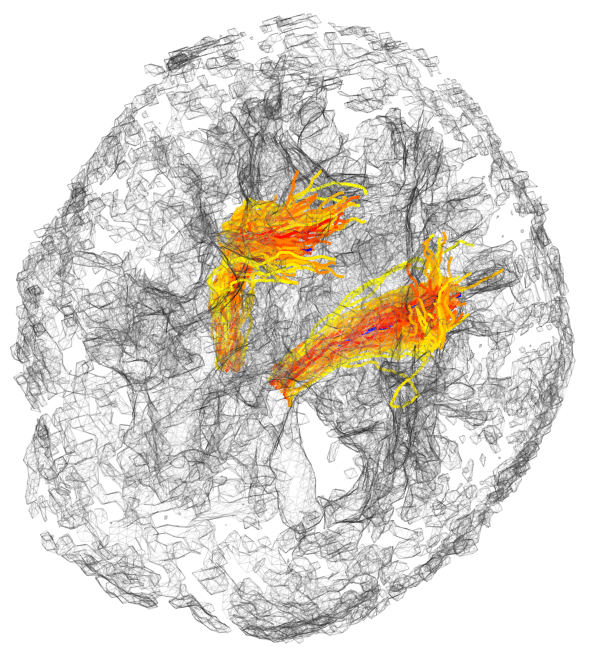}
	\quad
	\includegraphics[width=0.475\textwidth]{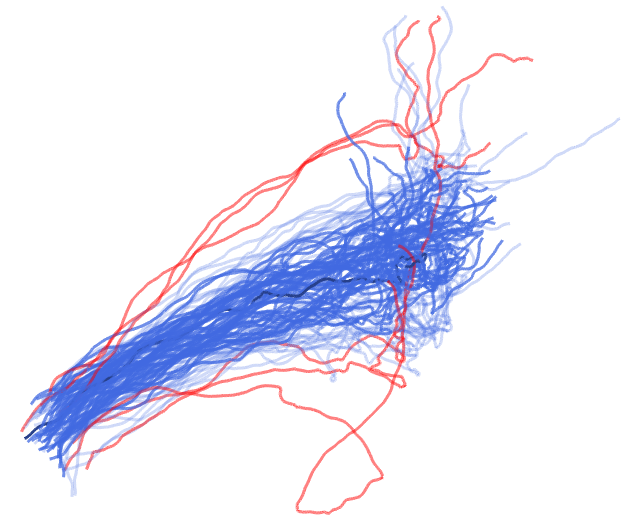}
	\end{center}
	\caption{Illustrations of the ordering of the white matter fibers for one subject using our \curvedepth.  
	(Left) Whole brain fiber data set for one twin; see \url{http://biostatisticien.eu/DataDepthFig8} for an interactive 3D applet. 
	(Right) Result of bundle ordering for the right side of the brain only. 
	We only display the first $100$ fibers in the data set, among which 6 are identified as outliers and colored in red (their depth is less than 0.075).
	}
	\label{fig:boxplot:brain}
\end{figure}


We demonstrate on Figure~\ref{fig:boxplot:brain:comparison} that our \curvedepth\ approach gives better results in terms of outlier detection than four other existing depth measures that can be applied to three-dimensional curves. These multivariate functional depth-based competitors (with an arc-length parametrization) are the \emph{modified multivariate band depth} (mMBD), SBD, mSBD, saPRJ, and our \curvedepth. We observe that the 15 fibers having the lowest depth as computed by our \curvedepth{} are located outside of the bundle, while there are fibers with a low depth value inside this bundle for the competitors. Furthermore the range of depth values associated to our \curvedepth{} is the widest among the 5 methods considered here. And there is a clearer separation between the depths of outliers and the other fibers. Notice that it is hard to distinguish outliers using SBD and mSBD and that the bottom fiber which is clearly outside the bundle is not detected as an outlier by SBD. Finally, mMBD and saPRJ detect fewer outliers than our \curvedepth, some of them being the same as the ones detected by our approach.

\begin{figure}[H]
\begin{center}
	\begin{tabular}{ccccc}
	{\small mMBD} & {\small SBD} & {\small mSBD} & {\small saPRJ} & {{\CurveDepth}} \\
		\includegraphics[width=0.155\textwidth]{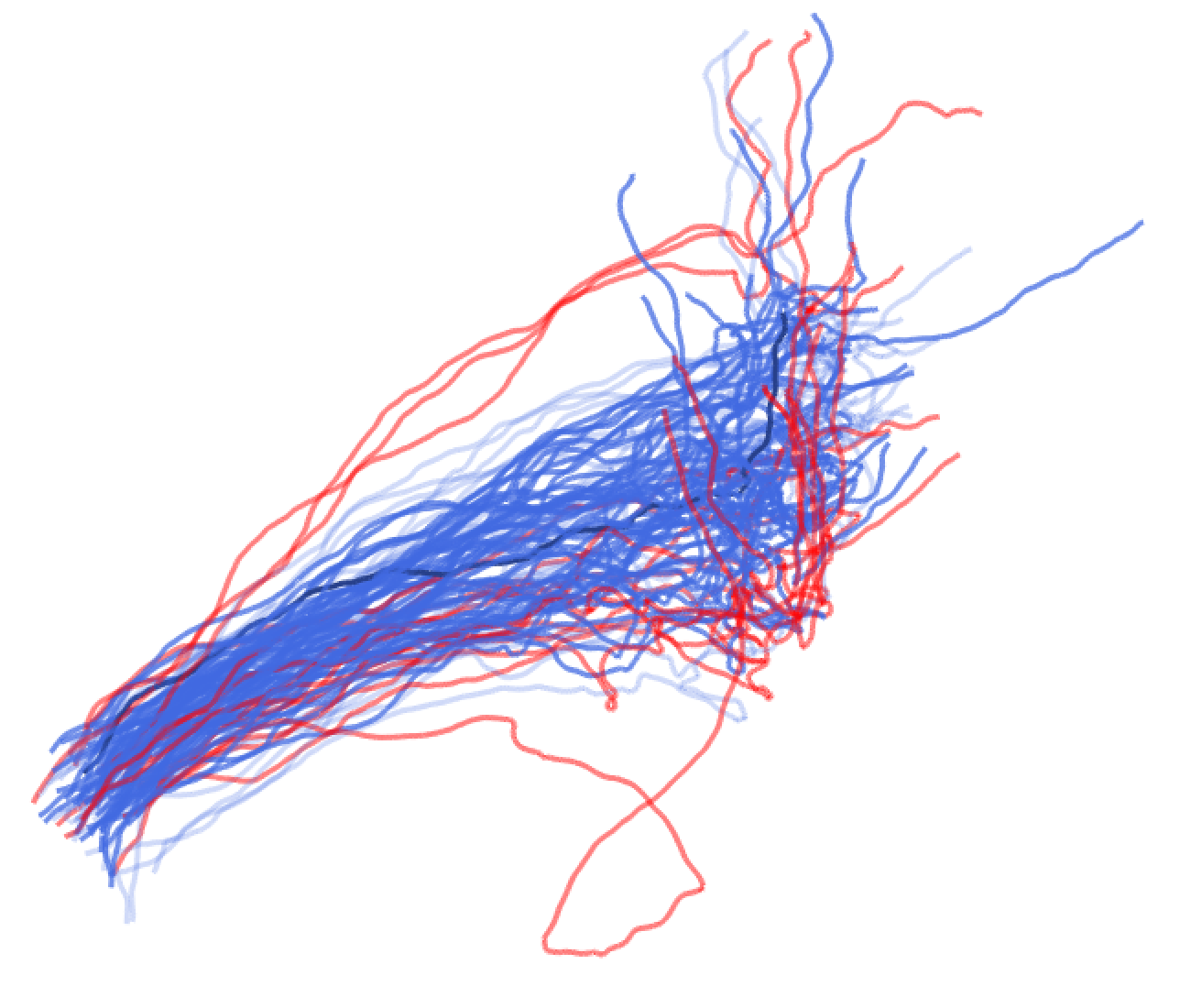}&\includegraphics[
		width=0.155\textwidth]{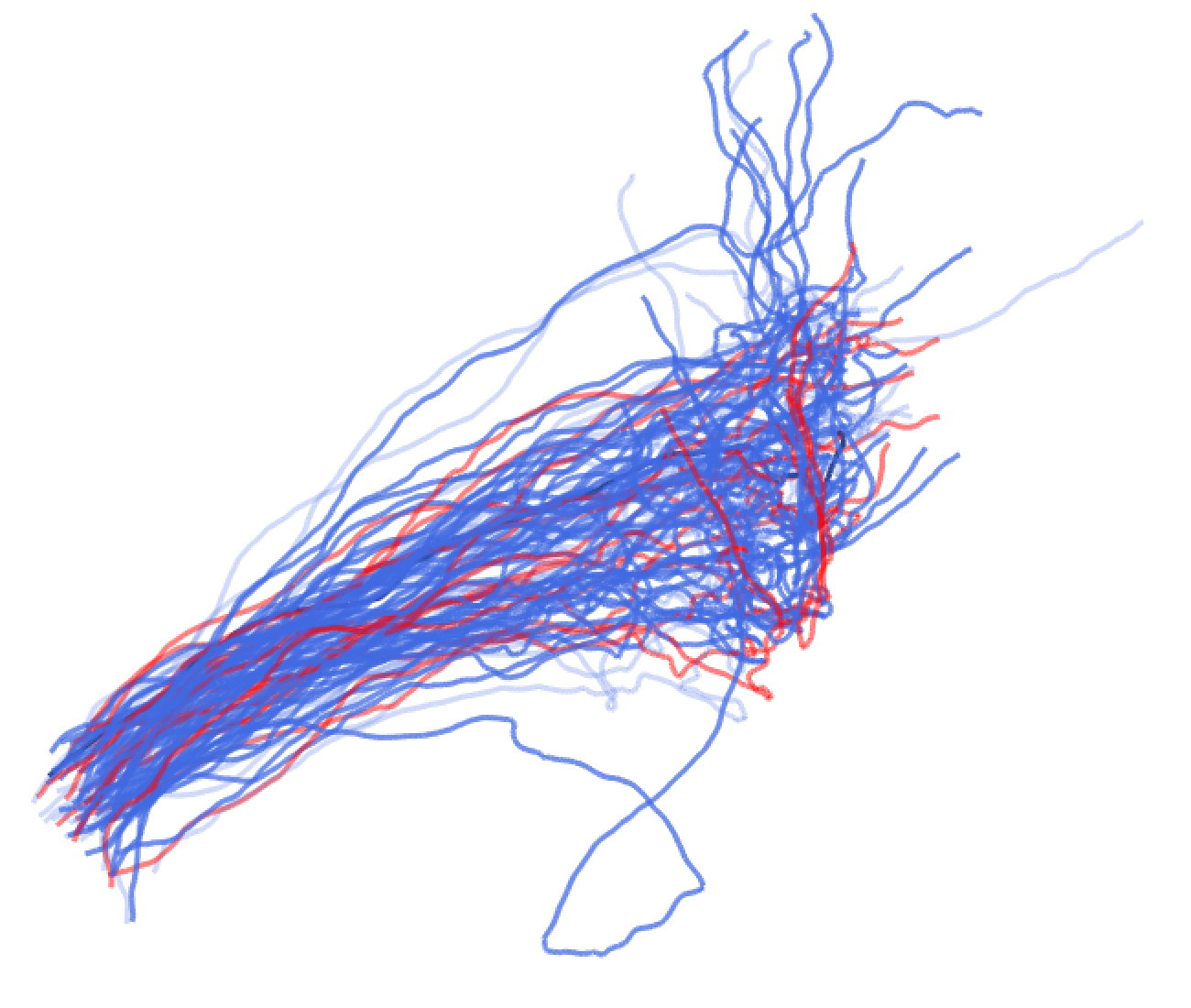}&\includegraphics[width=0.155\textwidth,page=3]{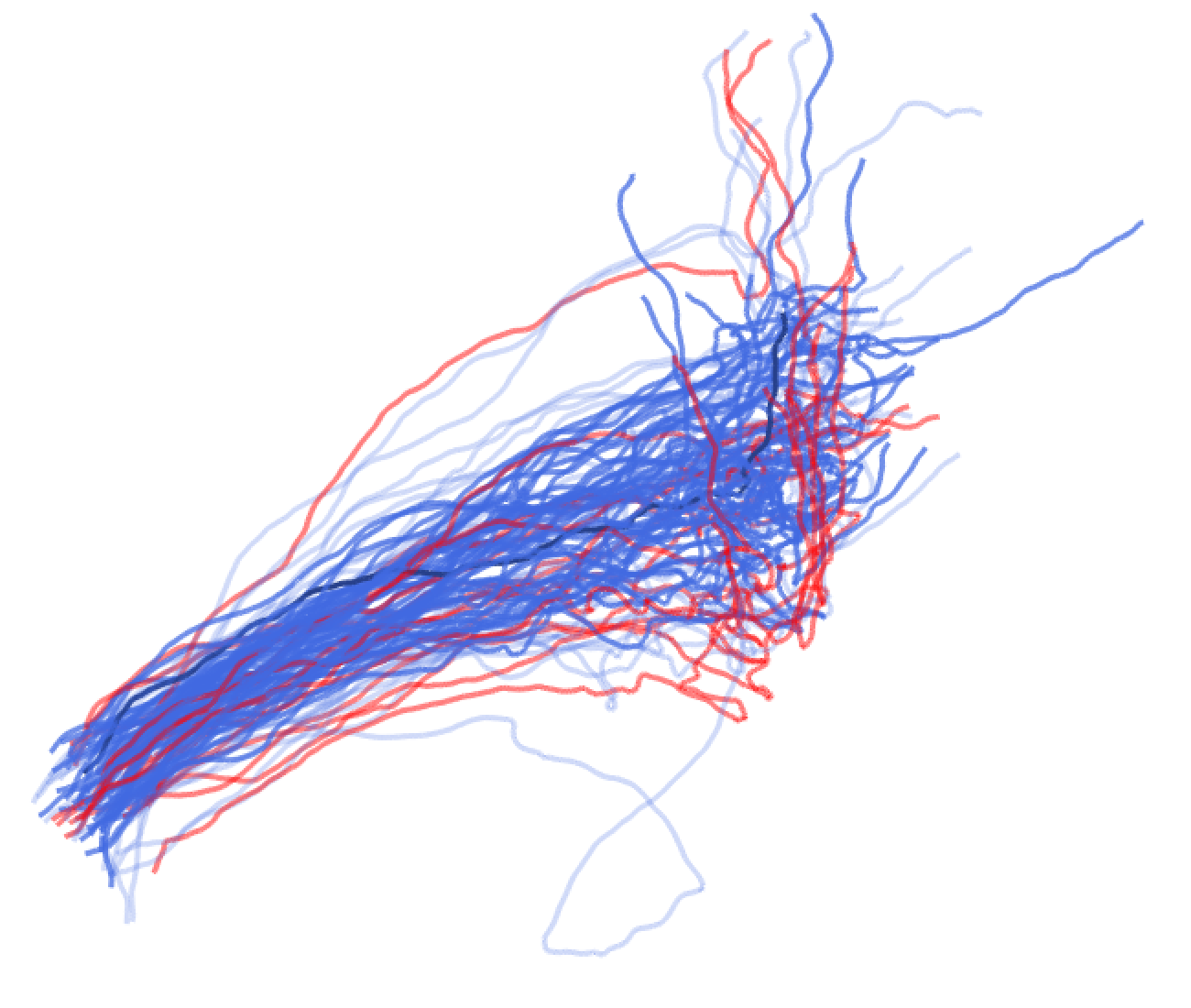}&\includegraphics[width=0.155\textwidth]{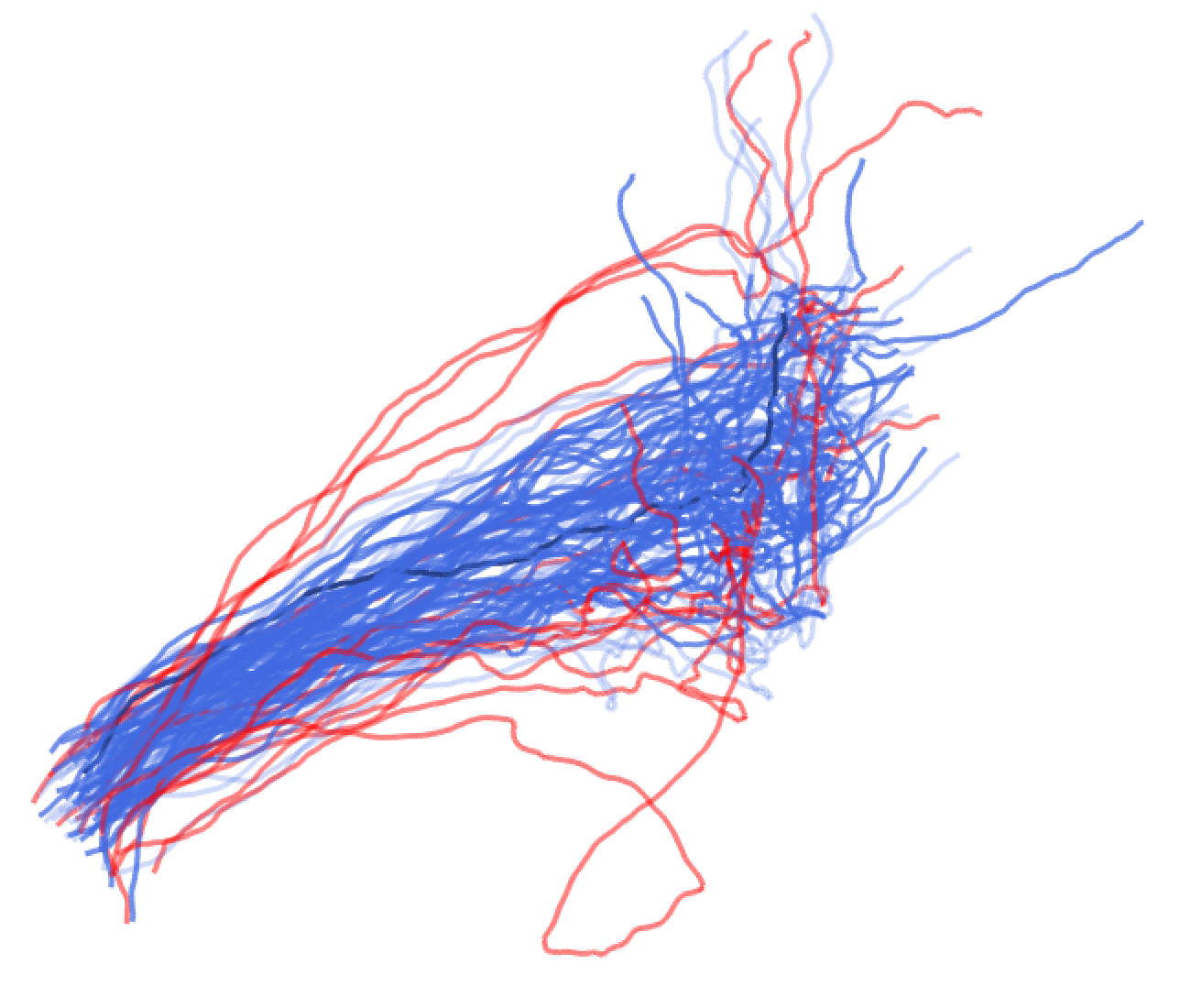}&\includegraphics[width=0.155\textwidth]{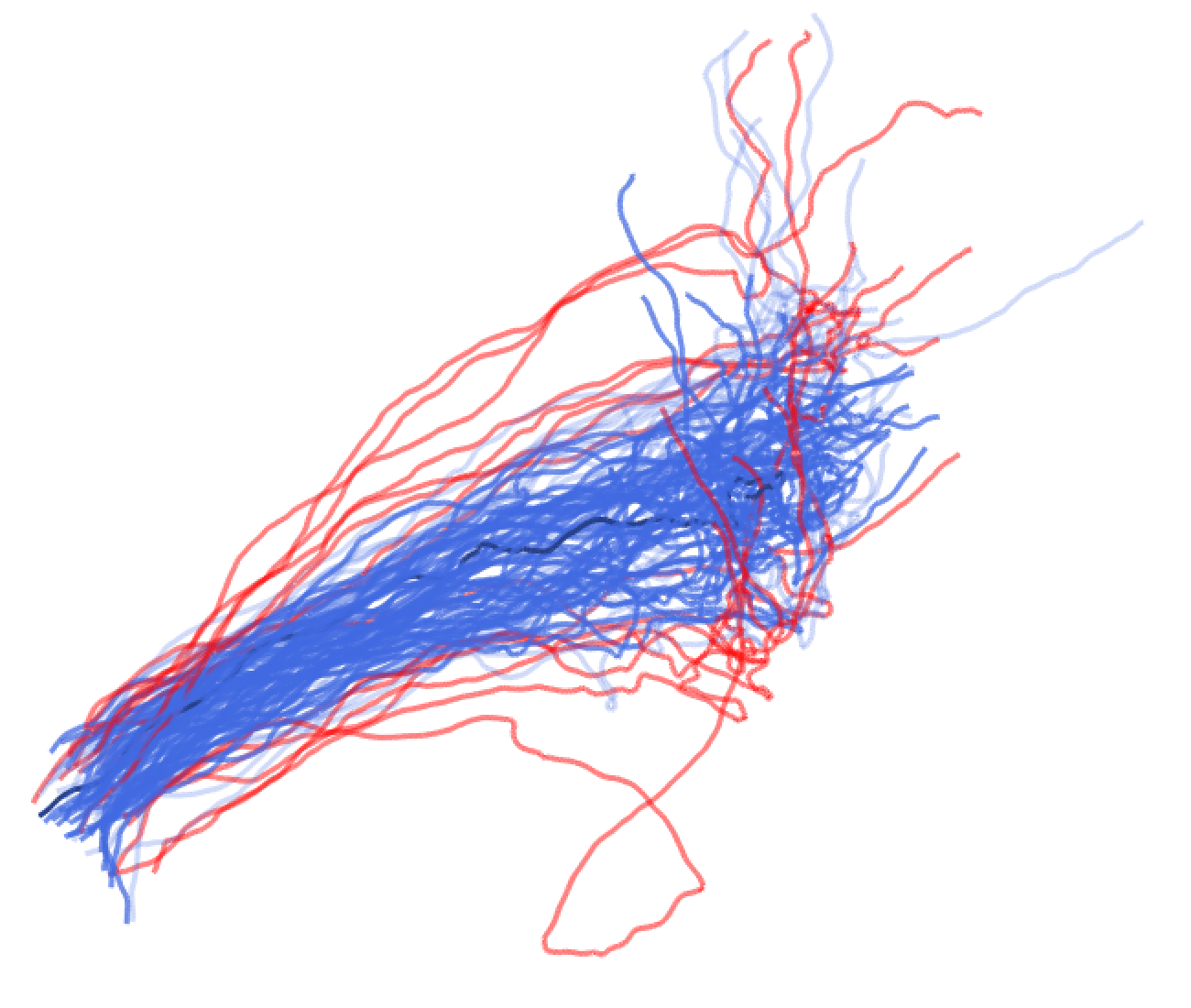} \\
		\includegraphics[width=0.155\textwidth]{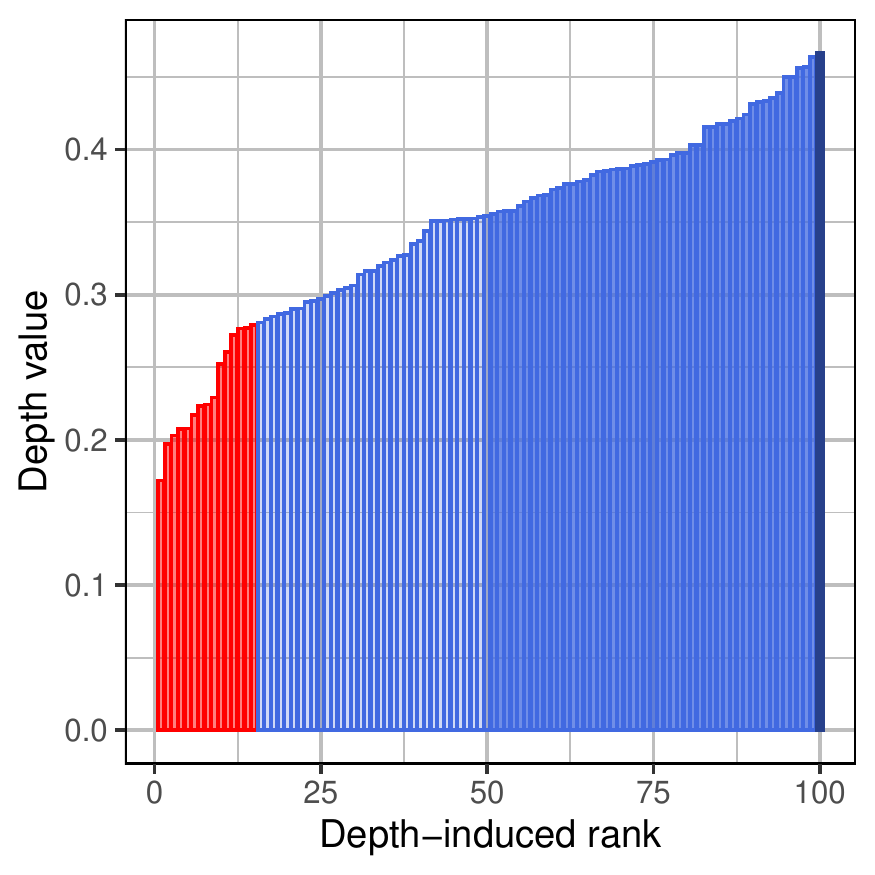}&\includegraphics[
		width=0.155\textwidth]{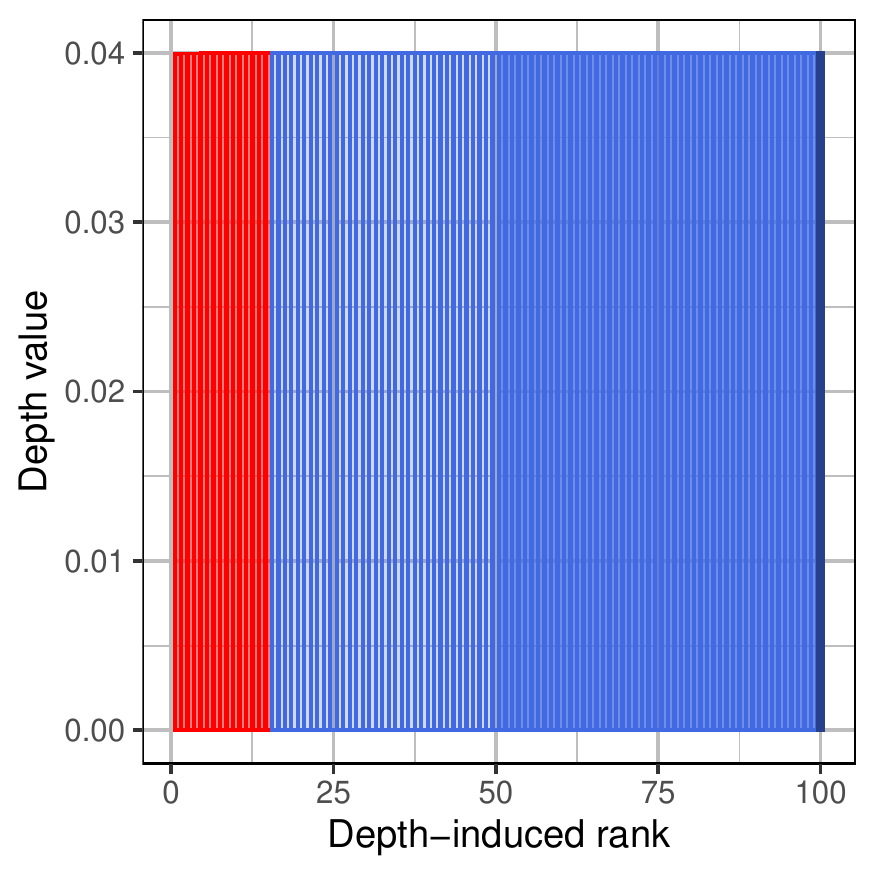}&\includegraphics[width=0.155\textwidth]{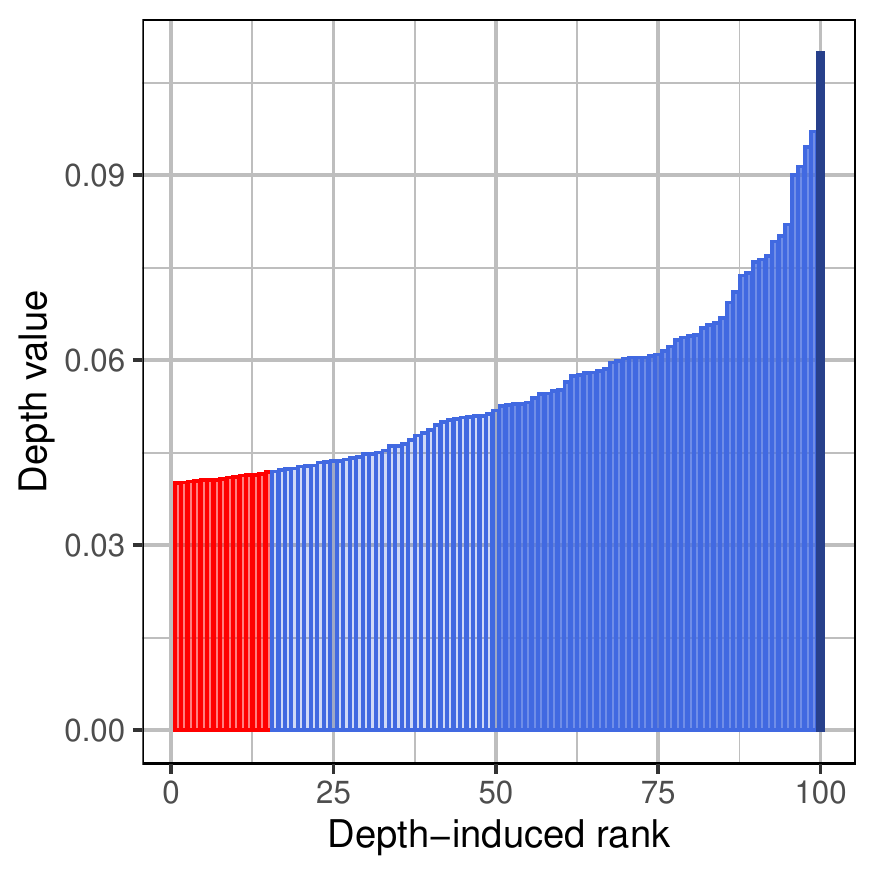}&\includegraphics[width=0.155\textwidth]{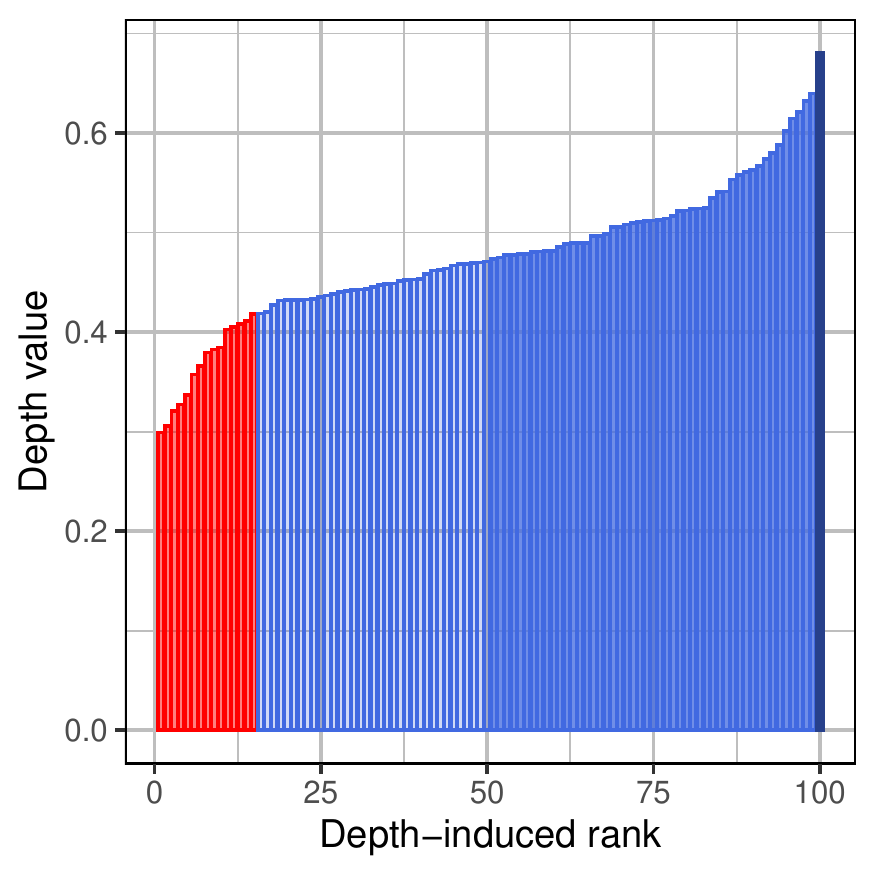}&\includegraphics[width=0.155\textwidth]{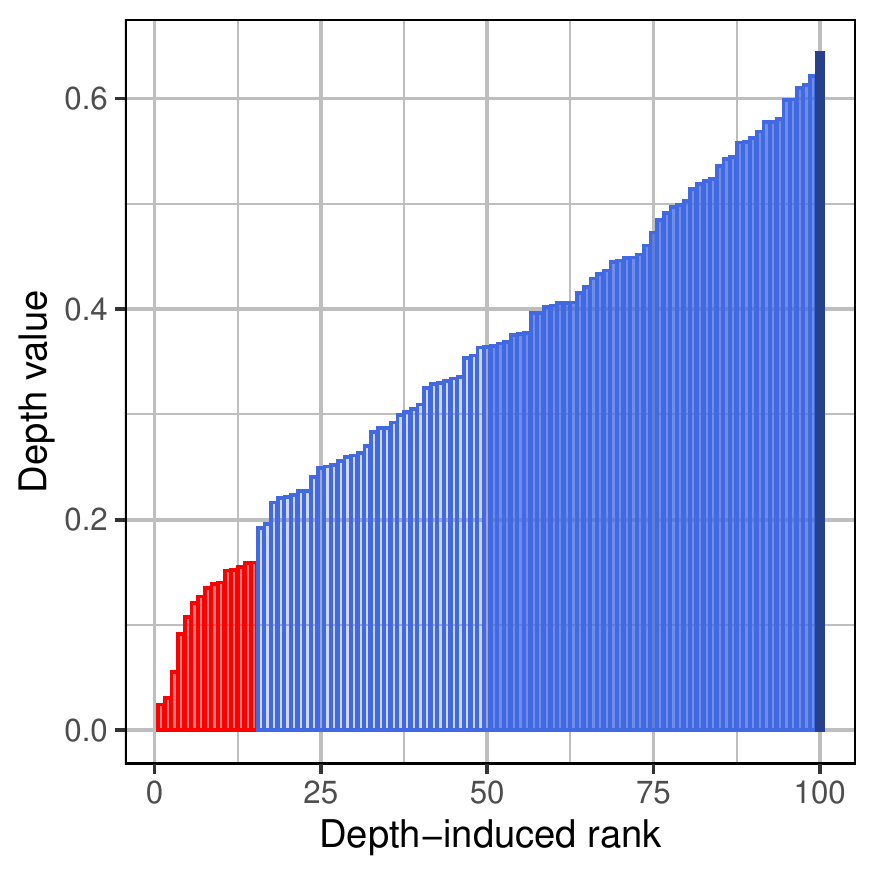}
	\end{tabular}
\end{center}
	\caption{Top: Curve boxplots for a sample of $100$ WM-fibers from the right side of the brain. The set of $100$ curves is partitioned into $15/35/49/1$ curves: $15$ with the smallest value of depth, considered as outliers (red), $35$ with a larger value of depth, considered as outer curves (light blue), $49$ with the largest value of depth, considered as the more central curves (blue), and finally the deepest curve of all (dark blue). Depth methods are (from left to right) mMBD, SBD, mSBD, saPRJ, and our \emph{\curvedepth}. Bottom: Corresponding depth-ranked histograms.}\label{fig:boxplot:brain:comparison}
\end{figure}

\subsubsection{Curve Registration}\label{ssec:curve-registration}

Image registration is one of the main pre-processing steps in any statistical analysis of brain imaging data. Its aim is to geometrically match up image volumes of brain structures, for example for structure localization or difference detection. Broadly, this consists in finding rotation and translation parameters that will minimize a certain cost function (e.g., least squares or mutual information) which quantifies how well aligned two images are. Image registration is an active field of research since existing algorithms still have defects, for example they might suffer from directionality bias \citep{Dojat2014}. All standard libraries dedicated to the analysis of fMRI, MRI and DTI brain imaging data contain an image registration procedure; see, e.g., {\tt RNiftyReg} \citep{Clayden2017} in the R software \citep{RSoftware2019}.

Here, our approach to register the bundles at hand is to first extract one single best representative curve for each bundle (namely the deepest one; see the dark blue fiber in Figure~\ref{fig:boxplot:brain}) and then to match these representatives as best as possible. 
In the twin DTI data set considered here, our aim was to register 68 bundles, of about 1,000 fibers each, located in the left hemisphere say. To reach this goal, we first computed the deepest fiber within each bundle, noted thereafter $d_j$, $j=1,\ldots,68$. We then computed the deepest fiber among $d_1,...,d_{68}$, which is denoted $D$. Finally, for each bundle $j$, we found the rigid transformation (in terms of rotation, translation and centering) that minimizes the distance \eqref{eq:distance} between the curves $d_j$ and $D$. Registration is then achieved by applying each one of these rigid transformations to all the fibers within the corresponding bundle. This process is illustrated in Figure~\ref{fig:registration:brain}.

\begin{figure}
{\small
	\begin{center}
	\begin{tabular}{ccc}
		Subject $104$ & Subject $110$ & Subject $131$ \\
	\includegraphics[width=0.28\textwidth]{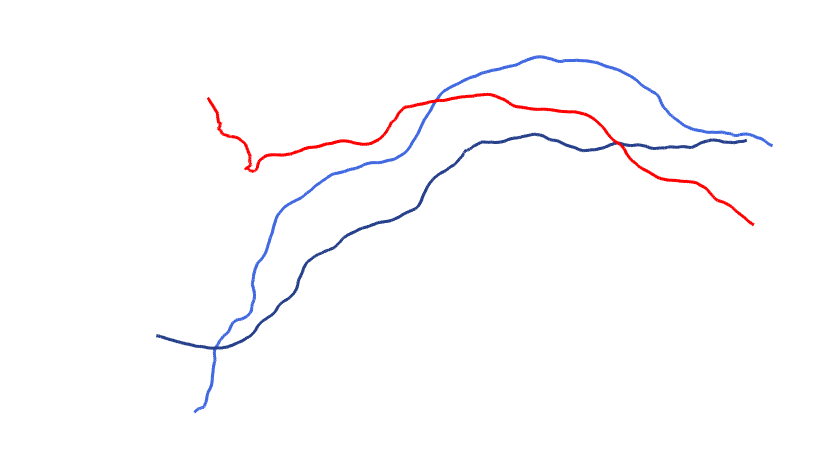} & \includegraphics[width=0.28\textwidth]{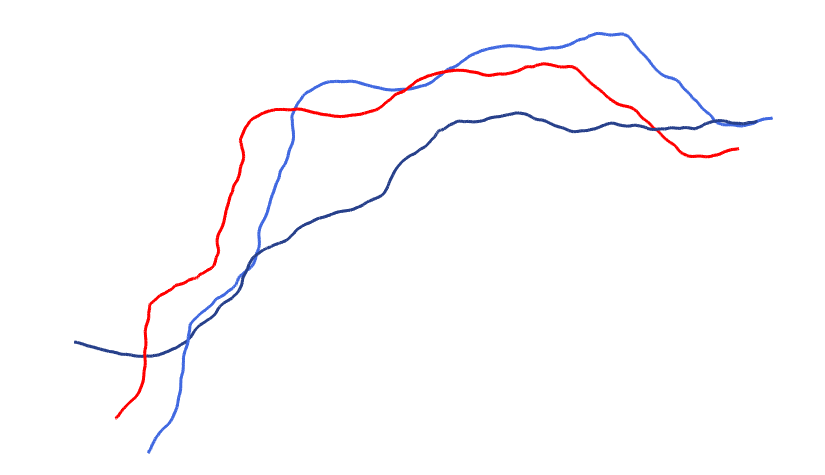} & \includegraphics[width=0.28\textwidth]{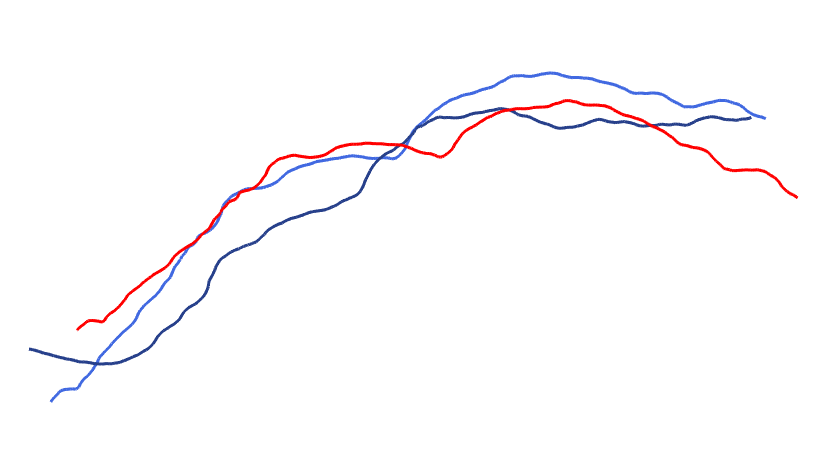} \\
	\end{tabular}
	\end{center}
}
	\caption{Illustration of the registration process. Subject 235 is the reference subject (i.e., the subject whose deepest curve is $D$, the \emph{deepest of all}). The red and the dark blue curves are the deepest curves (before registration) of the given subject and of subject $235$, respectively. We bring the red curve as close as possible, in terms of distance \eqref{eq:distance}, to the dark blue curve. The transformed curve (i.e, after registration) is the light blue curve. Distances from each red curve (i.e., before registration) and from each light blue curve (i.e., after registration) to the deepest of all are $10.271$ and $3.245$ (for subject $104$), $4.539$ and $3.395$ (for subject $110$), and $3.329$ and $2.084$ (for subject $131$), respectively.}
	\label{fig:registration:brain}
\end{figure}

\subsubsection{A Statistical Comparison Between MZ and DZ Twins}\label{ssec:ddplotsCompar}

After having performed curve registration, comparison of the empirical distributions is possible. Given two distributions $P_0,P_1\in\mathcal{P}$ on the space of curves $\Bg$, we consider the mapping that yields the $DD$-plot \citep{liu1999}: 
\begin{equation}
	\Bg\rightarrow[0,1]^2\,,\,\Cc\mapsto \left(D(\Cc|{P_0}),D(\Cc|{P_1})\right).
\end{equation}
For two random samples of curves $\{\Xx_1^{(0)},\ldots,\Xx_{n_0}^{(0)}\}$ and $\{\Xx_1^{(1)},\ldots,\Xx_{n_1}^{(1)}\}$ from $P_0$ and $P_1$ respectively, 
the empirical $DD$-plot can be constructed as:
\begin{eqnarray}
\bigcup_{k=0,1}\left\{\left({D}(\Xx^{(k)}_{i}|\Xx_1^{(0)},\ldots,\Xx_{n_0}^{(0)}),{D}(\Xx^{(k)}_{i}|\Xx_1^{(1)},\ldots,\Xx_{n_1}^{(1)})\right)\,,i=1,\ldots ,n_k\right\}. \nonumber 
\end{eqnarray}


For six pairs of twins, $DD$-plots are presented in Figure~\ref{fig:ddplots:twins}, whose contribution is twofold. First, as a proof of concept, the empirical distributions of two MZ twins are very similar since the points are concentrated around the diagonal of the $DD$-plot while those of DZ twins differ \citep[see also][]{liu1999}. Second, this closeness of the MZ twins underlines the high quality of the curve registration using the the geometrical matching (Section~\ref{ssec:curve-registration}) in the sense that (each of) these two bundles of curves are meant to substantially coincide.

Recently, neuroscientists have discovered that several structures in the brain are influenced by our genetics; see, e.g., \citep{Wen2016}. This suggests a genetically-driven spatial organisation of corticospinal brain fibers. This biological hypothesis can be statistically confirmed by applying the depth-based Wilcoxon testing procedure introduced by \cite{liu1993quality} and further described in \citep{lopez2009concept}. For each pair of twins, we considered $500$ fibers selected at random from the first twin as a reference sample. We then used $50$ fibers from each twin (selected at random among the remaining fibers) to calculate the test statistic value. The $p$-values, computed using the normal asymptotic null distribution given by \cite{Lehmann1975}, are provided in Figure~\ref{fig:ddplots:twins}. They are small for DZ twins and large for MZ twins, a statistical evidence in favour of this biological hypothesis.

\begin{figure}[H]
{\small
	\begin{center}
	\begin{tabular}{ccc}
		\qquad$105$ \textit{vs.} $205$ (DZ) & \qquad$120$ \textit{vs.} $220$ (DZ) & \qquad$132$ \textit{vs.} $232$ (DZ) \\
		\includegraphics[width=0.275\textwidth,trim=0mm 0mm 10mm 15mm,clip=true,page=1]{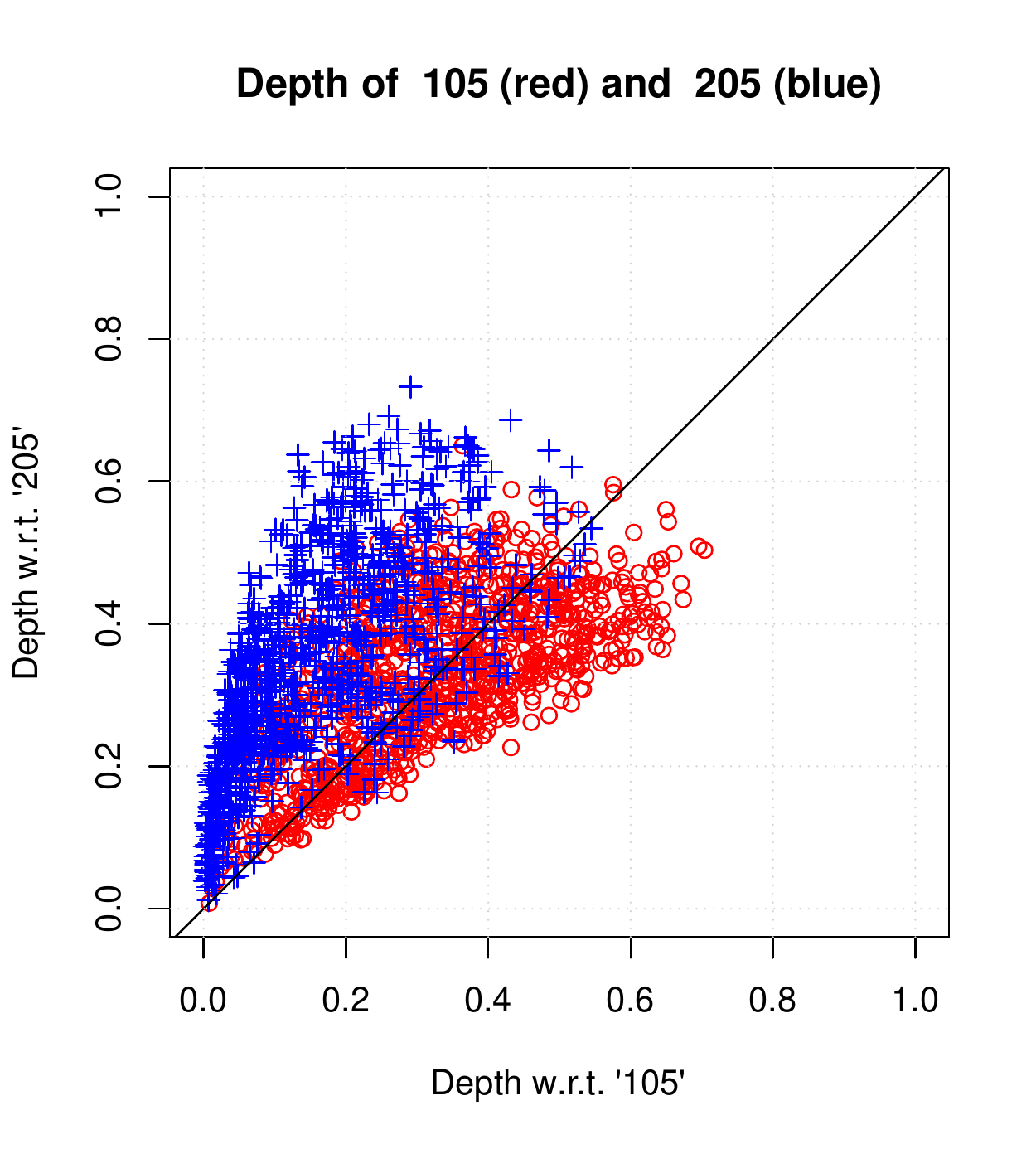} & \includegraphics[width=0.275\textwidth,trim=0mm 0mm 10mm 15mm,clip=true,page=1]{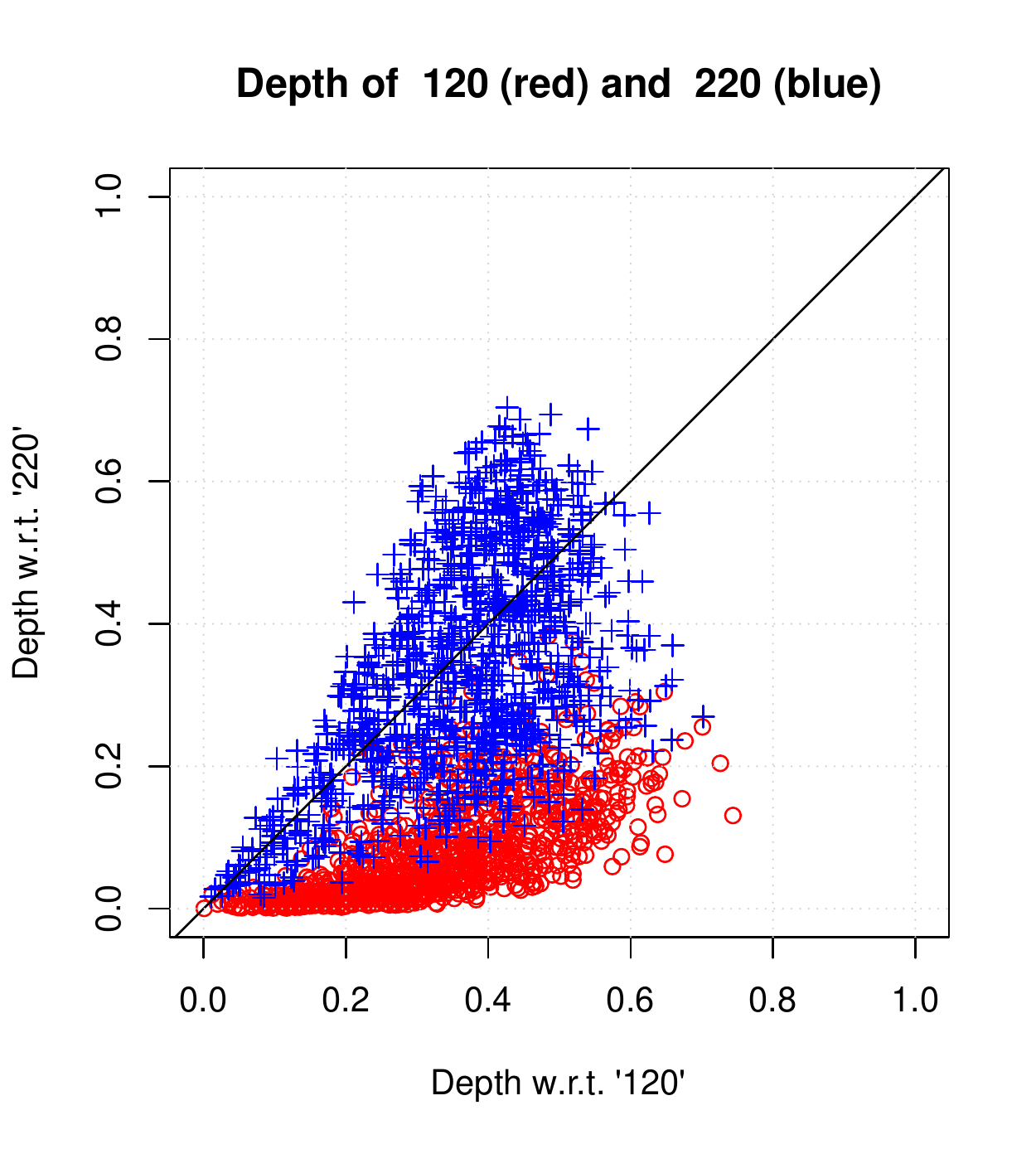} & \includegraphics[width=0.275\textwidth,trim=0mm 0mm 10mm 15mm,clip=true,page=1]{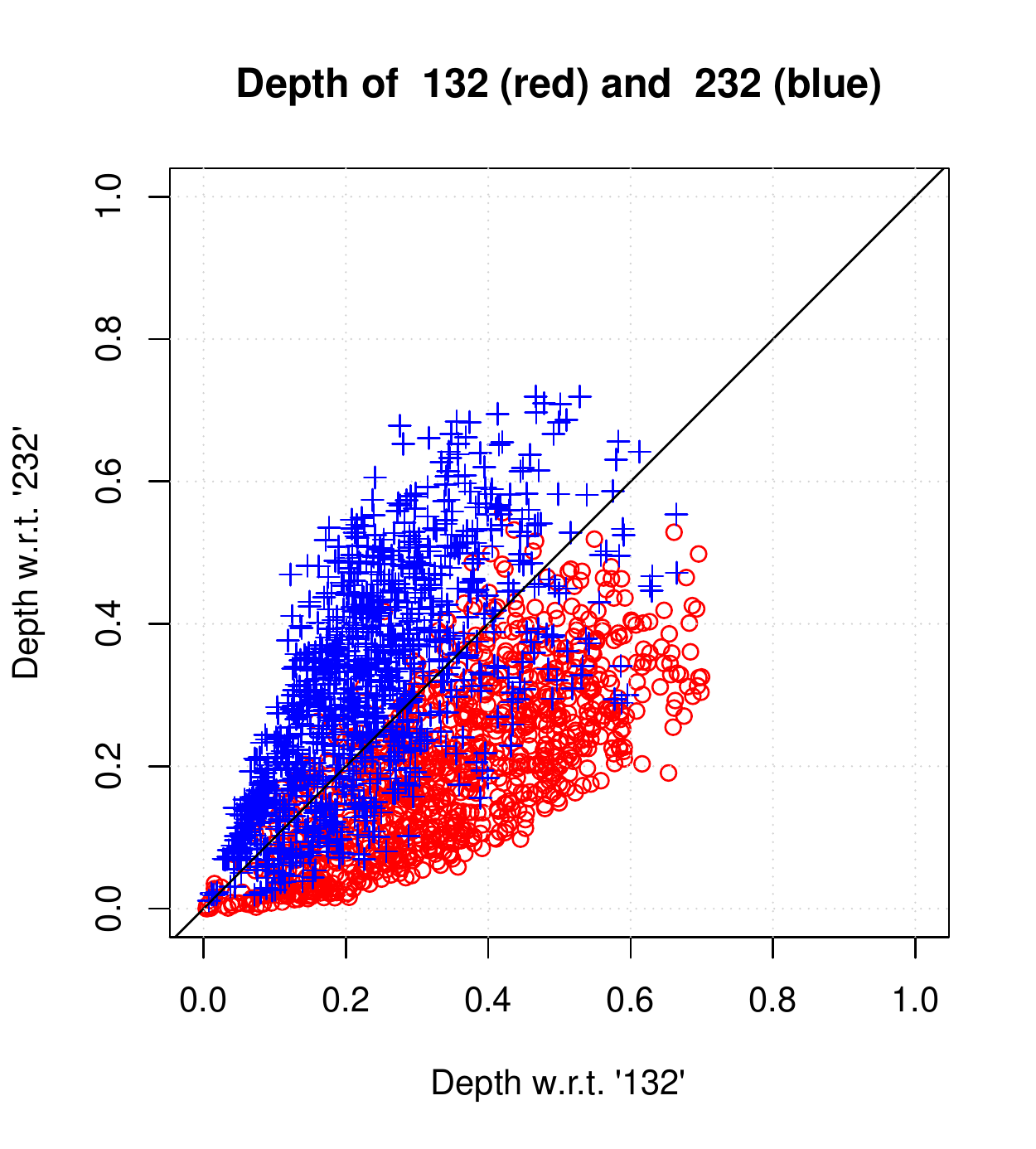} \\
		 \qquad$104$ \textit{vs.} $204$ (MZ) & \qquad$106$ \textit{vs.} $206$ (MZ) & \qquad$131$ \textit{vs.} $231$ (MZ) \\
		\includegraphics[width=0.275\textwidth,trim=0mm 0mm 10mm 15mm,clip=true,page=2]{ddplots-twinsDMZ_120-220-104-204.pdf} & \includegraphics[width=0.275\textwidth,trim=0mm 0mm 10mm 15mm,clip=true,page=2]{ddplots-twinsDMZ_105-205-106-206.pdf} & \includegraphics[width=0.275\textwidth,trim=0mm 0mm 10mm 15mm,clip=true,page=2]{ddplots-twinsDMZ_132-232-131-231.pdf}
	\end{tabular}
	\end{center}}
	\caption{$DD$-plots of six pairs of twins (red circles for 1xx; blue ``+'' signs for 2xx) with associated $p$-values in parenthesis. (Top) three DZ, namely: $105$ and $205$ ($1-e9$), $120$ and $220$ ($0.017$), $132$ and $232$ ($0.003$). (Bottom) three MZ: namely $104$ and $204$ ($0.733$), $106$ and $206$ ($0.366$), $131$ and $231$ ($0.366$).}
	\label{fig:ddplots:twins}
\end{figure}

\subsection{Classification Algorithms for Unparameterized Curves}\label{sec:classif}

Automatic clustering of white matter fibers is an important sub-task in understanding brain connectivity and integrity, see e.g., \cite{jin2014automatic}.
With this motivation in mind, we extend to the context of curves two classification algorithms: the $DD$-plot procedure \citep{li2012} and the unsupervised depth-based clustering \citep{jornsten2004clustering}.
To illustrate the performance of these two procedures when used in conjunction with our \curvedepth, we start by considering the problem of recognition of hand-written digits from the now famous training MNIST data set\footnote{\url{http://yann.lecun.com/exdb/mnist/}}. This is done in a supervised way in Section~\ref{ssec:appclass} and in an unsupervised way in Section \ref{ssec:appclust}. Finally, in Section~\ref{ssec:appclustdti} we produce an unsupervised clustering of the DT-MRI brain fibers in a data set previously studied by \citet[Section 4]{kurtek2012}.


\subsubsection{Supervised Classification of Hand-written Digits}\label{ssec:appclass}

As a proof-of-concept, we show that our \curvedepth\ can be used to produce a linear classifier able to discriminate between two classes of images representing the digits `0' and~`1'. We illustrate its results on $100$ observations from each class. The original MNIST images have been preprocessed in order to transform them into pixelized curves (i.e., each pixel of an image should have at most two neighboring pixels on the vertical, horizontal and diagonal directions). A few examples of the preprocessed digit images are plotted in Figure~\ref{fig:digits}.


\def\wb{0.135}
\newcommand{\trimval}{39mm 30mm 24mm 30mm}
\begin{figure}[H]
\begin{center}
\frame{\includegraphics[keepaspectratio=true,scale=\wb,trim=30.3mm 27.8mm 20mm 22.8mm,clip=true,page=1]{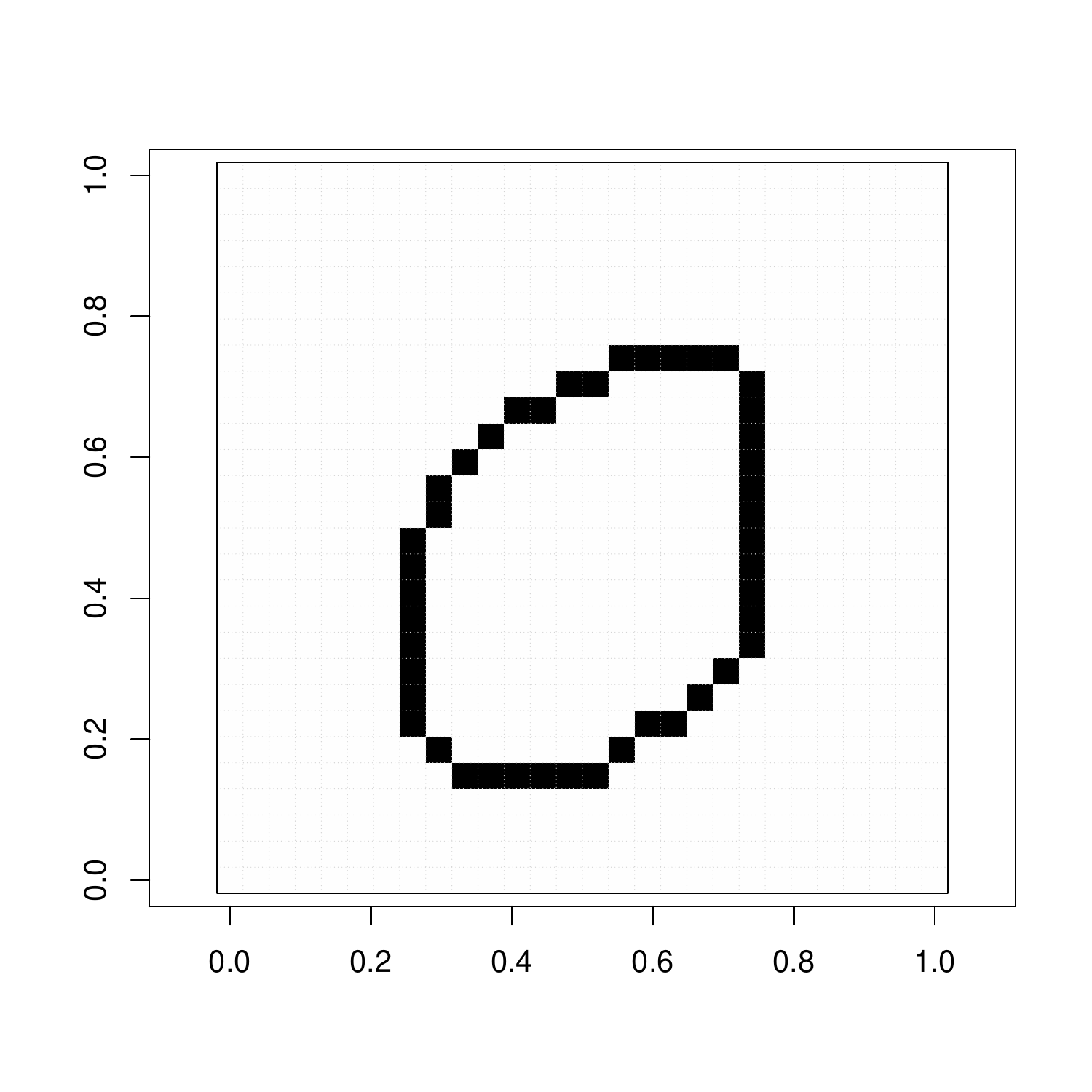}} \frame{\includegraphics[keepaspectratio=true,scale=\wb,trim=30.3mm 27.8mm 20mm 22.8mm,clip=true,page=2]{digits-zero.pdf}} \frame{\includegraphics[keepaspectratio=true,scale=\wb,trim=30.3mm 27.8mm 20mm 22.8mm,clip=true,page=3]{digits-zero.pdf}} \frame{\includegraphics[keepaspectratio=true,scale=\wb,trim=30.3mm 27.8mm 20mm 22.8mm,clip=true,page=4]{digits-zero.pdf}} \frame{\includegraphics[keepaspectratio=true,scale=\wb,trim=30.3mm 27.8mm 20mm 22.8mm,clip=true,page=5]{digits-zero.pdf}} \frame{\includegraphics[keepaspectratio=true,scale=\wb,trim=30.3mm 27.8mm 20mm 22.8mm,clip=true,page=6]{digits-zero.pdf}} \frame{\includegraphics[keepaspectratio=true,scale=\wb,trim=30.3mm 27.8mm 20mm 22.8mm,clip=true,page=7]{digits-zero.pdf}} \frame{\includegraphics[keepaspectratio=true,scale=\wb,trim=30.3mm 27.8mm 20mm 22.8mm,clip=true,page=8]{digits-zero.pdf}} \frame{\includegraphics[keepaspectratio=true,scale=\wb,trim=30.3mm 27.8mm 20mm 22.8mm,clip=true,page=9]{digits-zero.pdf}} \frame{\includegraphics[keepaspectratio=true,scale=\wb,trim=30.3mm 27.8mm 20mm 22.8mm,clip=true,page=10]{digits-zero.pdf}}\vspace{0.35ex}\\
\frame{\includegraphics[keepaspectratio=true,scale=\wb,trim=30.3mm 27.8mm 20mm 22.8mm,clip=true,page=1]{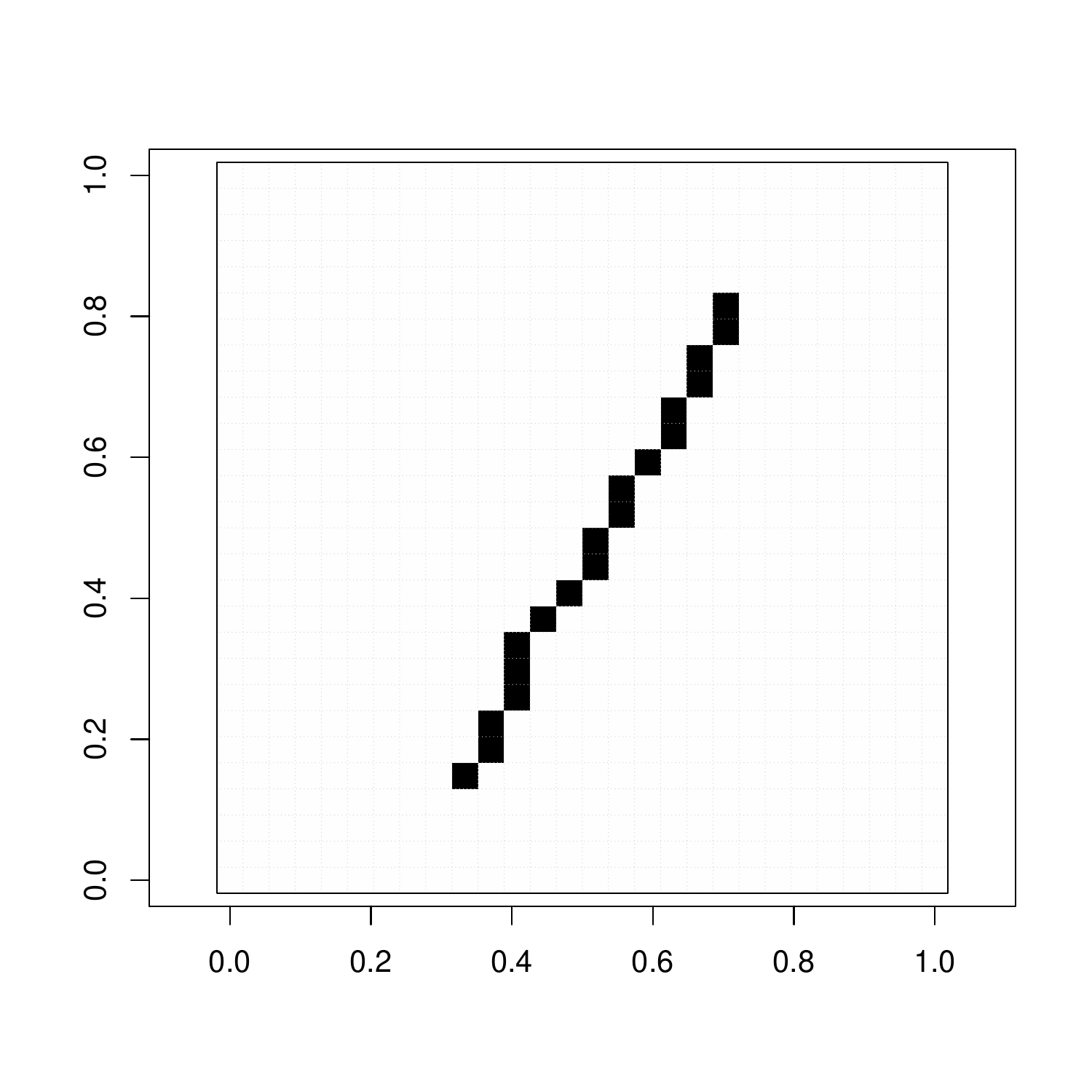}} \frame{\includegraphics[keepaspectratio=true,scale=\wb,trim=30.3mm 27.8mm 20mm 22.8mm,clip=true,page=2]{digits-one.pdf}} \frame{\includegraphics[keepaspectratio=true,scale=\wb,trim=30.3mm 27.8mm 20mm 22.8mm,clip=true,page=3]{digits-one.pdf}} \frame{\includegraphics[keepaspectratio=true,scale=\wb,trim=30.3mm 27.8mm 20mm 22.8mm,clip=true,page=4]{digits-one.pdf}} \frame{\includegraphics[keepaspectratio=true,scale=\wb,trim=30.3mm 27.8mm 20mm 22.8mm,clip=true,page=15]{digits-one.pdf}} \frame{\includegraphics[keepaspectratio=true,scale=\wb,trim=30.3mm 27.8mm 20mm 22.8mm,clip=true,page=6]{digits-one.pdf}} \frame{\includegraphics[keepaspectratio=true,scale=\wb,trim=30.3mm 27.8mm 20mm 22.8mm,clip=true,page=7]{digits-one.pdf}} \frame{\includegraphics[keepaspectratio=true,scale=\wb,trim=30.3mm 27.8mm 20mm 22.8mm,clip=true,page=8]{digits-one.pdf}} \frame{\includegraphics[keepaspectratio=true,scale=\wb,trim=30.3mm 27.8mm 20mm 22.8mm,clip=true,page=9]{digits-one.pdf}} \frame{\includegraphics[keepaspectratio=true,scale=\wb,trim=30.3mm 27.8mm 20mm 22.8mm,clip=true,page=10]{digits-one.pdf}}
\end{center}
\caption{First ten images from each of the two classes of curve-preprocessed MNIST digits.}
\label{fig:digits}
\end{figure}



A $DD$-plot built using our \curvedepth\ (see Section~\ref{ssec:ddplotsCompar}) can be exploited to classify curves. Indeed, this task is greatly simplified in the $DD$-plot space since a rule separating two classes needs only to be found in a Euclidean space of dimension two.
For the sample consisting of 100 `0's and 100 `1's, we applied the $DD\alpha$-procedure (an iterative heuristics in the $DD$-plot; see \cite{lange2014fast} for a detailed description). The resulting separation rule is plotted in solid green on Figure~\ref{fig:ddplotsep}.

\begin{figure}[t]
	\begin{center}
		\includegraphics[scale=0.675,trim = 0mm 0mm 0mm 0mm, clip]{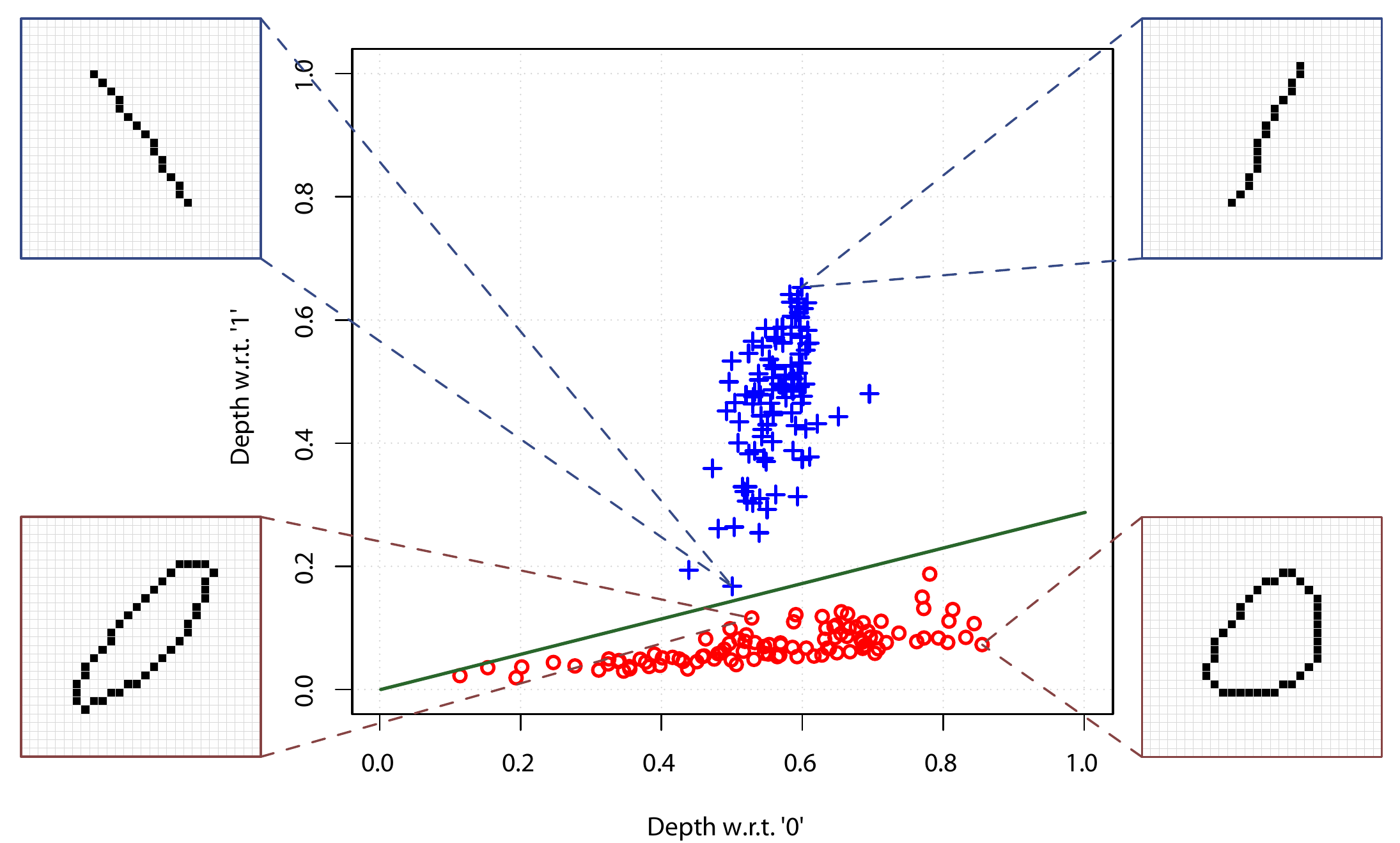}
	\end{center}
	\caption{$DD\alpha$-classifier for a subsample of 100 `0's and 100 `1's taken from the MNIST data set. For each one of the two classes, `magnified' observations correspond to the one having the highest depth value in its class (on the right-hand side), and to the one lying closest to the opposite class (on the left-hand side).}
	\label{fig:ddplotsep}
\end{figure}

One can observe (for this particular sample) the perfect separation of the two classes by a linear rule. 
In Figure~\ref{fig:ddplotsep}, on the right-hand side of the $DD$-plot `magnified' observations (`1' and `0') having the highest depth in each class are pictured; they are trivially well classified.
On the left-hand side of the $DD$-plot, we paint the most doubtful observations, i.e., those lying closest to a member of the opposite class.
The `1' here corresponds to the observation with the lowest depth in the sample of `1's; this can also be regarded as an atypical observation.
The situation is different with the `0' lying closest to the set of `1's. 
It has a rather average depth in its own class, but due to its oblong shape resembles a `1' and thus has a high depth value in the class of `1's relative to its depth in the class of `0's.

\subsubsection{Unsupervised Classification of Hand-written Digits}\label{ssec:appclust}

\cite{jornsten2004clustering} proposes the DDClust algorithm for clustering.
This non-parametric method is based on both distance-based distortion (captured by the silhouette width) and geometry of the curves (captured by the relative depth). 
We propose the original method with slight modifications and we illustrate it on the MNIST-digits data. 
Let $\{\Cc_1,\ldots,\Cc_n\}$ be an observed sample of curves from $\Bg.$
Our aim is to partition the data set into $K$ groups.
DDclust proceeds iteratively by assigning a curve $\Cc_i$ at each instance to the cluster where it has the highest depth.

For $k=1,\ldots,K,$ we denote by $I_k$ the set of indices of observations belonging to the cluster $k$ and by $P_k$ the probability measure on $\Bg$ defined as
$$
P_k = \frac{1}{n_k} \sum_{i\in I_k} \delta_{\Cc_i},
$$ 
where $n_k$ is the size of the cluster $I_k.$
Then $\Ii=\{I_1,\ldots,I_K\}$ is a partition of $\{1,\ldots,n\}.$

The within-cluster data depth of an observation $i\in I_k$ is $D(\Cc_i|P_k).$
The between-cluster data depth of an observation $i\in I_k$ is $\min_{\ell\neq k} D(\Cc_i|P_\ell).$
The relative depth of an observation $i\in I_k$ is then defined as
\begin{equation}
	ReD_i(\Ii) = D(\Cc_i|P_k) - \min_{\ell\neq k} D(\Cc_i|P_\ell).
\end{equation}
The within-cluster average distance of an observation $\Cc_i\in I_k$ is 
$$d(\Cc_i|k)=\frac{1}{n_k-1}\sum_{j\in I_k\setminus\{i\}} d_\Bg(\Cc_i,\Cc_j),$$ where $n_k-1$ is the size of $I_k\setminus\{i\}.$
The closest average distance of an observation $i\in I_k$ among foreign clusters is $\min_{\ell\neq k} d(\Cc_i|\ell).$
The silhouette width of an observation $i$ belonging to cluster $k$ is
\begin{equation}
	Sil_i(\Ii) = \frac{\min_{\ell\neq k} d(\Cc_i|\ell) - d(\Cc_i|k)}{\max\{d(\Cc_i|k),\min_{\ell\neq k} d(\Cc_i|\ell)\}}.
\end{equation} 
The clustering cost of an observation $i$ for the partition $\Ii=\{I_1,\ldots,I_K\}$ is
\begin{equation}
	C_i(\Ii) =  (1 - \lambda)Sil_i(\Ii) + \lambda ReD_i(\Ii),
\end{equation}
where $\lambda\in[0,1]$ being a constant defining trade-off between depth and distance. 
The total clustering cost can then be formulated as
\begin{equation}
	C(\Ii) =  \frac{1}{n}\sum_{k=1}^K \sum_{i\in I_k} C_i(\Ii).
\end{equation}

Here we employ the original clustering algorithm by \cite{jornsten2004clustering} with slight modifications, which we briefly describe right below and send the reader to the source for details. 
For a fixed number of clusters $K$, we start with an initial partition $\Ii$ which may be generated at random. 
For each observation $i=1,\ldots, n,$ we compute its clustering cost $C_i(\Ii).$
Then the set of observations considered for a potential reallocation is defined as the set of indices:
$$
\Rr = \{ i\ :\ C_i(\Ii) < T\},
$$
where $T\leq0$ is a prefixed threshold.
For a random subset $E$ from $\Rr,$ we reallocate each index in $E$ to its closest cluster (the one with highest depth for this observation) getting a new partition $\tilde{\Ii}$ that is accepted if $C(\tilde{\Ii})>C(\Ii)$ and with probability $1-\exp\left(\beta(C(\Ii)-C(\tilde{\Ii}))\right)/2$ otherwise ($\beta$ is a temperature parameter). 
The whole procedure is given in Algorithm~\ref{alg:clust}, which can be found in Section~\ref{app:ssec:clustcurv} of the Supplementary Materials.



We ran our clustering algorithm {\textsc{DDclustCurve}} (with $K=3$) on a set of $300$ preprocessed MNIST images of the digits `0', `1' and `7. The results are very satisfactory (empirical error rate $=1\%$, $3$ errors). The resulting $C_i(\mathcal{I})$-s are plotted in Figure~\ref{fig:clustReD}. 

\begin{figure}[H]
	\begin{center}
		\includegraphics[width=.8\textwidth,trim = 0mm 5mm 0mm 20mm, clip]{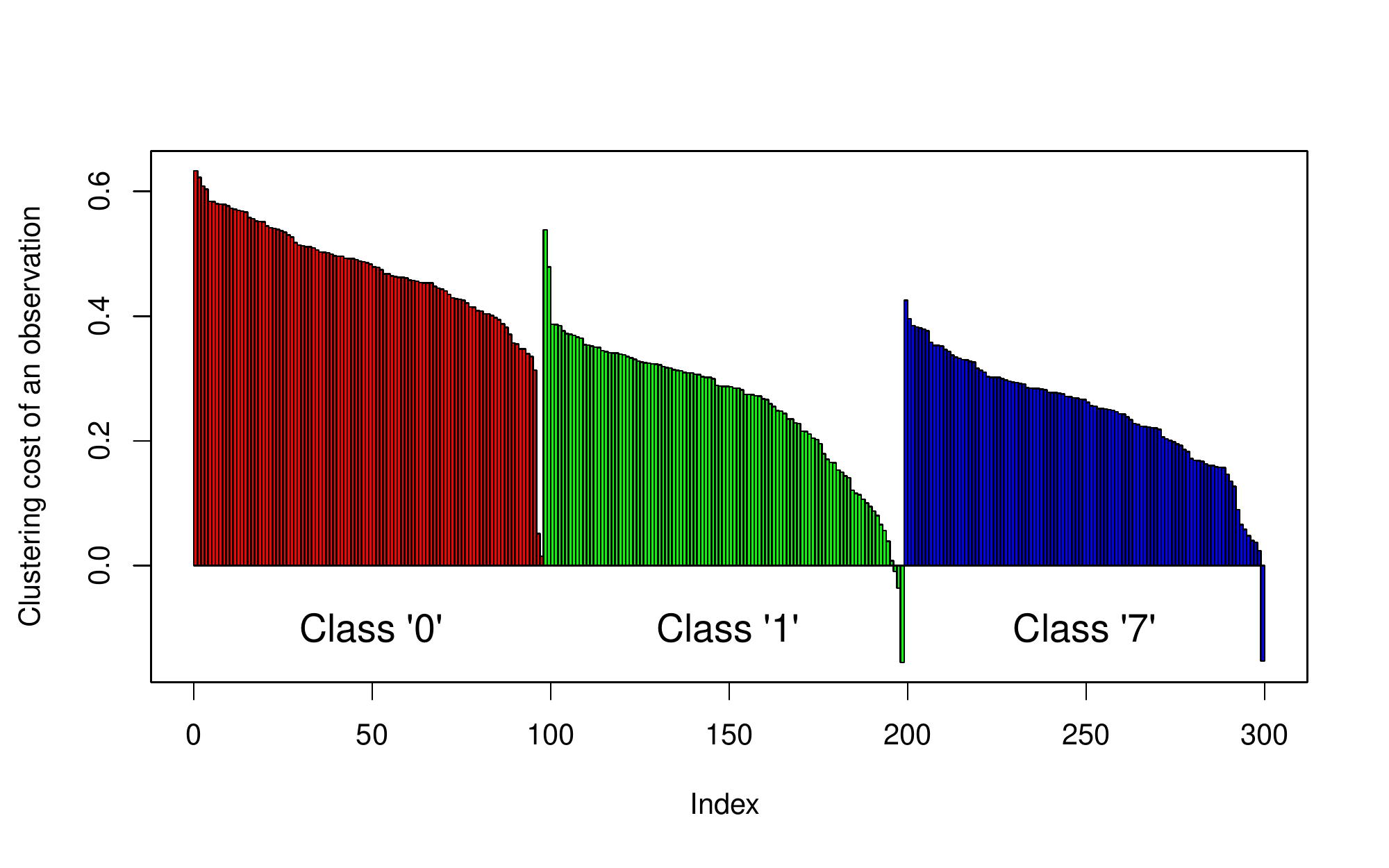}
	\end{center}
	\caption{Clustering cost (ordered decreasingly within each class) of $300$ digits from the MNIST library after convergence of the clustering algorithm (Algorithm~\ref{alg:clust} in the Supplementary Materials, an adaptation of \cite{jornsten2004clustering}'s algorithm). The colors correspond to the correct classes of digits `0' (red), `1' (green), and `7' (blue). According to the clustering criterion (threshold $T$ set at $0$), only $3$ observations from class `1' and $1$ observation from class `7' are misclassified; the (true) clustering error is $1\%$ (or $3$ observations).}
	\label{fig:clustReD}
\end{figure}

\subsubsection{Unsupervised Classification of DT-MRI Fiber Tracts}\label{ssec:appclustdti}

To further illustrate the exploratory potential of the proposed depth notion, we additionnaly apply the clustering Algorithm \ref{alg:clust} by \cite{jornsten2004clustering} to the DT-MRI brain fibers considered previously by \citet[Section 4]{kurtek2012}.
Automatic clustering of white matter fibers is an important sub-task in understanding brain connectivity and integrity, see e.g., \cite{jin2014automatic}.

The data consist of one bundle of fibers for each one of four subjects. These bundles contain $176,$ $68,$ $48$ and $88$ fibers respectively. 
The results of our clustering coincide for subjects 1 and 3 with those obtained by \citet[Figure 4]{kurtek2012}. The results differ for subjects 2 and 4 but our own interpretation is geometrically sound; see Figure \ref{fig:Kurtek}.
For subject 2, the original red and blue groups in \cite{kurtek2012} are grouped together into one single group (the red one in Figure~\ref{fig:Kurtek}),  while their original green group is split in two parts (green and blue in Figure~\ref{fig:Kurtek}). It is worthwhile noting that a closer look to the scatter plot in \citet[bottom of second column in Figure 4]{kurtek2012} tends to justify this splitting. 
Subject 4, on the other hand, illustrates that our approach takes into account different features of the data. 
 
\begin{figure}[t]
	\begin{center}
	\begin{tabular}{cc}
	Subject 1 $(n=176)$ & Subject 2 $(n=68)$\\
	\includegraphics[width=0.4\textwidth,trim = 0mm 0mm 0mm 0mm, clip]{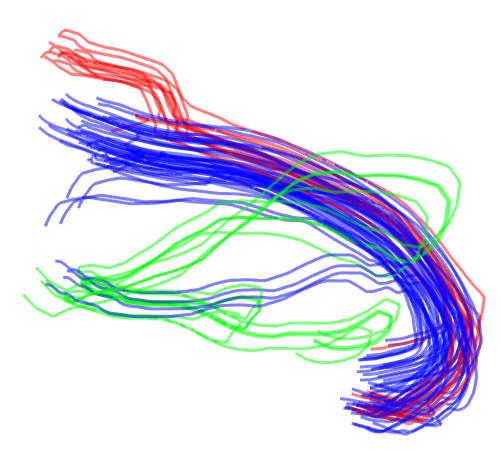} &\includegraphics[width=0.4\textwidth,trim = 0mm 0mm 0mm 0mm, clip]{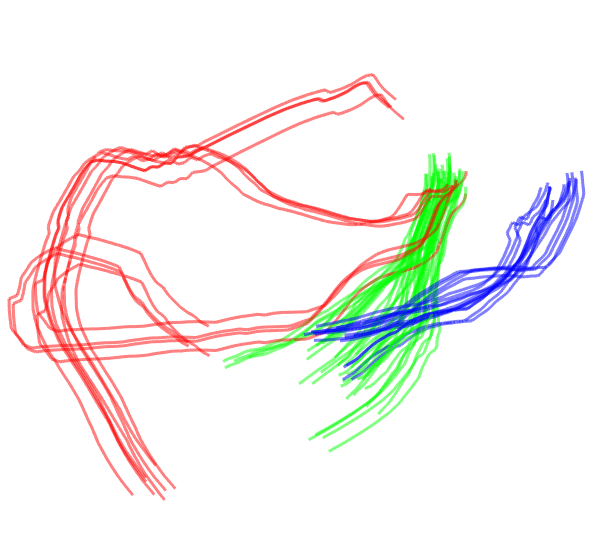}\\
	Subject 3 $(n=48)$ & Subject 4 $(n=88)$\\
		\includegraphics[width=0.4\textwidth,trim = 0mm 0mm 0mm 0mm, clip]{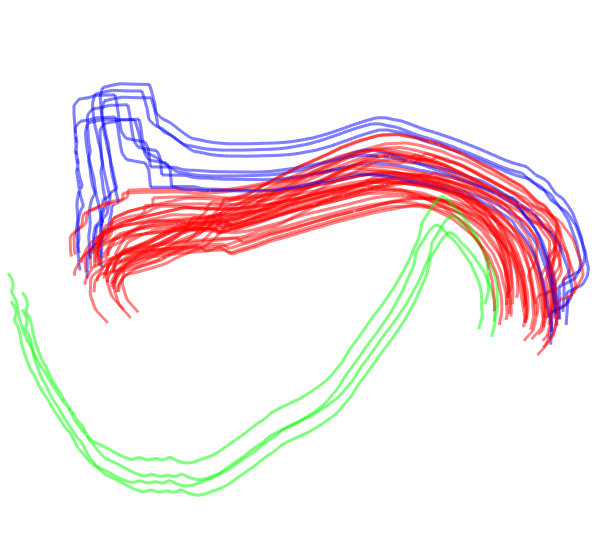} &\includegraphics[width=0.4\textwidth,trim = 0mm 0mm 0mm 0mm, clip]{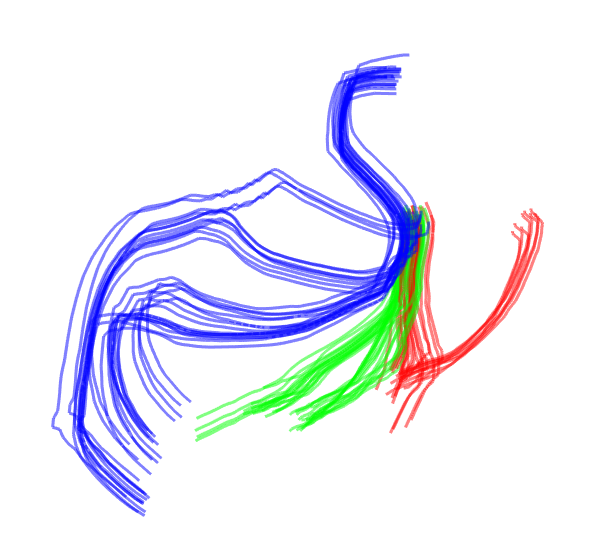}

	\end{tabular}
	\end{center}
	\caption{Clustering of DT-MRI fibers \citep[Source:][]{kurtek2012}.}
	\label{fig:Kurtek}
\end{figure}

\section{Concluding Remarks}\label{sec:Conclusion}

In this work, we introduced a new notion of depth for continuous curves having finite length and we investigated its properties (boundedness, similarity invariance, vanishing at infinity). By construction, our curve depth is invariant to reparametrizations and it is defined on a non-linear space, namely the space of unparametrized curves. It is applicable to curve data embedded in a space of any (finite) dimension. It is a tool that can advantageously compete with functional data depths when dealing with curve objects that should not be considered as functional data (e.g., DTI data). In that sense, our curve depth can be seen as an extension of the notion of statistical depth function to non-standard data types; see also the works by \cite{ley2014} and \cite{paindaveine2018}.  We envision a rich palette of applications for this curve depth. We gave various examples of its use, e.g., for the spatial alignment or unsupervised classification of brain fibers. We illustrated its superiority to other existing depth methods for some applications, for instance in terms of its ability to detect spatial outliers. One can think of other interesting applications of our curve depth, e.g., for handling handwriting data, or 2D and 3D trajectories of animal species or vehicles. A ready-to-use implementation of algorithms that approximate depths of curves via Monte Carlo or that compute the distance between two curves suggests a basis for direct application of the developed methodology in other contexts. Being the most time demanding part of the algorithm, the computation of our point \curvedepth\ can be performed efficiently in dimension two while approximations can be successfully used in higher dimensions, which is illustrated in the performed experiments. These computations are moderately sensitive to the choice of the size $m$ of Monte Carlo samples or smooth curves. This is confirmed by simulation and for real data applications (e.g., we took $m=50$ to cluster brain imaging data obtained from \cite{kurtek2012}). Implementation of the proposed methodology can be found in the \texttt{R}-package \texttt{curveDepth} \citep{curveDepth} available on the CRAN \citep{RSoftware2019}. The data on brain fibers used in this article are available from the authors.

\section*{Supplementary Materials}

\begin{description}
	\item[Additional results:] These contain theoretical details on the space of curves, definitions of our \curvedepth\ function and its properties, algorithms, additional simulation results and some details on data preprocessing. (``CurveDepthSupplement.pdf'')
	\item[Reproducing scripts:] Reproducing \texttt{R}-scripts for experiments contained in the article with descriptions included in files. (``CurveDepthReproduce.zip'')
\item[Animations:] A depth-colored animation of a few brain fibers (\url{http://biostatisticien.eu/DataDepthFig8}) and an illustration of two parametrizations of an `S'-shaped curve (\url{http://biostatisticien.eu/EquivalentCurves}).
\end{description}







\section*{Acknowledgments}

The authors are grateful to Wei Wen for providing the OATS dataset and to Zhaohua Ding and Sebastian Kurtek for providing the DT-MRI brain fibers dataset used in Section~\ref{sec:classif}. We would like to thank Gery Geenens, Karl Mosler and Lionel Truquet  for fruitful discussions about theoretical aspects. This paper includes results produced on the computational cluster Katana at UNSW Sydney.

\bibliographystyle{agsm}


\clearpage

\setcounter{section}{0}

\begin{center}
 {\LARGE Depth for Curve Data and Applications\\ Supplementary Materials} \\
    \vspace{.5cm}
  \author{Pierre Lafaye de Micheaux\hspace{.2cm}\\
    School of Mathematics and Statistics, UNSW Sydney,  Australia\\
    and \\
    Pavlo Mozharovskyi \\
    LTCI, T\'{e}l\'{e}com Paris, Institut Polytechnique de Paris \\
    and \\
    Myriam Vimond \\
    Univ Rennes, Ensai, CNRS, CREST - UMR 9194}

    \vspace{.5cm}
        
    \date{February 20, 2020}
\maketitle
\end{center}

\section{Impact of Parametrization on Functional Depth}\label{app:illustration}

\subsection{Simulated S Letters}\label{app:illustration:S}

We parameterize a $2D$ S-shaped curve (the red one in Figure~\ref{app:fig:curves3d}~(a)) using either parametrization A:
\begin{align}
\label{eqn:parA}
\begin{split}
	x_1(t) &= - \bigl(\cos(t) + 1\bigr)\indicatorf\{t < \frac{3\pi}{2}\} -\bigl(\cos(3t - 3\pi) + 1\bigr)\indicatorf\{t \ge \frac{3\pi}{2}\} + 1,
\\
	x_2(t) &= \phantom{-}\bigl(\sin(t) + 1\bigr) \indicatorf\{t < \frac{3\pi}{2}\} - \bigl(\sin(3t - 3\pi) + 1\bigr)\indicatorf\{t \ge \frac{3\pi}{2}\},
\end{split}
\end{align}
or parametrization B:
\begin{align}
\label{eqn:parB}
\begin{split}
    	x_1(t) &= - \bigl(\cos(3t) + 1\bigr)\indicatorf\{t < \frac{\pi}{2}\} - \bigl(\cos(t + \pi) + 1\bigr)\indicatorf\{t \ge \frac{\pi}{2}\} + 1,
\\
    	x_2(t) &= \phantom{-}\bigl(\sin(3t) \,+ 1\bigr)\indicatorf\{t < \frac{\pi}{2}\} - \bigl(\sin(t + \pi) + 1\bigr)\indicatorf\{t \ge \frac{\pi}{2}\}.
\end{split}
\end{align}

For parametrization A (\ref{eqn:parA}), the argument $t$ ``moves slowly'' on the first half of the curve while it ``moves fast'' on the second half. This pattern is reversed for parametrization~B (\ref{eqn:parB}); see Figure~\ref{app:fig:parametrizations} and also \url{http://biostatisticien.eu/EquivalentCurves} for an interactive visualization.

\begin{figure}[H]
\begin{center}
	\includegraphics[keepaspectratio=true,scale=0.475,trim=0mm 0mm 0mm 10mm,clip=true]{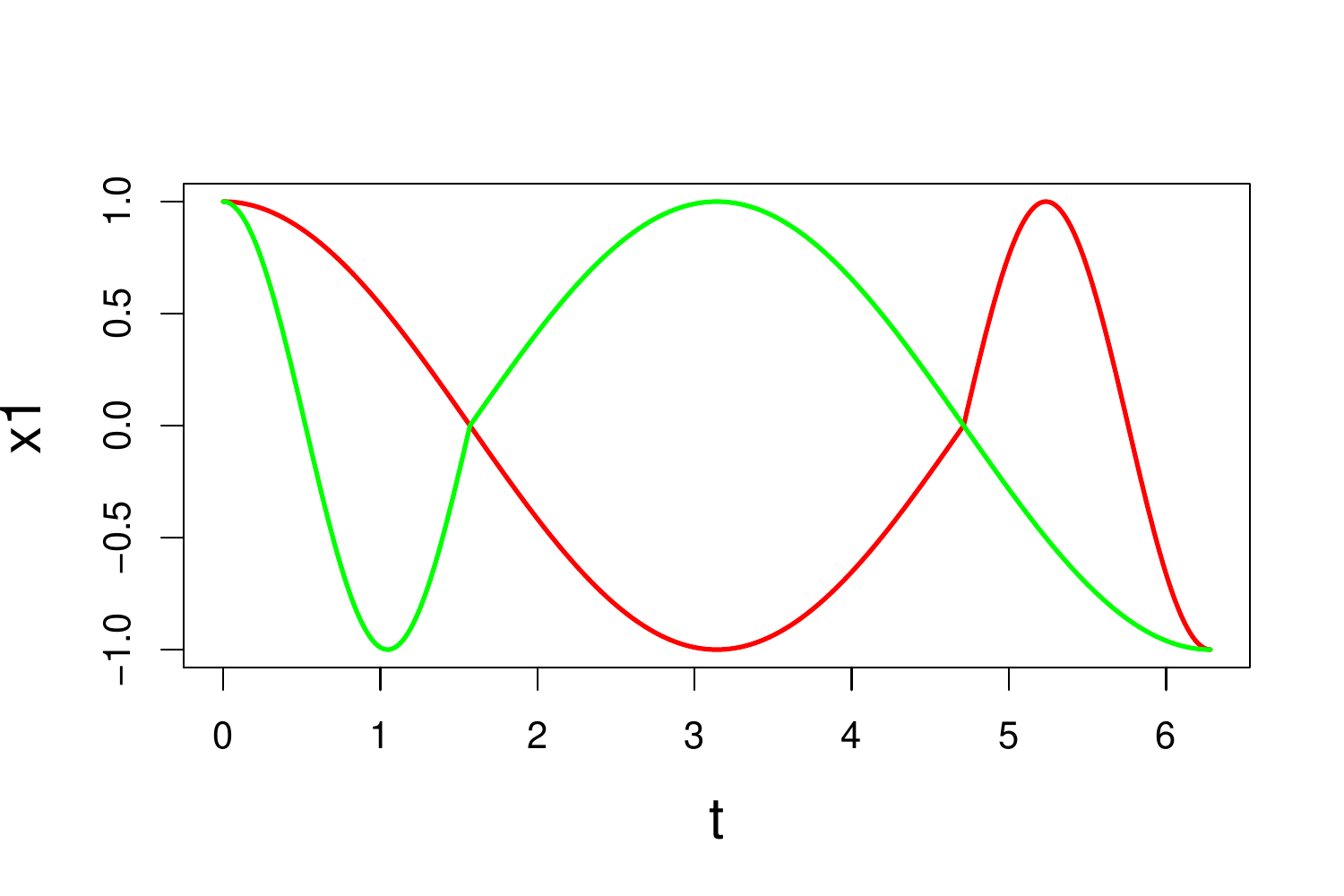}\includegraphics[keepaspectratio=true,scale=0.475,trim=0mm 0mm 0mm 10mm,clip=true]{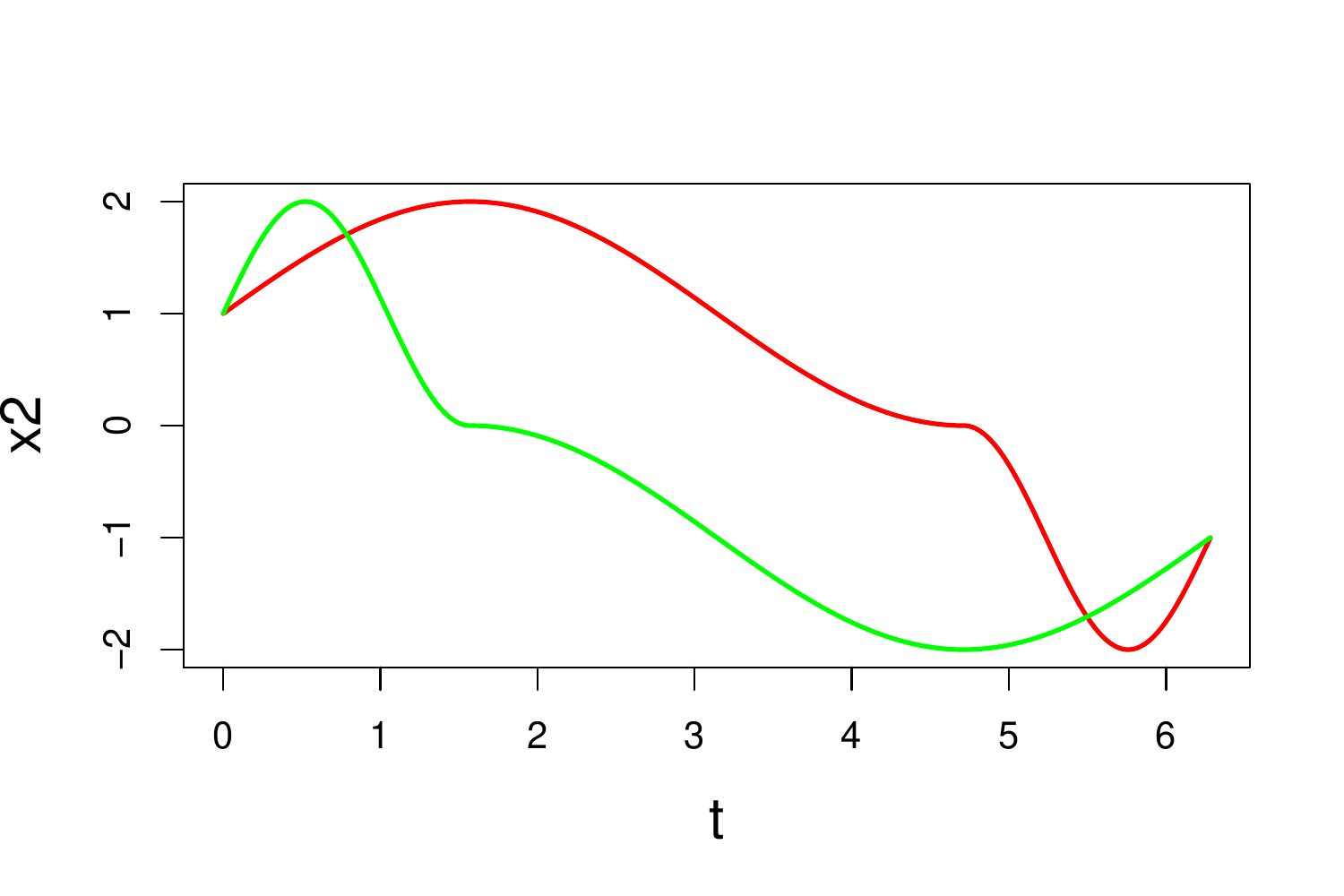}

	\hfill(a)\hfill\hfill(b)\hfill\hfill
	\caption{Parametrization A (in red) and parametrization B (in green) for coordinates $x_1$ (a) and $x_2$ (b).}
	\label{app:fig:parametrizations}
\end{center}
\end{figure}


A set of 50 S-shaped curves is then obtained by randomly shifting and rotating an ``ideal'' S-curve, as well as changing its length; see Figure~\ref{app:fig:curves3d} (a). More precisely, both location coordinates, the rotation angle, and the difference of length w.r.t. the beginning and the end of the ``ideal'' S-curve are drawn from a normal distribution centered at zero.


Depth-based rankings given by MFHD \citep{claeskens2014multivariate} and mSBD \citep{lopezpintado2014}, both using parametrizations A and B, are displayed in Figure~\ref{app:fig:curves3d} (b) and (c), respectively.  

\begin{figure}[H]
\begin{center}
	\includegraphics[keepaspectratio=true,scale=0.325,trim=0mm 0mm 0mm 10mm,clip=true]{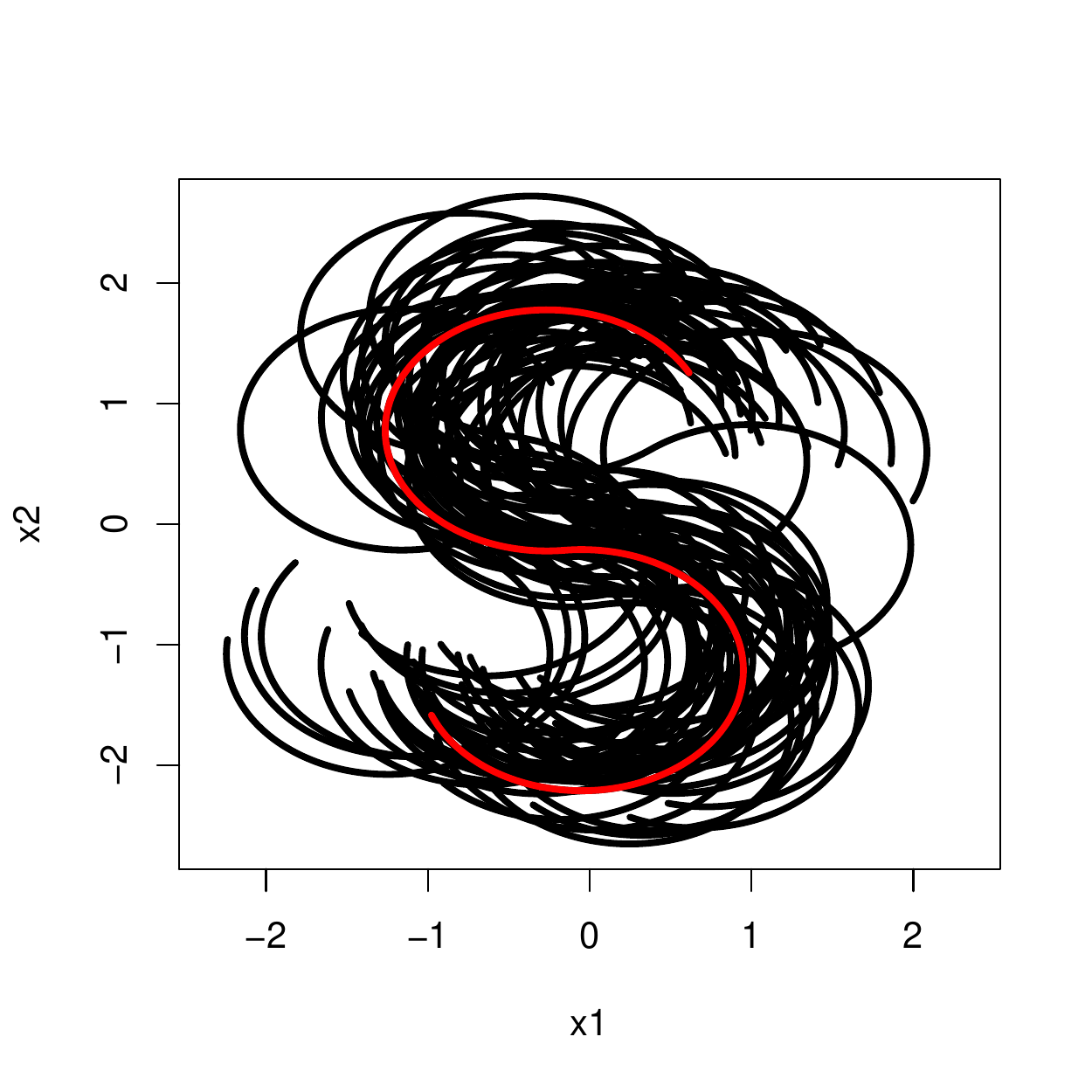}		    \includegraphics[keepaspectratio=true,scale=0.325,trim=0mm 0mm 0mm 10mm,clip=true]{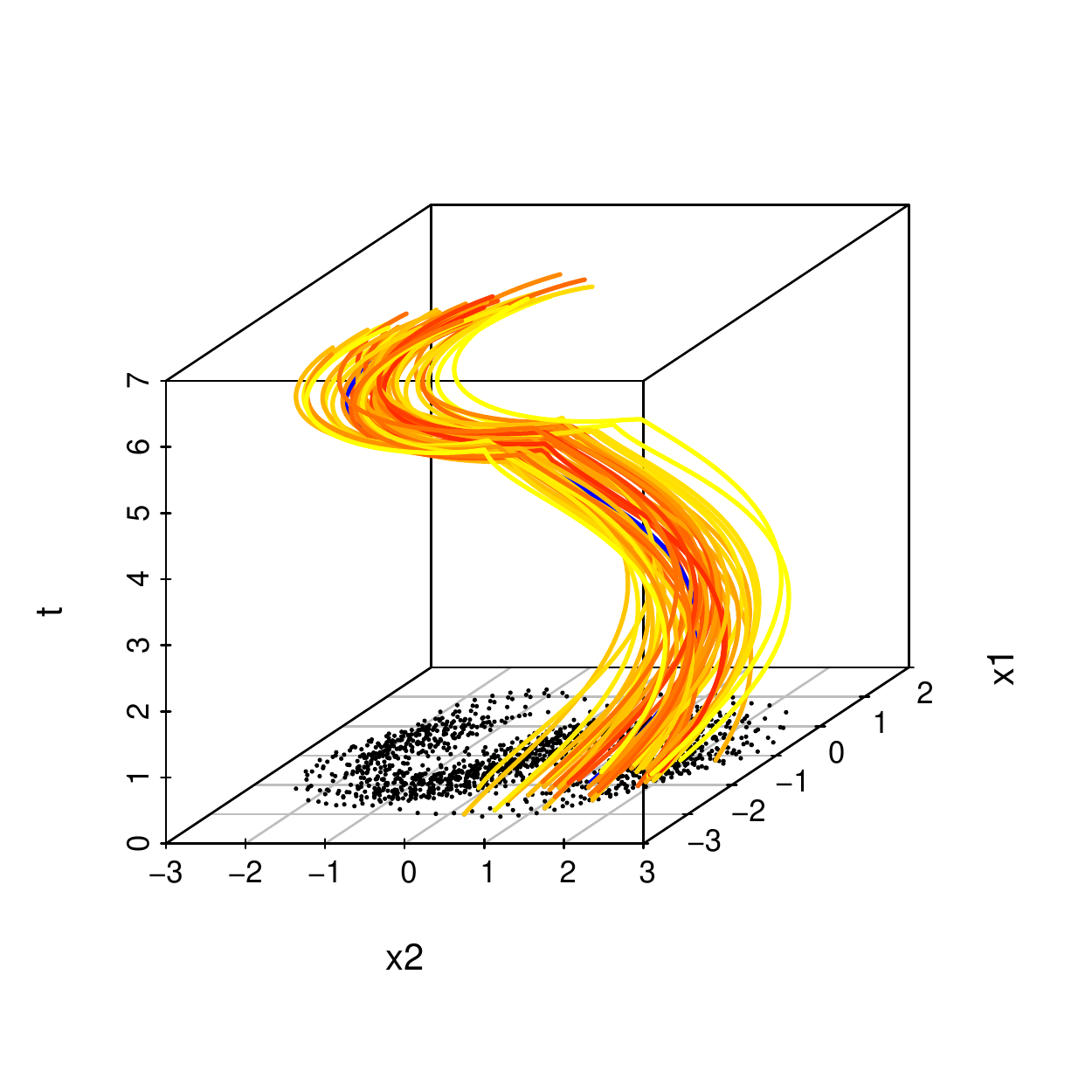}       	\includegraphics[keepaspectratio=true,scale=0.325,trim=0mm 0mm 0mm 10mm,clip=true]{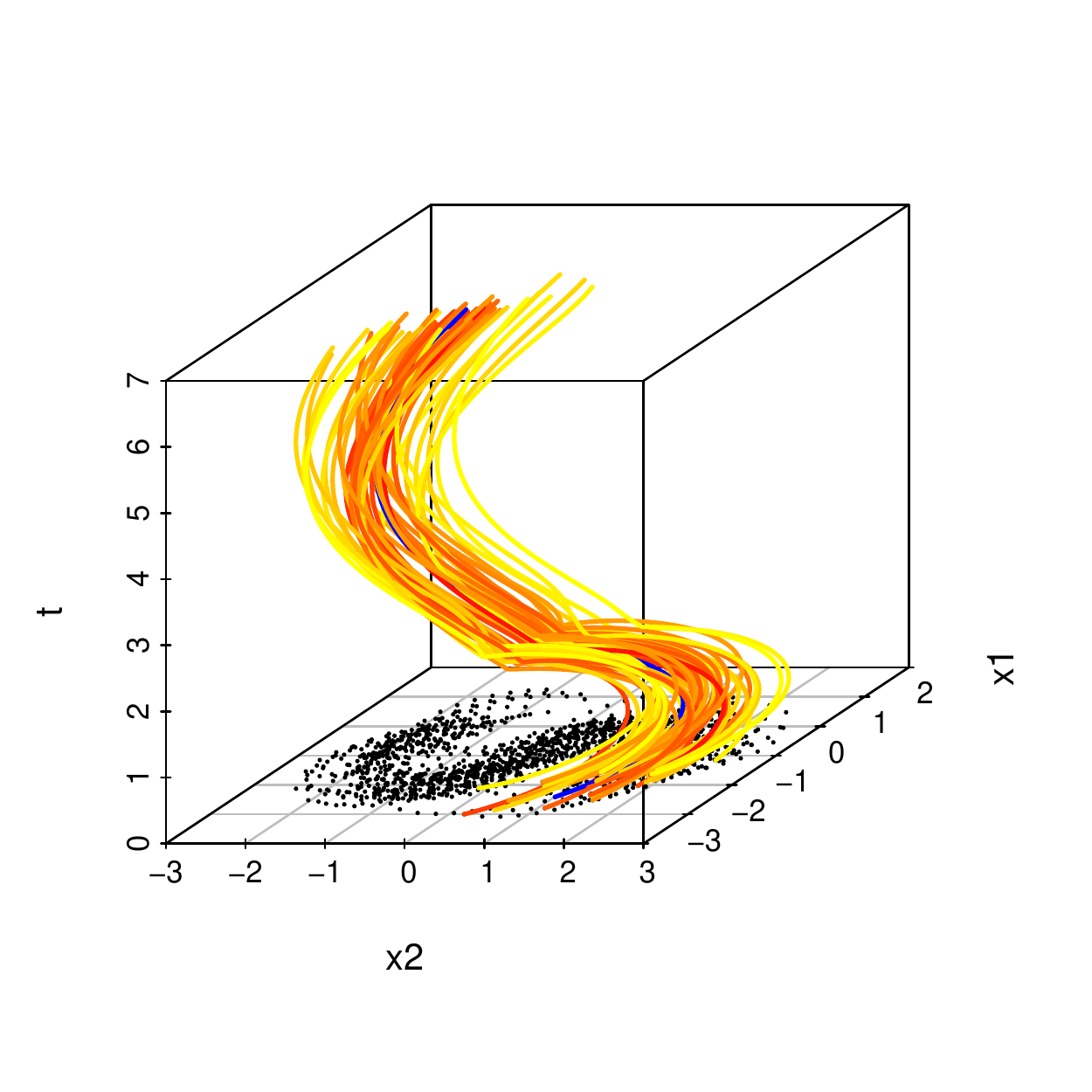}
	
	\hfill\hfill\quad(a)\qquad\hfill\qquad(b)\qquad\hfill\quad(c)\hfill\hfill\hfill
	\caption{A set of $50$ curves derived from the red pattern (a), together with their corresponding depth-colored functional representations for parametrizations A (b) and B (c). The depth of each curve is calculated w.r.t.\ to the same sample of $50$ curves. Here, we used the multivariate functional halfspace depth by \cite{claeskens2014multivariate}. The depth increases from yellow to red, the deepest curve being colored in blue.}
	\label{app:fig:curves3d}
\end{center}
\end{figure}


\subsection{Cursive Handwriting Sample}\label{cursiveexample}

We applied the multivariate functional halfspace depth developped by \cite{claeskens2014multivariate} (with weight function set to a constant) to a set of 20 planar curves taken from \citep[Cursive handwriting sample]{ramsay2017fda}. These curves were parameterized via two continuous functions $u\mapsto (x(u),y(u))\in\mathbb{R}^2$, where the parameter $u\in[0,1]$ represents either the time or the arc-length. (Note that an equivalent representation of such a curve, standard in multivariate functional data analysis, is through a vector of two real-valued functions defined over $[0, 1]$, as in the previous subsection.) Depth rankings are different depending on the parametrization chosen; see Figure~\ref{fig:fda3dbox}. 

\begin{figure}[H]
	\begin{center}
		\includegraphics[width=0.425\textwidth,trim = 0mm 20mm 15mm 25mm, clip,page=1]{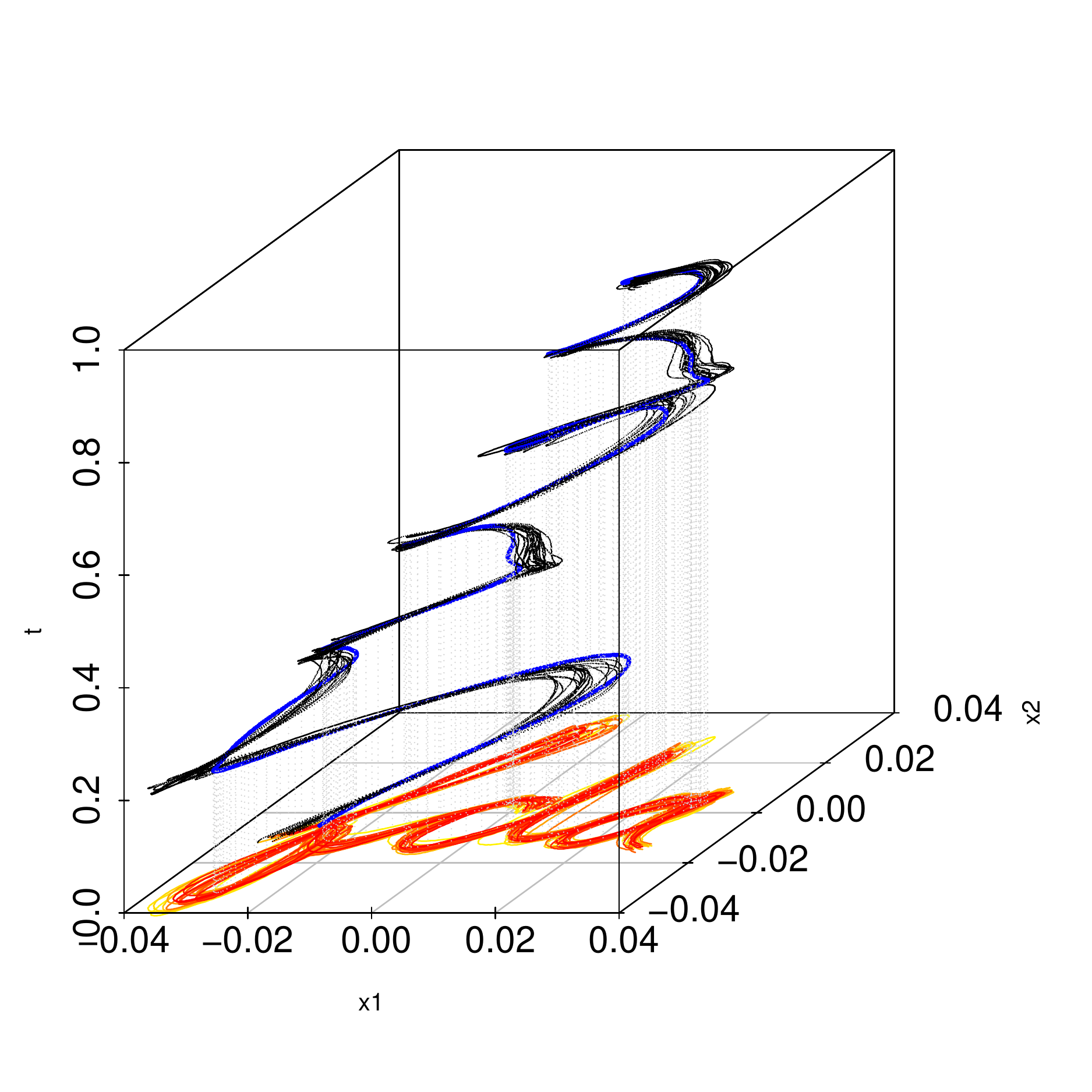}\quad\includegraphics[width=0.425\textwidth,trim = 0mm 20mm 15mm 25mm, clip,page=2]{pic-fda-3dbox}
	\end{center}
	\caption{Functional depth based ranking (obtained using the multivariate functional halfspace depth by \cite{claeskens2014multivariate}) of plane curves for parametrization by time (left) or by arc-length (right) obtained using the package \texttt{MFHD} \citep{MFHD}. The depth increases from yellow to red in the two dimensional trace. The 3-dimensional blue curve indicates the observation with the extreme rank difference for the two parametrizations (rank $1$ for time parametrization and rank $13$ for arc-length parametrization).
	Source: data set \texttt{handwrit} of the \texttt{R}-package \texttt{fda} \citep{ramsay2017fda}.  
}
	\label{fig:fda3dbox}
      \end{figure}

Only four curves out of twenty are assigned the same ranks (namely 3, 10, 19 and 20). For the sixteen others, depth-induced rankings are different, sometimes by a large amount. For instance, one curve is ranked~1 (minimal depth) for one parametrization and 13 (quite high depth) over 20 for the other; see Table~\ref{tab:rankings}.

\begin{table}[H]
	\caption{Depth ranks for parametrization by time or by arc-length.}
	\label{tab:rankings}
	\begin{center}
	{\scriptsize
          \begin{tabular}{l|cccccccccccccccccccc}
            \hline
			Time & 2 & 3 & 13 & 12 & 4 & 8 & 1 & 17 & 11 & 9 & 7 & 19 & 15 & 20 & 18 & 16 & 14 & 5 & 6 & 10 \\ \hline
			Length & 6 & 3 & 16 & 14 & 5 & 7 & 13 & 11 & 1 & 17 & 2 & 19 & 8 & 20 & 12 & 18 & 15 & 4 & 9 & 10 \\ \hline
		\end{tabular}
	}
	\end{center}
\end{table}

It thus appears that to obtain meaningful results, a proper parametrization of curves is needed. (This could be the speed of writing in this handwriting recognition example.)

\subsection{Historic Hurricanes Tracks}

We applied the multivariate functional halfspace depth developped by \cite{claeskens2014multivariate} and the multivariate simplicial depth developed by \cite{lopezpintado2014} to the historical hurricane tracks (obtained from \url{https://coast.noaa.gov/hurricanes/}) that go through the circular region of size 56 nautical miles centered at location $24.5$N by $78$W. We considered two parametrizations : the arc-length parametrization (A) and the parametrization by the time (B); see Figure \ref{fig:hurricanes:comparison}. In this example, the results we obtained seem  less sensitive to the choice of a parametrization than in the previous subsections. Nevertheless, this illustrates that multivariate functional depth functions tend to detect outliers which do not appear to be geometrically aberrant.

\begin{figure}[H]
\begin{center}
	\begin{tabular}{ccccc}
	{\small MFHD, par. A} & {\small MFHD, par. B} & {\small mSBD, par. A} & {\small mSBD, par. B} & {\small Curve Depth} \\
		\includegraphics[width=0.155\textwidth,page=1]{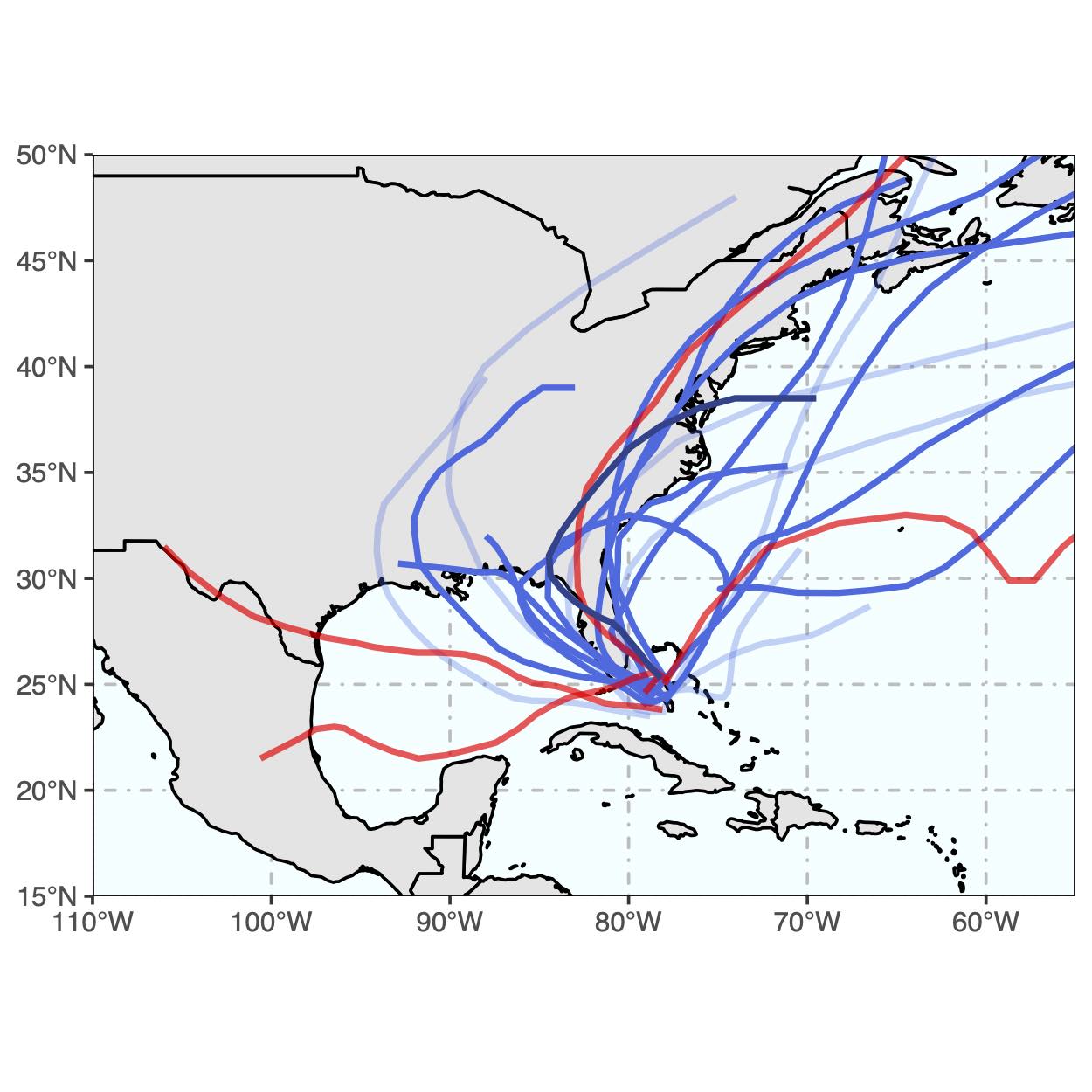}&\includegraphics[width=0.155\textwidth,page=2]{pic-sim_hurricanes-2Dbp_cf1.jpg}&\includegraphics[width=0.155\textwidth,page=3]{pic-sim_hurricanes-2Dbp_cf1.jpg}&\includegraphics[width=0.155\textwidth,page=4]{pic-sim_hurricanes-2Dbp_cf1.jpg}&\includegraphics[width=0.155\textwidth,page=5]{pic-sim_hurricanes-2Dbp_cf1.jpg}
	\end{tabular}
\end{center}
	\caption{Comparison of depth based ordering for two parametrizations A (time) and B (arc length) provided respectively by MFHD (a)--(b), mSBD (c)--(d), and by our new depth for unparameterized curves (e). The depth increases from yellow to red. Deepest curves are plotted in blue. Curves with low value of depth are plotted in red. Source: an ensemble of $23$ historic hurricane tracks originating in the Gulf of Mexico between 1918 - 2018.}\label{fig:hurricanes:comparison}
\end{figure}

\section{The Space of Unparametrized Curves}\label{app:curves}

Several of the results in this section can be found in \cite[Section~2]{Kemppainen2017}. The authors of this article borrowed material from \citet[Section~2.1]{aizenman1999holder} and \citet[Section~2.5]{Burago2001}.

\subsection{Equivalence Relation for Parametrized Curves}\label{app:curves:equiv}
We denote $\Gamma$ the set of increasing continuous functions $\gamma:[0,1]\to [0,1]$ such that $\gamma(0)=0$ and $\gamma(1)=1.$
Two parametrized curves $\beta_1:[0,1]\to\RR^d$ and $\beta_2:[0,1]\to\RR^d$ are equivalent (i.e., describe the same unparametrized curve) if and only if there exist two reparametrizations $\gamma_1,\gamma_2\in\Gamma$ such that $\beta_1\circ\gamma_1=\beta_2\circ\gamma_2.$

In order to describe the equivalence class associated to $\beta_1,$ we consider never-locally-constant functions.  
A parametrized curve $\beta:[0,1]\to\RR^d$ is said to be \emph{never-locally-constant} if there exists no non-empty sub-interval $(a,b)\subset[0,1]$ such that the restriction of $\beta$ to the interval $[a,b],$ denoted as $\beta_{\Large|_{[a,b]}},$ is a constant function. 
According to \citet[Exercice 2.5.3]{Burago2001}, each equivalence class admits one representative which is never-locally-constant, for example its arc-length parametrization.
The equivalence class associated to the never-locally-constant path $\beta$ in $\spacefun$ is,
$$
\Cc=\{\beta\circ\gamma:\,\gamma\in\Gamma\}.
$$
The set of unparmetrized curves $\Bg$ is the quotient space of $\spacefun$ by the equivalence relation defined above.

\subsection{The Metric Space of Unparmetrized Curves}\label{app:curves:metric}
Following \cite{Kemppainen2017}, we endow the space of curves $\Bg$ with the Fr\'{e}chet metric $d_\Bg$ defined as
\begin{equation}\label{eq:distance}
d_\Bg\left(\Cc_1,\Cc_2\right) = \inf\left\{\|\beta_1-\beta_2\|_\infty,\beta_1\in\Cc_1,\ \beta_2\in\Cc_2\right\},\quad\Cc_1,\Cc_2\in\Bg,
\end{equation}
where $\|\beta\|_{\infty} = \sup_{t\in[0,1]}|\beta(t)|_2$ for $\beta\in\spacefun$.

\begin{lemma}\label{lem:SpaceCurve:complete}
The metric space $(\Bg,d_\Bg)$ is separable and complete.
\end{lemma}

Proof of Lemma~\ref{lem:SpaceCurve:complete} relies on the following lemma.

\begin{lemma}\label{lem:SpaceCurve:never-local-constant}
Let $\beta_1$ and $\beta_2$ be two never-locally-constant paths on $[0,1]$. 
Let $\Cc_i$ be the unparametrized curve associated to $\beta_i$ and
$$
\Cc_i^{\mathrm{Hom}}=\{\beta_i\circ\psi;\quad\text{$\psi:[0,1]\to[0,1]$ is homeomorphic increasing continuous}\},
$$
a subset of $\Cc_i$, $i=1,2$. Then, we have
$$
d_\Bg\left(\Cc_1,\Cc_2\right) = d_\Bg\left({\Cc_1}^{\mathrm{Hom}},{\Cc_2}^{\mathrm{Hom}}\right).
$$
\end{lemma}

\begin{proof}[Proof of Lemma \ref{lem:SpaceCurve:never-local-constant}]
We note that for every reparametrization $\gamma\in\Gamma,$ there exists a sequence $(\psi_n)_n$ of increasing homeomorphisms that converges uniformly to $\gamma.$
Then using the uniform continuity of the parametrized curves, we deduce that every point of the equivalence class $\Cc_i$ is the uniform limit of sequence of ${\Cc_i}^{\mathrm{Hom}}$ for $i=1,2.$
\end{proof}

\begin{proof}[Proof of Lemma~\ref{lem:SpaceCurve:complete}]
  First, note that $(\Bg,d_\Bg)$ is a metric space \cite[Lemma~2.1]{aizenman1999holder}. It remains to prove that it is separable and complete.

\textit{1. The topological space $(\Bg,d_\Bg)$ is separable.}
The topological space $(\spacefun,\|\cdot\|_\infty)$ is separable \cite[Exemple~1.3]{billingsley2013convergence}, 
so by definition it contains a countable dense subset $\mathbb{D}$. 
Then the set of equivalence classes associated to the paths of $\mathbb{D}$ is a countable dense subset of $(\Bg,d_\Bg).$

\textit{2. The topological space $(\Bg,d_\Bg)$ is complete.}
Let $(\Cc_m)_m$ be a Cauchy sequence of $(\Bg,d_\Bg).$
Let $(\epsilon_k)_k$ be a sequence of positive real numbers such that the series of general term $(\epsilon_k)$ converges. 
Using Lemma \ref{lem:SpaceCurve:never-local-constant} it is possible to build a sub-sequence $(n_k)_k$ and a sequence of never-locally-constant parametrizations $\beta_{n_k}$ of $\Cc_{n_k}$ such that,
$$
\forall k\leq 1,\quad \|\beta_{n_k}-\beta_{n_{k+1}}\|_\infty\leq\epsilon_k.
$$
Then $(\beta_{n_k})$ is a Cauchy sequence of the complete space $(\spacefun,\|\cdot\|_\infty).$
There exists $\beta\in\spacefun$ such that $\lim_{k\to\infty} \|\beta_{n_k}-\beta\|_\infty=0.$
Since the sequence $(\Cc_n)_n$ is a Cauchy sequence and that $\beta_{n_k}$ is a parametrization of $\Cc_{n_k},$ we deduce that $(\Cc_n)_n$ converges to the equivalence class of $\beta$ in $(\Bg,d_\Bg).$
\end{proof}

\subsection{Mesurability of the Line Integral}\label{app:curves:integral}
The length of a parametrized curve $\beta,$ denoted $L(\beta),$ is defined as the supremum of the set of chordal lengths,
\begin{equation}\label{eq:chordal.length}
L_{\tau}(\beta)=\sum_{j=1}^J |\beta(\tau_j)-\beta(\tau_{j-1})|_2,
\end{equation}
corresponding to all finite partitions $\tau$ of $[0,1]$ : $0=\tau_0<\tau_1<\ldots <\tau_J=1.$
A parametrized curve $\beta$ is \emph{rectifiable} if $L(\beta)$ is finite.

\begin{rem}
For a rectifiable parametrized curve $\beta,$ we have,
$$
L(\beta) =\lim_{J\to \infty} \sum_{j=1}^J \left|\beta(j/J)-\beta\left(({j-1})/J\right)\right|_2.
$$
\end{rem}

The length is a property of the equivalence class : all parametrizations of $\Cc\in\Bg$ have the same length. We denote by $L(\Cc)$ the length of $\Cc.$

For a rectifiable parametrized curve $\beta,$ we define the length reparametrization \citep[see][Theorem 1.3]{vaisala2006lectures} :
\begin{align*}
s_\beta:[0,1]&\to[0,1]\\
		t &\mapsto L(\beta^t)/L(\beta),
\end{align*}
where $\beta^t$ is the restriction of $\beta$ to the interval $[0,t].$
The function $s_\beta$ is increasing and continuous, that is $s_\beta\in\Gamma.$
Moreover, one can define the generalized inverse of $s_\beta,$ 
\begin{align*}
q_\beta:[0,1]&\to[0,1]\\
		u &\mapsto \inf\{t: s_\beta(t)\geq u\}.
\end{align*}
The function $q_\beta$ is left continuous and admits a limit from the right for all $u\in[0,1]$ \citep[see][Proposition 1]{embrechts2013note}.
According to \citet[Theorem 2.4]{vaisala2006lectures}, for each rectifiable curve $\Cc$ there exists a unique parametrization $\beta_\Cc:[0,1]\to\RR^d,$ called \emph{the arc-length parametrization}, such that $L(\beta_\Cc^t)=tL(\Cc),$ for all $t\in[0,1].$ The arc-length parametrization is never-locally-constant.

\begin{lemma}\label{lem:SpaceCurve:mesureQ}
Let $\Bg_L$ be the set of rectifiable unparametrized curves with a positive length.
\begin{enumerate}
\item $L:\Bg\to[0,+\infty]$ is measurable and $\Bg_L$ is a measurable set.
\item Let $u\in[0,1]$ be fixed. The application, 
\begin{align*}
\{\beta\in\spacefun:\ L(\beta)<\infty\} &\to[0,1]\\
		\beta &\mapsto q_\beta(u)
\end{align*}
is measurable.
\item For all non negative bounded functions $f:\RR^d\to\RR,$ the application
\begin{align*}
I:\Bg_L &\to\RR\\
  \Cc &\mapsto \int_\Cc f(s) ds :=\int_0^1 f(\beta_\Cc(t)) dt
\end{align*}
is measurable.
\end{enumerate}
\end{lemma}

\begin{proof}[Proof of Lemma~\ref{lem:SpaceCurve:mesureQ}]
\textit{1. $L$ is measurable.}
Let $\tau$ be a partition of $[0,1].$ 
The function $\beta\in\spacefun\mapsto L_{\tau}(\beta)\in\RR_+$ is measurable.
Then the length function $\beta\mapsto L(\beta)$ is measurable as the limit of measurable functions.
Moreover, for $A$ a borelian of $[0,\infty],$ we have :
$$ 
\beta\in L^{-1}(A)\ \text{ if and only if }\ \Cc_\beta \subset L^{-1}(A).  
$$
Then $L:\Bg\to [0,+\infty]$ is measurable and $\Bg_L=L^{-1}(]0,\infty[).$

\textit{2. $\beta\mapsto q_\beta(u)$ is measurable.}
From the previous item, we deduce that for all $t\in[0,1],$ the function
\begin{align*}
\{\beta\in\spacefun:\ L(\beta)<\infty\} &\to[0,1]\\
		\beta &\mapsto s_\beta(t)
\end{align*}
is measurable. Let $u$ be in $[0,1]$ fixed. We remark that,
$$\small
\{\beta\!\in\!\spacefun: 0<L(\beta)<\infty\ \text{and}\ q_\beta(u) \leq t \} = \{\beta\!\in\!\spacefun: 0<L(\beta)<\infty\ \text{and}\ s_\beta(t) \geq u \}. 
$$
Then $\beta\mapsto q_\beta(u)$ is measurable too.

\textit{3. $I:\Cc\mapsto I(\Cc)$ is measurable.}
It suffices to prove the lemma when $f$ is continuous.
Let $\Cc$ be in $\Bg_L.$
Using Riemann sums we have:
\begin{align*}
I(f) = \lim_{n\to\infty} I_n(\Cc),\quad\text{where } I_n(\Cc)=\frac{1}{n} \sum_{i=1}^n f\left(\beta_\Cc(i/n)\right).
\end{align*}
Let $\beta$ be a parametrization of $\Cc,$ then $\beta=\beta_\Cc\circ s_\beta$ and $\beta_\Cc=\beta\circ q_\beta.$
Then we can rewrite $I_n(\Cc)$ as
\begin{align*}
I_n(\beta)=\frac{1}{n} \sum_{i=1}^n f\left(\beta(q_\beta(i/n))\right).
\end{align*}
We deduce from the previous point that the function,
\begin{align*}
\{\beta\!\in\!\spacefun: 0<L(\beta)<\infty\}\to \RR,\quad \beta\mapsto I_n(\beta),
\end{align*}
is measurable, and its limits is measurable too.
\end{proof}

Further, we define the probability distribution $\mu_{\Cc}$ on the Borel sets of $\RR^d$
$$
\text{for all borel sets $A$ of $\RR^d,$}\qquad \mu_\Cc(A) = \frac{1}{L(\Cc)}\int_\Cc\mathds{1}_A(s) ds\,,
$$
with $\mathds{1}_A(x)$ being the indicator function that takes the value 1 if $x\in A$ and 0 otherwise.

\begin{lemma}\label{lem:SpaceCurve:approxmu}
Let $\beta$ be a parametrisation of $\Cc.$ 
Let $U_J$ be a random variable such that its distribution is a mixture distribution,
$$
\sum_{j=1}^J \frac{|\beta(j/J)-\beta(({j-1})/J)|_2}{L_{\tau_J}(\beta)}\ \Uu_{[(j-1)/J,{j}/J]},
$$
where $L_{\tau_J}(\beta)$ is the chordal-length \eqref{eq:chordal.length} associated to the partition $\tau_J=(j/J)_{j=0,\ldots,J}$ and  $\Uu_{[a,b]}$ is the uniform distribution on the interval $[a,b].$
The sequence of random variables $(\beta(U_J))_{J\geq 1}$ converges in distribution to $\mu_\Cc.$
\end{lemma}

\begin{proof}[Proof of Lemma~\ref{lem:SpaceCurve:approxmu}]
Let $\ell_\beta(t)=L(\beta_{|[0,t]})$ be the length of the parametrized curve $\beta$ on $[0,t].$
Let $f:\RR^d\to\RR$ be a continuous bounded function.
Using that the functions $f,\beta,\beta_\Cc$ and $\ell_\beta$ are continuous (and uniformly continuous on a compact set), for all $\epsilon>0$ there exists $J_\epsilon\geq 1$ such that,
\begin{align*}
&\forall J\geq J_\epsilon,\ [L(\Cc)-L_{T_J}(\beta)|<\epsilon L(\beta)/2\quad\text{and}\quad L(\Cc)/2<L_{T_J}(\beta),\\
&\forall J\geq J_\epsilon,\ \forall t\in[0,1], \exists j\in\{1,\ldots,J\} : t\in[(j-1)/J,j/J]\ \text{and } |f\circ\beta(j/J)-f\circ\beta(t)| <\epsilon.
\end{align*}
Then we can show that
\begin{align*}
\lim_{J\to\infty} \mathbb{E}_P(f(\beta(U_J)) &= \lim_{J\to\infty} \sum_{j=1}^J \frac{|\beta(j/J)-\beta(({j-1})/J)|_2}{L(\Cc)} f(\beta(j/J)),\\
\int f d\mu_\Cc &= \lim_{J\to\infty} \sum_{j=1}^J \frac{\ell_\beta(j/J)-\ell_\beta(({j-1})/J)}{L(\Cc)} f(\beta(j/J)).
\end{align*}
Noticing that for all $j=1,\ldots,J,$ $\ell_\beta(j/J)-\ell_\beta(({j-1})/J)\geq |\beta(j/J)-\beta(({j-1})/J)|_2,$ we can bound the difference,
\begin{align*}
&\left|\sum_{j=1}^J \frac{\ell_\beta(j/J)-\ell_\beta(({j-1})/J)}{L(\Cc)} f(\beta(j/J)) - \sum_{j=1}^J \frac{|\beta(j/J)-\beta(({j-1})/J)|_2}{L(\Cc)} f(\beta(j/J))\right|\\
&\qquad\qquad\qquad\qquad\qquad\leq \frac{\max_{t\in[0,1]}|f(\beta(t))|}{L(\Cc)} \left(L(\Cc)-L_{\tau_J}(\beta)\right).
\end{align*}
\end{proof}

\begin{rem}[The order does not matter]\label{rem:order}
In the paper, we define an unparametrized curve $\Cc=\Cc_\beta$ with an order : the starting point is $\beta(0)$ and the end point is $\beta(1)$.
We can also define an unparametrized curve without orientation via an equivalence relation on the set of parametrized curves up to a larger set of reparametrization,
\begin{multline*}
\Gamma_{\textrm{new}}=\left\{\gamma:[0,1]\to[0,1] : \text{$\gamma$ is continuous and monotonic, }\right.\\ \left.
 (\gamma(0),\gamma(1))\in\{(0,1),(1,0)\}\right\}.
\end{multline*}
Using the same arguments, one can show that the resulting space of curves endowed with the associated Frechet metric is separable and complete. 
The difference is that there exist two arc-length parametrizations : one $\beta_{\Cc}^+$ from $\beta(0)$ to $\beta(1)$ and the other $\beta_{\Cc}^-$ from $\beta(1)$ to $\beta(0)$ such that,
$$
\beta_{\Cc}^-(t)=\beta_{\Cc}^+(1-t).
$$
Then the definition of the probability measure $\mu_\Cc$ is invariant whether we use $\beta_{\Cc}^+$ or $\beta_{\Cc}^-$ to define it.
\end{rem}


\section{Definition of the Depth Functions}

In order to prove that our curve depth is well defined, we have to show that the function $x\mapsto D(x|Q_P,\mu_\Cc)$ is measurable and that $D(\Cc|Q_P)$ is bounded by $1$ for all $\Cc\in\Bg_L$. This requires the use of a Monte Carlo scheme that we describe in the next subsection.

\subsection{Monte Carlo Approximation of the Curve Depth}\label{app:ssec:montecarlo}

Notice that the curves $\Xx_1,\ldots \Xx_n$ and $\Cc$ are known, so this means that $\mu_\Cc$ and $Q_n$ are formally known too. However, the computation of $\mu_\Cc(H)$ and $Q_n(H)$ for an arbitrary halfspace $H$ can be untractable. Consequently, it is necessary to estimate $\mu_\Cc(H)$ using either a quadrature formula or a Monte Carlo approach. We choose here a Monte Carlo approximation of \eqref{eq:depth:emp:curve}. In Section \ref{app:sec:simulation} we conduct a simulation study on the size $m$ of the Monte Carlo scheme.

We generate samples of size $m$ from the observed (realized) curves $\Xx_1,\ldots,\Xx_n$:  
\begin{align*}
\text{for all $i=1,\ldots,n$},\quad 
\text{ given $\Xx_i,$}\quad X_{i,1},\ldots,X_{i,m} \text{ are i.i.d.\ distributed from } \mu_{\Xx_{i}}^{},
\end{align*}
and two independent samples from the curve $\Cc,$
\begin{align*}
&Y_{1}, \ldots,  Y_{m} \text{ are i.i.d.\ distributed from } \mu_{\Cc}^{},\\
&Z_{1}, \ldots,  Z_{m} \text{ are i.i.d.\ distributed from } \mu_{\Cc}^{}.
\end{align*}

We use $\YY_{m}=\{Y_{1}, \ldots , Y_{m}\}$ to estimate the distribution $\mu_\Cc$ by the empirical distribution $\widehat{\mu}_{m}$:
$$
\widehat{\mu}_{m} = m^{-1}\sum_{j=1}^{m}\delta_{Y_{i}}\,,
$$
where $\delta_x$ stands for the Dirac measure at $x\in\RR^d.$ 

Furthermore, we remark that the marginal distribution of $X_{i,j}$ is $Q_P$ (see Remark \ref{rem:Q}). 
Then, let $\widehat{Q}_{m,n}$ be the empirical distribution of the random sample $\XX_{n,m}=\{X_{i,j},\ i=1,\ldots, n;\ j=1,\ldots, m\}$:
$$
\widehat{Q}_{m,n}=(mn)^{-1}\sum_{ i=1}^n\sum_{j=1}^m\delta_{X_{i,j}}.
$$
Let $H$ be a closed halfspace of $\RR^d.$ 
A plug-in estimator of $Q_n(H)/\mu_\Cc(H)$ is $\widehat{Q}_{m,n}(H)/\widehat{\mu}_{m}(H).$ 

To ensure the consistency of the Monte Carlo estimate of \eqref{eq:depth:emp:curve} we need to control the ratio $\widehat{Q}_{m,n}(H)/\widehat\mu_m(H)$ for all $H$ such that $\mu_\Cc(H)>0,$ given that $\mu_\Cc(H)$ is approximated by $\widehat{\mu}_m(H)$. 
This is known to be a challenging problem having no general solution, see e.g., \cite{BrodaK16}. 
To circumvent this, we consider only a subset of all halfspaces in $\RR^d$ for the computation of the Monte Carlo estimate of the depth. 
Let $\Delta$ be in $(0,1/2).$
We denote by $\HH_{\Delta}^{n,m}$ the collection of closed halfspaces $H$ such that either $\widehat{Q}_{m,n}(H)=0$ or $\widehat{\mu}_{m}(H)>\Delta$, almost surely. 
For all $x$ in the locus of $\Cc$, we define 
{\begin{align}
\widehat{D}(x|\widehat{Q}_{m,n},\widehat{\mu}_{m},\HH_{\Delta}^{n,m})&\!=\! \inf_{u\in\mathcal{S}}\left\{\widehat{Q}_{m,n}(H_{u,x})/\widehat{\mu}_{m}(H_{u,x})\,:\ H_{u,x}\in\HH_{\Delta}^{n,m}\right\}\label{eq:depth:mc:x}.
\end{align}}
Then, we use $\ZZ_{m}=\{Z_{1}, \ldots ,  Z_{m}\}$ to estimate the integral \eqref{eq:depth:emp:curve} w.r.t.\ the probability measure $\mu_\Cc$,
\begin{align}
\widehat{D}_{n,m,\Delta}(\Cc|\Xx_1,\ldots,\Xx_n)&\!=\! \frac{1}{m}\sum_{i=1}^{m} \widehat{D}(Z_{i}|\widehat{Q}_{m,n},\widehat{\mu}_{m},\HH_{\Delta}^{n,m}).\label{eq:depth:mc:curve}
\end{align}

\begin{figure}[H]
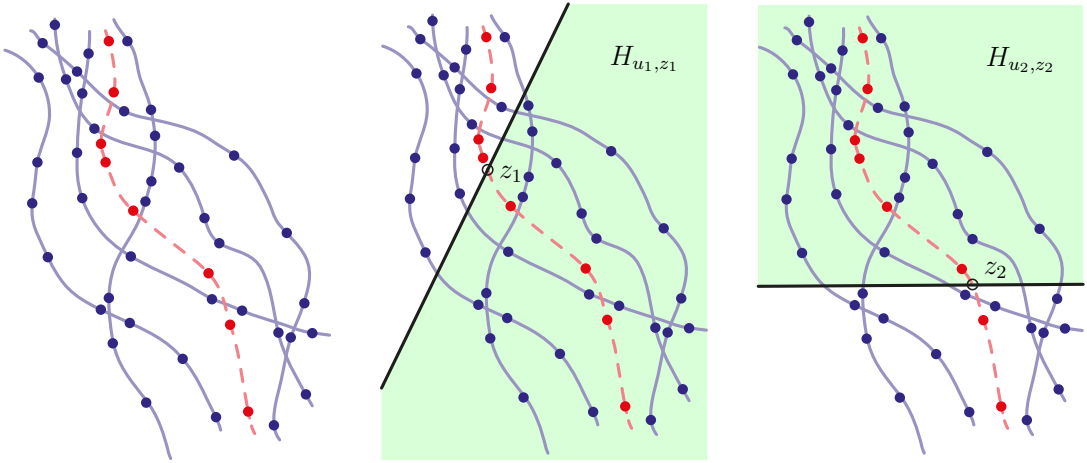

	\begin{center}
		\resizebox{0.29\linewidth}{!}{    
		\input{pic_approx01.tex}
		}
		\resizebox{0.29\linewidth}{!}{    
		\input{pic_approx02.tex}
		}
    	\resizebox{0.29\linewidth}{!}{    
		\input{pic_approx03.tex}
		}
	\end{center}
	\caption{
          Illustrations to the statistical model and Monte Carlo appproximation of the depth with a sample of five curves in blue and the curve $\Cc$ in red: (left) samples $\XX_{n,m}$ (blue) and $\YY_m$ (red) of points on the observed curves; (middle and right) illustration of calculations of \eqref{eq:depth:mc:x} for $z_1$ and $z_2$ on the red curve see Remark \ref{rk:calculation}.}
	\label{fig:illustration}
\end{figure}

\begin{rem}\label{rk:calculation}
To provide an intuitive reasoning, we artificially restrict the choice of the infimum in \eqref{eq:depth:mc:x} to two halfspaces where the number of observed curves is $n=5$ and the size of the Monte Carlo sample for each curve is $m=8$; see Figure~\ref{fig:illustration}. 
Let $z_1$ and $z_2$ be two points in the locus of $\Cc$ (red middle curve). 
Consider two halfplanes, say $H_{u_1,z_1}$ and $H_{-u_1,z_1}$,  yielded by the line in Figure~\ref{fig:illustration}, middle, when calculating $\widehat{D}(z_1|\widehat{Q}_{m,n},\widehat{\mu}_{m})$. 
For each of these halfplanes, we obtain $(\widehat{Q}_{m,n}(H_{u_1,Z_1}) = 25/40,\,\widehat{\mu}_{m}(H_{u_1,z_1}) = 4/8)$ and $(\widehat{Q}_{m,n}(H_{-u_1,z_1}) = 15/40,\,\widehat{\mu}_{m}(H_{-u_1,z_1}) = 4/8)$, respectively. 
Among $H_{u_1,z_1}$ and $H_{-u_1,z_1}$, $H_{-u_1,z_1}$ will be chosen as $\widehat{Q}_{m,n}(H_{-u_1,z_1}) < \widehat{Q}_{m,n}(H_{u_1,z_1})$ and $\widehat{\mu}_{m}(H_{u_1,z_1}) = \widehat{\mu}_{m}(H_{-u_1,z_1})$, and thus the rationale follows the traditional multivariate Tukey depth as this would be the case in the absence of the denominator $\widehat{\mu}_{m}(H_{\cdot,z_1})$. 
On the other hand, in Figure~\ref{fig:illustration}, right, the values of the denominators in \eqref{eq:depth:mc:x} differ giving pairs of portions equal to $(\widehat{Q}_{m,n}(H_{u_2,z_2}) = 25/40,\,\widehat{\mu}_{m}(H_{u_2,z_2}) = 6/8)$ and $(\widehat{Q}_{m,n}(H_{-u_2,z_2}) = 15/40,\,\widehat{\mu}_{m}(H_{-u_2,z_2}) = 2/8)$. In this case, halfplane $H_{u_2,z_2}$ with higher portion of $\widehat{Q}_{m,n}$ will be chosen due to the difference of $\widehat{\mu}_{m}(H_{u_2,z_2})$ and $\widehat{\mu}_{m}(H_{-u_2,z_2})$.
\end{rem}

Theorem~\ref{thm:consistency0} below states that the Monte Carlo approximation of the \curvedepth~\eqref{eq:depth:mc:curve} converges in probability to the population version \eqref{eq:depth:curve} when $n,m\rightarrow\infty$ or to the sample version \eqref{eq:depth:emp:curve} when $m\rightarrow\infty$.
Then Theorem~\ref{thm:consistency} holds. 
Let $\mu$ be a probability measure defined on $\RR^d,$ and let $\widehat{\mu}_m$ be the empirical measure defined on a $m$-sample of $\mu.$
We denote by $\HH$ the collection of all halfspaces in $\RR^d$ and define
$$
\|\widehat{\mu}_m-\mu\|_{\HH}:=\sup_{H\in\HH}|\widehat{\mu}_m(H)-\mu(H)|.
$$
According to \citet[Chapter 26, Theorem 1]{shorack2009empirical}, the class $\HH$ satisfies the Glivenko-Cantelly property. 
Then $\|\widehat{\mu}_m-\mu\|_\HH$ converge a.s. to zero as $m\to +\infty.$
Moreover, we have
\begin{equation}\label{eq:shorack}
\displaystyle\overline{\lim}_{m} \lambda_m^{-1} \|\widehat{\mu}_m-\mu\|_\HH \leq C\quad\text{a.s.},
\end{equation}
where $C=\sqrt{d+1+1/2}$ and $\lambda_m=(\log(m)/m)^{1/2}$; see \citet[Chapter 26, Exercise 2]{shorack2009empirical}.

\begin{thm}\label{thm:consistency0}
Let $\Cc\in\Bg_L$ be an unparametrized curve such that $\mu_\Cc$ is non-atomic. 
Let $P$ be a probability measure in the space of curves such that $P\in\Pp$ and $Q_P$ is non-atomic.
Let $(\Delta_m)$ be a decreasing sequence of positive numbers such that $(\Delta_m)$ and $(\lambda_m/\Delta_m^2)$ converge to zero as $m\to\infty.$
Then:
\begin{itemize}
\item the Monte Carlo approximation $\widehat{D}_{n,m,\Delta_m}(\Cc|\Xx_1,\ldots,\Xx_n)$ converges in probability to $D(\Cc|\Xx_1,\ldots,\Xx_n)$ as $m\to\infty,$
\item the Monte Carlo approximation $\widehat{D}_{n,m,\Delta_m}(\Cc|\Xx_1,\ldots,\Xx_n)$ converges in probability to $D(\Cc|P)$ as $m,n\to\infty,$
\item the sample Tukey curve depth $D(\Cc|\Xx_1,\ldots,\Xx_n)$ converges in probability to $D(\Cc|P)$ as $n\to\infty.$
\end{itemize}
\end{thm}

To prove the boundness of the \curvedepth, it suffices to show that for all $x,$ $D(x|Q_P,\mu_\Cc)$ is bounded by $1.$

\subsection{Boundness of Data Depth}

Let $A$ be a Borel set of $\RR^d,$ we denote by $\partial A$ the boundary of the set $A.$
For instance the boundary of $H_{x,u}$ is $\partial H_{x,u}=\left\{y\in\RR^d:\ (y-x)^\top u= 0\right\},$ $u\in\Ss,$ $x\in\RR^d.$

\begin{lemma}\label{lem:depth:bounded}
{Let $\mu$ and $Q$ be non-atomic measures on $\RR^d,$ $d\geq 1.$ 
For all $x\in \RR^d$, there exists a closed halfspace $H\in\HH$ such that $x\in \partial H$ and $Q(H)/\mu(H)\leq 1$ (with a convention $0/0=0$).}
\end{lemma}

\begin{proof}[Proof of Lemma \ref{lem:depth:bounded}]
Let $x$ be a fixed point of ${\RR^d}$.

\textit{1. The one-dimensional case, $d=1$.}
Since $Q_P$ and $\mu_\Cc$ are non-atomic measures, we get either $Q_P((-\infty,x])\leq \mu_\Cc((-\infty,x])$ or $Q_P([x,+\infty))\leq \mu_\Cc([x,+\infty)),$ for all $x\in\RR$. Then the lemma is proved.

\textit{2. The multi-dimensional case, $d\geq 2$.} If there exists $u\in\Ss$ such that $Q(\partial H_{u,x})=\mu(\partial H_{u,x})=0,$ then the lemma is proved. 
We show recursively that, for $k=1,\ldots,d-1,$ there exists an affine subspace $A_k$ of dimension $k$ such that $Q(A_k)=\mu(A_k)=0$ and $x\in A_k$. Then we consider $u\in\Ss$ such that $\partial H_{u,x}=A_{d-1}$ and the first assertion is true.

\textit{For $k=1,$} let $\mathcal{A}_n$ be the set of affine subspaces of dimension $1$ such that for all $A\in\mathcal{A}_n$: $x\in A$ and either $Q(A)>1/n$ or $\mu_\Cc(A)>1/n.$
The set $\mathcal{A}_n$ is finite since the intersection of $A\in\mathcal{A}_n$ is the singleton $\{x\}$ (and $Q(\{x\})=\mu(\{x\})=0$),
$$
1\geq \mu\left(\cup_{A\in\mathcal{A}_n}A\right) = \sum_{A\in\mathcal{A}_n}\mu(A)> \#\mathcal{A}_n /n.
$$
Then the set $$\cup_{n\geq 1}\mathcal{A}_n=\{A \text{ affine subspaces of dimension $1$}:\ x\in A\ \text{and}\ \mu(A)>0 \text{ or } Q(A)>0\}$$ is countable as the countable union of finite sets. Since the set of affine subspaces of dimension $1$ which contain $x$ is continuous, there exists $A_1\notin\mathcal{A}_n$ such that $Q(A_1)=\mu_\Cc(A_1)=0$ and $x\in A_1$.

\textit{Assume that $k\geq 2.$}
Using the recursive assumption for $k-1$, there exists an affine subspace $A_{k-1}$ of dimension $k-1$ such that $Q(A_{k-1})=\mu(A_{k-1})=0$ and $x\in A_{k-1}$.  
Let $\mathcal{A}_n$ be the set of affine subspaces $A$  of dimension $k$ such that $A_{k-1}\subset A,$ $Q(A)>1/n$ or $\mu(A)>1/n$. 
Using the same previous argument the subset $\mathcal{A}_n$ is finite and there exists $A_k\notin\cup_{n\geq 1}\mathcal{A}_n$.     
\end{proof}

\subsection{Mesurability}
To prove the measurability, it suffices to show that $x\mapsto D(x|Q_P,\mu_\Cc)$ and $x\mapsto D(x|Q_n,\mu_\Cc)$ are the limits of measurable  functions (see Lemma \ref{lem:depth:mc:conv-ps} and Lemma \ref{lem:depth:mc:conv-ps2} respectively).

\begin{lemma}\label{lem:depth:mc:measurable}
The function $x\in\RR^d\mapsto \widehat{D}(x|\widehat{Q}_{m,n},\widehat{\mu}_{m},\HH_{\Delta}^{n,m})\in\RR_+$ is measurable a.s.
\end{lemma}

\begin{lemma}\label{lem:depth:mc:bounded}
Let $\Cc\in\Bg_L$ be an unparametrized curve such that $\mu_\Cc$ is non-atomic.  
The function $x\mapsto \widehat{D}(x|\widehat{Q}_{m,n},\widehat{\mu}_{m},\HH_{\Delta}^{n,m})\in\RR_+$ is bounded $\mu_\Cc$-a.s.\ by $1$ if $\Delta\leq 1/m$ and by $(1-\Delta)^{-1}$ otherwise.
\end{lemma}

\begin{lemma}\label{lem:depth:mc:conv-ps}
Let $\Cc\in\Bg_L$ be an unparametrized curve such that $\mu_\Cc$ is non-atomic. 
Let $P$ be a probability measure on the space of curves such that $P\in\Pp$ and $Q_P$ is non-atomic.  
Let $(\Delta_m)$ be a decreasing sequence of positive numbers such that $(\lambda_m/\Delta_m^2)$ and $(\Delta_m)$ converges to zero as $m\to\infty.$
Then for all $x\in\RR^d,$ $\widehat{D}(x|\widehat{Q}_{m,n},\widehat{\mu}_{m},\HH_{m,n})$ converges almost surely to $D(x|Q_P,\mu_\Cc)$ as $n,m\to\infty.$
\end{lemma}

\begin{lemma}\label{lem:depth:mc:conv-ps2}
Let $\Cc\in\Bg_L$ be an  unparametrized curve such that $\mu_\Cc$ is non-atomic. 
Let $P$ be a probability measure in the space of curves such that $P\in\Pp$ and $Q_P$ is non-atomic.  
Let $(\Delta_m)$ be a decreasing sequence of positive numbers such that $(\lambda_m/\Delta_m^2)$ and $(\Delta_m)$ converges to zero as $m\to\infty.$
Then for all $x\in\RR^d,$ $\widehat{D}(x|\widehat{Q}_{m,n},\widehat{\mu}_{m},\HH_{m,n})$ converges almost surely to $D(x|Q_n,\mu_\Cc)$ as $m\to\infty.$
\end{lemma}

In what follows, we introduce the probability space $(\Omega,\mathcal{F},\mathbb{Q})$ generated by the sequences $(X_{i,j})_{j\geq 1,i\geq 1},$ $(Y_j)_{j\geq 1}$ and $(Z_j)_{j\geq 1}$ ($P$ and $\Cc$ are fixed).
For $\omega\in\Omega,$ we denote by $X_{i,j}(\omega),$ $Y_{j}(\omega)$ and $Z_{j}(\omega)$ the respective coordinates of $\omega$ for the variables $X_{i,j},$ $Y_j$ and $Z_{j}.$
Similarly, let $T$ be a random variable which is a function of $(X_{i,j})_{i\geq 1,j\geq 1},$  $(Y_j)_{j\geq 1}$ and $(Z_j)_{j\geq 1}$.
We denote by $T(\omega)$ the value taken by this variable at points $(X_{i,j}(\omega)),$ $(Y_{j}(\omega))$ and $(Z_{j}(\omega))$.

Applying~(3.6) to the empirical measures $\widehat{\mu}_m$ and $\widehat{Q}_{m,n}$ with the assumptions of Theorem~3.1, there exists $\tilde{\Omega}\subset\Omega$ such that $\mathbb{Q}(\tilde{\Omega})=1$ and for all $\omega\in\tilde{\Omega}$ there exists $N_\omega\in\NN$:
\begin{align*}
\forall m\geq N_\omega,\ \forall n\geq 1,\qquad  \|\widehat{\mu}_m(\omega)-\mu_{\Cc}^{}\|_\HH &\leq 2C\lambda_m\quad\text{and}\quad \|\widehat{Q}_{m,n}(\omega)-Q_n\|_\HH \leq 2C\lambda_{m},\\
2C\lambda_m &<\Delta_m,\\
\forall m\geq N_\omega,\ \forall n\geq N_\omega,\quad \|\widehat{Q}_{m,n}(\omega)-Q_P\|_\HH &\leq 2C\lambda_{mn}.
\end{align*}

\begin{proof}[Proof of Lemma \ref{lem:depth:mc:conv-ps}]
We introduce the variable,
$$
D(x|Q,\mu,\HH_{\Delta}^{n,m})=\inf\left\{Q(H_{u,x})/\mu(H_{u,x}),\ u\in\Ss,\ H_{u,x}\in \HH_{\Delta}^{n,m}\right\},
$$
where $Q$ and $\mu$ are two probability measures on $\RR^d.$ 
It is straightforward to show that, for all $\omega\in\tilde{\Omega},$ there exists $N_\omega\in\NN$ such that for all $m,n\geq N_\omega,$ we get,  
\begin{align*}\small
&\sup_{x\in \RR^d}\left|\widehat{D}(x|\widehat{Q}_{m,n}(\omega),\widehat{\mu}_{m}(\omega),\HH_{\Delta_m}^{n,m}(\omega)) - D(x|Q_P,\widehat{\mu}_{m}(\omega),\HH_{\Delta_m}^{n,m}(\omega))\right|
\leq \frac{\|\widehat{Q}_{m,n}(\omega)-Q_P\|_\HH}{\Delta_{m}}, \\
&\sup_{x\in \RR^d}\left|D(x|Q_P,\widehat{\mu}_{m}(\omega),\HH_{\Delta_m}^{n,m}(\omega))-D(x|Q_P,\mu_{\Cc},\HH_{\Delta_m}^{n,m}(\omega))\right|
\leq \frac{\|\widehat{\mu}_{m}(\omega)-\mu_\Cc^{}\|_\HH}{\Delta_{m}^2\left(1-2C\lambda_m/\Delta_m\right)}.
\end{align*}
Then we deduce that the variable
$$
\sup_{x\in \RR^d}\left|D(x|\widehat{Q}_{m,n},\widehat{\mu}_{m},\HH_{\Delta_m}^{n,m}) - D(x|Q_P,\mu_{\Cc},\HH_{\Delta_m}^{n,m})\right|
$$
converges a.s.\ to zero.
It remains to show that $D(x|Q_P,\mu_{\Cc},\HH_{\Delta_m}^{n,m})$ converges a.s.\ to $D(x|Q_P,\mu_\Cc)$ as $m,n\to\infty$ for a fixed point $x\in\RR^d.$ 

\textit{Case 1 : there exists $u_0\in\Ss$ such that $Q_P(H_{u_0,x})=0.$}
Then for all $(m,n),$ $H_{u_0,x}\in \HH_{\Delta_m}^{n,m},$ and $D(x|Q_P,\mu_\Cc)=D(x|Q_P,\mu_{\Cc},\HH_{\Delta_m}^{n,m})$ a.s. 

\textit{Case 2 : for all $u\in\Ss,$ $Q_P(H_{u,x})>0.$} 
Due to the fact that $D(x|Q_P,\mu_\Cc)$ is bounded by~$1$ (see Lemma \ref{lem:depth:bounded}) there exists a sequence $(u_k)$ of $\Ss$ such that
$$
D(x|Q_P,\mu_\Cc)=\lim_{k\to\infty}\frac{Q_P(H_{u_k,x})}{\mu_\Cc(H_{u_k,x})},\quad \text{and}\quad {Q_P(H_{u_k,x})}\geq{\mu_\Cc(H_{u_k,x})}>0.
$$
First we consider the sub-case where the sequence $(\mu_\Cc(H_{u_k,x}))$ is lower-bounded by a positive constant $\kappa>0$ (if $d=1$, only this case occurs because $\Ss$ is a finite set).
Since $(\lambda_m)$ and $(\Delta_m)$ are decreasing sequences, we have $\mu_\Cc(H_{u_k,x})\geq 2C\lambda_m+\Delta_m,$ for $m$ large enough and for all $k\in\NN.$
Then for all $\omega\in\tilde{\Omega},$ there exists $N_\omega\in\NN$ such that
\begin{align*}
\forall m\geq N_\omega,\ \forall k\in\NN,\qquad\widehat{\mu}_m(\omega)(H_{u_k,x})&\geq \mu_\Cc(H_{u_k,x}) - \|\widehat{\mu}_{m}(\omega)-\mu_\Cc^{}\|_\HH\\
&\geq \Delta_m,
\end{align*}
which means that $\forall m\geq N_\omega,$ $\forall k\in\NN,$ $H_{u_k,x}\in \HH_{\Delta_m}^{n,m}(\omega).$ 
Therefore we have, for all $m\geq N_\omega,$ $D(x|Q_P,\mu_\Cc)=D(x|Q_P,\mu_{\Cc},\HH_{\Delta_m}^{n,m}(\omega).$ 

A second sub-case occurs when the sequence $(\mu_\Cc(H_{u_k,x}))$ is decreasing to zero, i.e., for all $m$ there exists $M_m\in\NN$ such that for all $k>M_m,$ $\mu_\Cc(H_{u_k,x})<2C\lambda_m+\Delta_m.$
Let $m_0\in\NN$ such that there exists $k_0\in\NN$ for which $\mu_\Cc(H_{u_{k_0},x})\geq 2C\lambda_{m_0}+\Delta_{m_0}.$
Then we consider the increasing sequence $(k_m)_{m\geq m_0}$ of integers defined recursively by,
$$
k_{m_0} = k_0 \quad\text{and}\quad k_{m+1}=\sup\{k\geq k_m\ :\ \mu_\Cc(H_{u_{k_{m+1}},x})\geq 2C\lambda_{m+1}+\Delta_{m+1}\}.
$$
For all $\omega\in\tilde{\Omega},$ there exists $N_\omega\in\NN$ such that,
\begin{align*}
\forall m\geq N_\omega,\quad\widehat{\mu}_m(\omega)(H_{u_{k_m},x})\geq \Delta_m,
\end{align*}
i.e., $H_{u_{k_m},x}\in \HH_{\Delta_m}^{n,m}(\omega).$
Thus we obtain for all $\omega\in\tilde{\Omega},$ for all $m\geq N_\omega$ that
\begin{align*}
D(x|Q_P,\mu_\Cc)\leq D(x|Q_P,\mu_\Cc,\HH_{\Delta_m}^{n,m}(\omega))\leq \frac{Q_P(H_{u_{k_m},x})}{\mu_\Cc(H_{u_{k_m},x})}\quad\text{and}\quad \lim_{m\to\infty}\frac{Q_P(H_{u_{k_m},x})}{\mu_\Cc(H_{u_{k_m},x})}=D(x|Q_P,\mu_\Cc).
\end{align*}
\end{proof}

\begin{proof}[Proof of Lemma \ref{lem:depth:mc:conv-ps2}]
As the proof of Lemma \ref{lem:depth:mc:conv-ps}, we have that, for all $\omega\in\tilde{\Omega},$ there exists $N_\omega\in\NN$ such that for all $m\geq N_\omega,$ and for all $n\geq 1,$   
\begin{align*}\small
&\sup_{x\in \RR^d}\left|\widehat{D}(x|\widehat{Q}_{m,n}(\omega),\widehat{\mu}_{m}(\omega),\HH_{\Delta_m}^{n,m}(\omega)) - D(x|Q_n,\widehat{\mu}_{m}(\omega),\HH_{\Delta_m}^{n,m}(\omega))\right|
\leq \frac{\|\widehat{Q}_{m,n}(\omega)-Q_n\|_\HH}{\Delta_{m}}, \\
&\sup_{x\in \RR^d}\left|D(x|Q_n,\widehat{\mu}_{m}(\omega),\HH_{\Delta_m}^{n,m}(\omega))-D(x|Q_n,\mu_{\Cc},\HH_{\Delta_m}^{n,m}(\omega))\right|
\leq \frac{\|\widehat{\mu}_{m}(\omega)-\mu_\Cc^{}\|_\HH}{\Delta_{m}^2\left(1-2C\lambda_m/\Delta_m\right)}.
\end{align*}
Then we deduce that the random variable
$$
\sup_{x\in \RR^d}\left|D(x|\widehat{Q}_{m,n},\widehat{\mu}_{m},\HH_{\Delta_m}^{n,m}) - D(x|Q_n,\mu_{\Cc},\HH_{\Delta_m}^{n,m})\right|
$$
converges a.s.\ to zero.
It remains to show that $D(x|Q_n,\mu_{\Cc},\HH_{\Delta_m}^{n,m})$ converges a.s.\ to $D(x|Q_P,\mu_\Cc)$ as $m\to\infty$ for a fixed point $x\in\RR^d.$ 

\textit{Case 1 : there exists $u_0\in\Ss$ such that $Q_n(H_{u_0,x})=0.$} 
Then for all $m,$ $H_{u_0,x}\in \HH_{\Delta_m}^{n,m},$ and $D(x|Q_n,\mu_\Cc)=D(x|Q_n,\mu_{\Cc},\HH_{\Delta_m}^{n,m})$ a.s.

\textit{Case 2 : for all $u\in\Ss,$ $Q_n(H_{u,x})>0.$} 
Due to the fact that $D(x|Q_n,\mu_\Cc)$ is bounded by~$1$ (see Lemma \ref{lem:depth:bounded}) there exists a sequence $(u_k)$ of $\Ss$ such that,
$$
D(x|Q_n,\mu_\Cc)=\lim_{k\to\infty}\frac{Q_n(H_{u_k,x})}{\mu_\Cc(H_{u_k,x})},\quad \text{and}\quad {Q_n(H_{u_k,x})}\geq{\mu_\Cc(H_{u_k,x})}>0.
$$
First we consider the sub-case where the sequence $(\mu_\Cc(H_{u_k,x}))$ is lower-bounded by a positive constant $\kappa>0$ (if $d=1$, only this case occurs).
Since $(\lambda_m)$ and $(\Delta_m)$ are decreasing sequences, we have $\mu_\Cc(H_{u_k,x})\geq 2C\lambda_m+\Delta_m,$ for $m$ large enough and for all $k\in\NN.$
Then for all $\omega\in\tilde{\Omega},$ there exists $N_\omega\in\NN$ such that,
\begin{align*}
\forall k\in\NN,\quad\widehat{\mu}_m(\omega)(H_{u_k,x})&\geq \mu_\Cc(H_{u_k,x}) - \|\widehat{\mu}_{m}(\omega)-\mu_\Cc^{}\|_\HH\\
&\geq \Delta_m,
\end{align*}
i.e., $\forall k\in\NN,$ $H_{u_k,x}\in \HH_{\Delta_m}^{n,m}.$ 
Therefore we have for all $m\geq N_\omega$ that $D(x|Q_n,\mu_\Cc)=D(x|Q_n,\mu_{\Cc},\HH_{\Delta_m}^{n,m}(\omega)).$

A second sub-case occurs when the sequence $(\mu_\Cc(H_{u_k,x}))$ is decreasing to zero, i.e., for all $m$ there exists $M_m\in\NN$ such that for all $k>M_m$ $\mu_\Cc(H_{u_k,x})<2C\lambda_m+\Delta_m.$
Let $m_0\in\NN$ such that there exists $k_0\in\NN$ for which $\mu_\Cc(H_{u_{k_0},x})\geq 2C\lambda_{m_0}+\Delta_{m_0}.$
Then we consider the increasing sequence $(k_m)_{m\geq m_0}$ of integers defined recursively by,
$$
k_{m_0} = k_0 \quad\text{and}\quad k_{m+1}=\sup\{k\geq k_m\ :\ \mu_\Cc(H_{u_{k_{m+1}},x})\geq 2C\lambda_{m+1}+\Delta_{m+1}\}
$$
For all $\omega\in\tilde{\Omega},$ there exists $N_\omega\in\NN$ such that,
\begin{align*}
\forall m\geq N_\omega,\quad\widehat{\mu}_m(\omega)(H_{u_{k_m},x})\geq \Delta_m,
\end{align*}
that means $H_{u_{k_m},x}\in \HH_{\Delta_m}^{n,m}(\omega).$
Thus we obtain for all $\omega\in\tilde{\Omega},$ for all $m\geq N_\omega$ that
\begin{align*}
D(x|Q_n,\mu_\Cc)\leq D(x|Q_n,\mu_\Cc,\HH_{\Delta_m}^{n,m}(\omega))\leq \frac{Q_n(H_{u_{k_m},x})}{\mu_\Cc(H_{u_{k_m},x})}\quad\text{and}\quad \lim_{m\to\infty}\frac{Q_n(H_{u_{k_m},x})}{\mu_\Cc(H_{u_{k_m},x})}=D(x|Q_n,\mu_\Cc).
\end{align*}
\end{proof}

\begin{proof}[Proof of Lemma \ref{lem:depth:mc:measurable}]
Notice that the function $(x,z,u)\in\RR\times\RR^d\times\Ss\mapsto u^\top(x-z)\in\RR$ is measurable, and that the function
\begin{align*}
w:\RR^d\times\Ss &\to \RR_+\cup\{+\infty\}\times\RR_+\\
  (x,u)          &\mapsto \left(\widehat{Q}_{m,n}(H_{u,x})/\widehat{\mu}_{m}(H_{u,x}),\widehat{\mu}_{m}(H_{u,x})\right)
\end{align*}
is also measurable and takes a finite number of values. We denote by $0=v_0 < v_1 <\cdots <v_p=+\infty$ the collection of values which are taken by the first coordinate of $w.$
Let $V_q$ be the inverse image of $\{v_q\}\times[\Delta,1]$ under $w,$ $q=0,\ldots, p.$
We may rewrite $D(x|\widehat{Q}_{m,n},\widehat{\mu}_{m},\Delta)$ as
$$
\widehat{D}(x|\widehat{Q}_{m,n},\widehat{\mu}_{b},\Delta) = \sum_{q=0}^{p}v_q \mathds{1}_{B_q}(x), 
$$
where $(B_q)_{q=0,\ldots, p}$ are measurable subsets of $\RR^d$ defined recursively by
\begin{align*}
B_0 &=\left\{x\in\RR^d:\ \exists u\in\Ss,\ \widehat{Q}_{m,n}(H_{u,x})=0 \right\},\\
B_q &=\left\{x\in\RR^d:\ \exists u\in\Ss,\ (x,u)\in V_q\right\}\setminus \left\{\bigcup_{r=0}^{q-1}B_r\right\},\quad q=1,\ldots, p.
\end{align*}
\end{proof}

\begin{proof}[Proof of Lemma \ref{lem:depth:mc:bounded}]
Let $x$ be a fixed point of the locus of ${\Cc}$.
Since $\mu_\Cc$ is a non-atomic measure,  $x$ is not in $\YY_{m}\cup\XX_{n,m}$ almost surely. Let $\omega\in\Omega$ be fixed.

\textit{1. The one-dimensional case, $d=1$.} We have that 
$\widehat{Q}_{m,n}(\partial H_{u,x})=\widehat{\mu}_{m}(\partial H_{u,x})=0,$ for all $u\in\{-1,1\}$.

\textit{2. The multi-dimensional case, $d\geq 2$.} 
Assume there exists a.s.\ an affine subspace $A_{d-1}$ of dimension $d-1$ such that $\widehat{Q}_{m,n}(A_{d-1})=\widehat{\mu}_{m}(A_{d-1})=0$ and $x\in A_{d-1}.$
Since $0<\Delta<1/2,$ there exists $u\in\Ss$ such that $\partial H_{u,x}=A_{d-1}$ and $\hat{\mu}_m(H_{u,x}) > \Delta.$ 
Then we define the non-empty subset $\Ss_x$ of $\Ss$ such that
$$
\Ss_x=\left\{u\in \Ss:\ H_{u,x}\in \HH_{\Delta}^{n,m}\ \text{ and }\ \widehat{Q}_{m,n}(\partial H_{u,x})=\widehat{\mu}_{m}(\partial H_{u,x})=0 \text{ a.s.} \right\}.
$$
If there exists $u\in\Ss_x$ such that $\widehat{Q}_{m,n}(H_{u,x})/\widehat{\mu}_{m}(H_{u,x})\leq 1,$ the lemma is proved.
Otherwise for all $u\in\Ss_x$, we have  
$$
\widehat{Q}_{m,n}(H_{u,x})/\widehat{\mu}_{m}(H_{u,x}) > 1,\ \text{ and }\  
0<\widehat{\mu}_{m}(H_{-u,x}) < \Delta.
$$
If $\Delta< 1/m,$ then there exists $u\in\Ss_x$ such that $\widehat{Q}_{m,n}(H_{u,x})/\widehat{\mu}_{m}(H_{u,x})\leq 1$ by the \textit{reductio ad absurdum} argument.
If $\Delta > 1/m,$ then for all $u\in\Ss_x,$ $\widehat{Q}_{m,n}(H_{u,x})/\widehat{\mu}_{m}(H_{u,x})\leq 1/(1-\Delta).$
It remains to show the existence for $d\geq 2$ of such an affine subspace $A_{k}$ recursively on the dimension $k=1,\ldots,d-1$ of $A_{d-1}$. 

For $k=1,$ there exists a finite number (at most $m+nm$) of affine lines which contain $x$ and a point of the sample $\YY_{m}(\omega)\cup\XX_{n,m}(\omega).$ 
Since the set of affine lines which contain $x$ is continuous, there exists an affine line $A_1$ such that $\widehat{Q}_{m,n}(w)(A_1)=0$ and $\widehat{\mu}_{m}(w)(A_1)=0.$

Assume that $k\geq 2.$ 
Using the recursive assumption, there exists an affine subspace $A_{k-1}$ of dimension $k-1$ such that 
$\widehat{Q}_{m,n}(\omega)(A_{k-1})=\widehat{\mu}_{m}(\omega)(A_{k-1})=0$ and $x\in A_{k-1}$.  
Let $\mathcal{A}$ be the set of affine subspaces $A$ of dimension $k$ such that $A_{k-1}\subset A$.  
Then there exist at most $m+nm$ affine subspaces of $\mathcal{A}$ which contain at least one point of the sample $\YY_{m}(\omega)\cup\XX_{n,m}(\omega).$
Then there exists an affine subspace $A_k$ which contains no points of the sample $\YY_{m}(\omega)\cup\XX_{n,m}(\omega).$
\end{proof}

\section{Proof of Theorem~\ref{thm:consistency0}}
\label{sec:proofs:consistency}
Conditionally on the samples $\XX_{n,m}$ and $\YY_{m},$ we apply the Hoeffding inequality on the independent sum of bounded variables $\widehat{D}(Z_{i}|\widehat{Q}_{m,n},\widehat{\mu}_{m},\HH_{\Delta_m}^{n,m}).$
Then for all $\epsilon>0$, we get that the event
\begin{align*}
\Omega_\epsilon =\Omega\setminus \left( \left|\widehat{D}_{n,m,\Delta}(\Cc|\Xx_1,\ldots,\Xx_n)-\int_{{\Cc}} D(x|Q_n,\mu_\Cc,\HH_{\Delta_m}^{n,m}) d\mu_\Cc(x)\right|>\epsilon\right)
\end{align*} 
has a probability larger than $1-2\exp(-2\epsilon^2 m).$ 
Using Lemma \ref{lem:depth:mc:conv-ps2} and Lemma \ref{lem:depth:mc:bounded}, the dominated convergence theorem implies that given $\Xx_1,\ldots,\Xx_n,$
$$
\forall\omega\in\Omega_\epsilon,\quad\lim_{m\to\infty} \int_{{\Cc}} \widehat{D}(s|\widehat{Q}_{m,n}(\omega),\widehat{\mu}_{m}(\omega),\HH_{\Delta_m}^{n,m}(\omega)) d\mu_\Cc(s) = \int_{{\Cc}} D(s|Q_n,\mu_\Cc) d\mu_\Cc(s).
$$
Then we deduce that $\widehat{D}_{n,m,\Delta}(\Cc|\Xx_1,\ldots,\Xx_n)$ converges in probability to $D(\Cc|\Xx_1,\ldots,\Xx_n)$ as $m\to\infty.$

Similarly, conditionally on the samples $\XX_{n,m}$ and $\YY_{m},$ the Hoeffding inequality on the independent sum of bounded variables $D(Z_{i}|\widehat{Q}_{m,n},\widehat{\mu}_{m},\HH_{\Delta_m}^{n,m})$ allows to consider,
\begin{align*}
\Omega_\epsilon =\Omega\setminus \left( \left|\widehat{D}_{n,m,\Delta}(\Cc|\Xx_1,\ldots,\Xx_n)-\int_{{\Cc}} D(x|Q_P,\mu_\Cc,\HH_{\Delta_m}^{n,m}) d\mu_\Cc(x)\right|>\epsilon\right),
\end{align*} 
where $\epsilon>0$ and $\mathbb{Q}(\Omega_\epsilon)\geq 1-2\exp(-2\epsilon^2 m).$ 
Using Lemma \ref{lem:depth:mc:conv-ps} and Lemma \ref{lem:depth:mc:bounded}, the dominated convergence theorem implies that given the sequence $(\Xx_i)_{i\geq 1},$
$$
\forall\omega\in\Omega_\epsilon,\quad\lim_{m,n\to\infty} \int_{{\Cc}} D(s|\widehat{Q}_{m,n}(\omega),\widehat{\mu}_{m}(\omega),\HH_{\Delta_m}^{n,m}(\omega)) d\mu_\Cc(s) = \int_{{\Cc}} D(s|Q_P,\mu_\Cc) d\mu_\Cc(s).
$$
Then we deduce that $\widehat{D}_{n,m,\Delta}(\Cc|\Xx_1,\ldots,\Xx_n)$ converges in probability to $D(\Cc|P)$ as $n,m\to\infty$ and therefore that $D(\Cc|\Xx_1,\ldots,\Xx_n)$ converges in probability to $D(\Cc|P)$ as $n\to\infty.$

\section{Properties of the \CurveDepth}

\begin{lemma}\label{lem:deep:invariance}
The depth of a curve is invariant up to the similarities group,
$$
D(rA\Cc+b|Q_{P_{rA\Xx+b}})=D(\Cc|Q_{P_{\Xx}}),
$$
where $A$ is a $d\times d$ orthogonal matrix, $r>0$ and $b\in\RR^d$.
\end{lemma}

\begin{proof}[Proof of Lemma \ref{lem:deep:invariance}]
Let $b$ be a translation vector in $\RR^d,$ $r>0$ be a scalar and $A$ be a $d\times d$ orthogonal matrix.
The arc-length parametrization of the curve $rA\Cc+b$ is $t\mapsto rA\beta_{\Cc}(t)+b.$
We notice by using the substitution rule $u=t/r$ that
\begin{align*}
\mu_{rA\Cc+b}(rA\Hh+b)&=
\int_0^{1} \mathds{1}_{rA\beta_{\Cc}(t)+b\in rA\Hh+b} dt\\
&=\int_0^{1} \mathds{1}_{\beta_{\Cc}(t)\in \Hh} dt \\
&=\mu_{\Cc}(\Hh).
\end{align*}
Denote by $Q_X$ the distribution of a given random vector $X.$ 
We deduce that
$$
Q_{rAX+b}(rA\Hh+b)=\int \mu_{rA\Cc+b}(rA\Hh+b)dP(\Cc)=Q_X(\Hh).
$$

Then we get
\begin{align*}
D(rA\Cc+b|Q_{P_{rA\Xx+b}}) = \int_0^{1} D(rA\beta_\Cc(t)+b|Q_{rAX+b},\mu_{rA\Cc+b})dt,
\end{align*}
where
\begin{align*}
D(rA\beta_\Cc(t)+b|Q_{rAX+b},\mu_{rA\Cc+b})&=\inf_{u\in\Ss}\left\{\frac{Q_{rAX+b}(rAH_{u,x}+b)}{\mu_{rA\Cc+b}(rAH_{u,x}+b)}\right\}\\
&=\inf_{u\in\Ss}\left\{\frac{Q_{X}(H_{u,x})}{\mu_{{\Cc}}(H_{u,x})}\right\}= D(x|\Cc,P_{\Xx}).
\end{align*}
\end{proof}

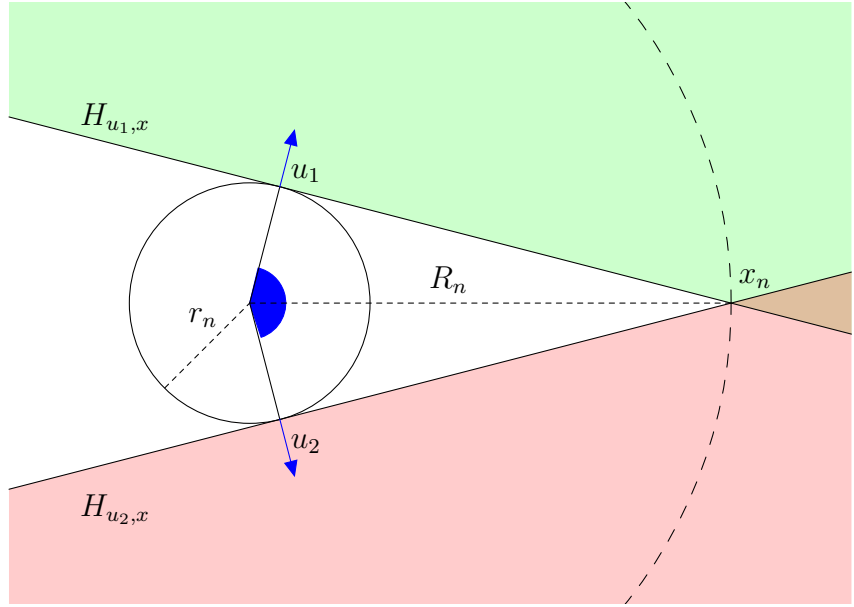
\begin{figure}[H]
\centering
  \begin{tikzpicture}[line cap=round,line join=round,>=triangle 45,x=1cm,y=1cm,scale=0.8]
\clip(-4.2,-5) rectangle (10,5);

\draw [shift={(0,0)},color=qqqqff,fill=qqqqff,fill opacity=0.1] (0,0) -- (-71.57:0.6) arc (-71.57:76.5:0.6) -- cycle;

\fill[color=zzttqq,fill=red!20,fill opacity=0.1] (-4,-3.096)-- (-4,-5.096) -- (10,-5.096)-- (10,0.516) -- cycle;
\fill[color=ttccqq,fill=green!20,fill opacity=0.1] (-4,3.096)-- (-4,5.096) -- (10,5.096)-- (10,-0.516) -- cycle;
\fill[color=ttccqq,fill=brown!50,fill opacity=0.1] (8,0) -- (10,0.516)-- (10,-0.516) -- cycle;

\draw(0,0) circle (2cm);
\draw [dash pattern=on 5pt off 5pt] (0,0) circle (8cm);

\draw [domain=-4:10] plot(\x,{(\x-8)*-0.258});
\draw [domain=-4:10] plot(\x,{(\x-8)*0.258});
\draw [dash pattern=on 2pt off 2pt] (-1.41,-1.41)-- (0,0);
\draw [dash pattern=on 2pt off 2pt] (0,0)-- (8,0);

\draw (-1.2,0.09) node[anchor=north west] {$r_n$};
\draw (2.8,0.8) node[anchor=north west] {$R_n$};

\draw (-3,3.5) node[anchor=north west] {$H_{u_1,x}$};
\draw (-3,-3) node[anchor=north west] {$H_{u_2,x}$};

\draw (0,0)-- (0.5,1.93);
\draw (0,0)-- (0.5,-1.93);

\draw[->,color=qqqqff] (0.5,1.93) -- (0.75,2.895);
\draw[->,color=qqqqff] (0.5,-1.93) -- (0.75,-2.895);
\draw (0.5,2.5) node[anchor=north west] {$u_1$};
\draw (0.5,-2) node[anchor=north west] {$u_2$};

\draw (8.4,0.4) node {$x_n$};
\end{tikzpicture}
\caption{Illustration for the proof of ``Vanishing at infinity'' property}
\end{figure}

\begin{lemma}\label{lem:deep:vanishing}
Let $(\Cc_n)$ be a sequence of unparametrized curves of length $\ell$ such that $\mu_{\Cc_n}$ is non-atomic, and
\begin{align*}
R_n = \inf_{x\in S_{\Cc_n}} |x|_2,\quad\text{and}\quad \lim_{n\to\infty} R_n=+\infty.
\end{align*}
Then the sequence of depths $(D(\Cc_n|P))$ converges to $0$ as $n\to\infty.$
\end{lemma}

We introduce the following notations and definitions.
We denote by $\Ss_r=\{x\in\RR^d:\ |x|_2=r\}$ the sphere of radius $r>0.$
A halfspace $H$ is \textit{tangent to} $\Ss_r$ if its boundary $\partial H$ is tangent to $\Ss_r$ and $H\cap\Ss_r$ is a singleton.

\begin{lemma}\label{lem:deep:vanishing:lemma}
  Let $x\in\RR^d$ such that $|x|_2=R$ and $d\geq 2.$
  Let $y\in\RR^d$ such that $|y|_2\geq |x|_2$ and $|x-y|_2\leq \ell <r.$
  There exist $u_1,u_2,\ldots, u_{2d-2}\in\Ss$ such that 
  \begin{itemize}
  \item for all $i=1,\ldots, 2d-2,$ $H_{u_i,x}$ is a tangent halfspace to $\Ss_r,$
  \item $y\in \cup_{i=1}^{2d-2} H_{u_i,x}.$
  \end{itemize}
\end{lemma}

\begin{proof}[Proof of Lemma \ref{lem:deep:vanishing}]
Since $Q_P$ is a probability meausure on $\RR^d,$ $Q_P$ is tight.
Then we can consider an increasing sequence $r_n$ such that: 
\begin{align*}
&1-\epsilon_n=Q_P\left(\{x\ :\ |x|_2\leq r_n\}\right) \text{ and } \lim_{n} \epsilon = 0,\\
&\ell<r_n<R_n/\sqrt{2} \text{ and } \epsilon_n\leq 1/(4d^2).
\end{align*}
Let $\Delta_n$ be a subset of $S_{\Cc_n}$ defined as
$$
\Delta_n=\left\{x\in S_{\Cc_n}\ :\ \forall u\in\Ss,\ Q_P(H_{u,x})>\epsilon_n \text{ or } \mu_{\Cc_n}(H_{u,x})<\sqrt{\epsilon_n}\right\}.
$$
Then for all $x\notin \Delta_n,$ $D(x|\mu_{\Cc_n},Q_P)\leq\sqrt{\epsilon_n}.$ 
It suffices to show that $\mu_{\Cc_n}(\Delta_n)<4d\epsilon_n$ and the lemma is proved. 
We define,
\begin{align*}
&R_n+\ell_n = \max_{x\in S_{\Cc_n}} |x|_2,\\
&F_n(t) = \mu_{\Cc_n}\left(\{x\ :\ |x|_2\leq R_n+t\}\right).
\end{align*}
Since $S_{\Cc_n}$ is a compact set, we get that $0<\ell_n<\ell,$ and $F_n$ is a cumulative distribution function (c.d.f.) whose support is $[0,\ell_n].$

\textit{1. The one-dimensional case, $d=1$.} The c.d.f.\ $F_n$ is a non atomic cumulative distribution function because $\mu_{\Cc_n}$ is non-atomic. Let $t_n$ be the quantile of order $1-\sqrt{\epsilon_n}$ of $F_n.$ Then $\Delta_n$ is included in the set of $x\in S_{\Cc_n}$ such that $|x|_2> t_n$, that is $\mu_{\Cc_n}(\Delta_n)\leq\sqrt{\epsilon_n}.$ 

\textit{2. The multi-dimensional case, $d\geq 2$.} 
The c.d.f.\ $F_n$ may have atoms.
First assume that $\mu_{\Cc_n}(\{x\ :\ |x|_2=R_n \})\geq 1-(2d-2)\sqrt{\epsilon_n}.$ 
We aim to show that $\Delta_n\subset \{x\ :\ |x|_2>R_n \}$ and then $\mu_{\Cc_n}(\Delta_n)<4d\epsilon_n.$
Let $x\in S_{\Cc_n}$ be such that $|x|_2=R_n$ and $x\in\Delta_n.$
Then for all $y\in S_{\Cc_n},$ $|y|_2 \geq |x|_2$ and $|x-y|_2\leq\ell .$
Using Lemma \ref{lem:deep:vanishing:lemma}, there exist $u_1,\ldots, u_{2d-2} \in \Ss$ such that :
$$
S_{\Cc_n}\subset \cup_{i=1}^{2d-2} H_{u_i,x}\qquad\text{and}\qquad\forall i=1,\ldots,2d-2,\quad \mu_{\Cc_n}(H_{u_i,x})<\sqrt{\epsilon_n},
$$
which is absurd.
Then $\Delta_n\subset \{x\ :\ |x|_2>R_n \}$ and $\mu_{\Cc_n}(\Delta_n)<4d\epsilon_n.$

Secondly assume that $\mu_{\Cc_n}(\{x\ :\ |x|_2=R_n \})< 1-(2d-2)\sqrt{\epsilon_n}$  : there exists a non negligible part of the curve $\Cc_n$ outside of $\Ss_{R_n}.$
Under this assumption, we distinguish two cases. 
The first sub-case is when $\mu_{\Cc_n}(\{x\ :\ |x|_2=R_n+\ell_n \})\geq (2d-2)\sqrt{\epsilon_n}.$
If $x\in \Delta_n$ then for all $y\in\{z\ :\ |z|_2=R_n+\ell_n \}\cap S_{\Cc_n},$ $|x|_2\leq |y|_2$ and $|x-y|_2\leq \ell.$
Using Lemma \ref{lem:deep:vanishing:lemma}, there exist $u_1,\ldots, u_{2d-2} \in \Ss$ such that :
$$
S_{\Cc_n}\cap\{x\ :\ |x|_2=R_n+\ell_n \}\subset \cup_{i=1}^{2d-2} H_{u_i,x}\qquad\text{and}\qquad\forall i=1,\ldots,2d-2,\quad \mu_{\Cc_n}(H_{u_i,x})<\sqrt{\epsilon_n},
$$
which is absurd.
Then $\Delta_n=\emptyset$ and $\mu_{\Cc_n}(\Delta_n)<4d\epsilon_n.$ 
 
The second sub-case is when $\mu_{\Cc_n}(\{x\ :\ |x|_2=R_n+\ell_n \})< (2d-2)\sqrt{\epsilon_n}.$
Let $t_n$ be the quantile of order $1-(2d-2)\sqrt{\epsilon_n}$ of $F_n.$ We know that $0<t_n\leq \ell_n.$
Using the same argument, we show that $\Delta_n\subset \{x\ :\: |x|_2 \geq t_n\}.$
If $\mu_{\Cc_n}(\{x\ :\ |x|_2=t_n \}) \geq (2d-2)\sqrt{\epsilon_n},$ one can show that $\Delta_n$ is a subset of $\{x\ :\: |x|_2 < t_n\}$ and $\mu_{\Cc_n}(\Delta_n)\leq (2d-2)\sqrt{\epsilon_n}.$
If $\mu_{\Cc_n}(\{x\ :\ |x|_2=t_n \}) < (2d-2)\sqrt{\epsilon_n},$ then 
$$
\mu_{\Cc_n}(\Delta_n)\leq \mu_{\Cc_n}(\{x\ :\ |x|_2=t_n \})+\mu_{\Cc_n}(\{x\ :\ |x|_2>t_n \}) \leq 4d\sqrt{\epsilon_n}.
$$
\end{proof}

\section{Algorithms}

\subsection{Procedures for Calculating Point \CurveDepth}\label{app:ssec:calcdepth}

Calculation of the point \curvedepth{} $\widehat{D}(x|\widehat{Q}_{m,n},\widehat{\mu}_{m},\HH_{\Delta}^{n,m})$ in $\mathbb{R}^2$ relies on the work by \cite{rousseeuw1996HD2}. The main idea is to regard all possible closed halfplanes by rotating a line containing $x$ in a counter-clockwise way. Here, the modifications are straightforward and narrow down to accounting for two differing samples $\widehat{Q}_{m,n}$ and $\widehat{\mu}_{m}$, threshold $\Delta$, and minimization functional represented by a ratio. The formal algorithm is detailed in Algorithm~\ref{alg:hdpart2d}, where w.l.o.g. $x=0$ due to translation invariance of the depth for convenience. The complexity of the algorithm is $O(mn\log(mn))$.

Algorithm~\ref{alg:hdpart3d} calculates $\widehat{D}(x|\widehat{Q}_{m,n},\widehat{\mu}_{m},\HH_{\Delta}^{n,m})$ in $\mathbb{R}^3$ modifying in the similar way the work by \cite{dyckerhoff2016HD}, and exploits Algorithm~\ref{alg:hdpart2d} as its basic element. The main idea is to perform Algorithm~\ref{alg:hdpart2d} for a projection of $\widehat{Q}_{m,n}$ and $\widehat{\mu}_{m}$ onto a plane orthogonal to the line connecting $0$ and one of the points from $\widehat{Q}_{m,n}$ and $\widehat{\mu}_{m}$. It can be easily extended to higher dimensions by additionally accounting for different combinations of points lying on this line on different sides form the (hyper)plane. The complexity of the algorithm is $O(m^2n^2\log(mn))$.

\noindent\hspace*{0\parindent}%
\begin{minipage}{0.875\textwidth}
\begin{algorithm}[H]
\caption{Routine for computing $D(\boldsymbol{0}|\widehat{Q}_{m,n},\widehat{\mu}_{m},\Delta)$ in dimension $2$}\label{alg:hdpart2d}
\begin{algorithmic}[1]
\Function{pTcd.2d}{$\bmy_1,...,\bmy_m$,$\bmx_1,...,\bmx_{m\cdot n},\Delta$}  \Comment{Point Tukey curve depth of $\boldsymbol{0}$}
\State{$n^z_{\bmy}\gets 0$}  \Comment{Number of $\mu$-points in the origin}
\State{$n^z_{\bmx}\gets 0$}  \Comment{Number of $Q$-points in the origin}
\State{$n^h_{\bmy}\gets 0$}  \Comment{Number of $\mu$-points in the halfplane}
\State{$n^h_{\bmx}\gets 0$}  \Comment{Number of $Q$-points in the halfplane}
\For{$i=1:m$}  \Comment{Go through all points sampled on $\mu$}
	\If{$|\bmy_i| = 0$}
		\State{$n^z_{\bmy} \gets n^z_{\bmy} + 1$}  \Comment{Count $\mu$-points in the origin}
	\Else
		\State{$P(i - n^z_{\bmy})\gets(\alpha = ATAN2(\bmy_i(2),\bmy_i(1)), c = 0)$}  \Comment{Save to all points}
		\If{$ATAN2(\bmy_i(2),\bmy_i(1)) < 0$}
			\State{$n^h_{\bmy}\gets n^h_{\bmy} + 1$}  \Comment{Count $\mu$-points in the (lower) halfplane}
		\EndIf
	\EndIf
\EndFor
\For{$i=1:(m \cdot n)$}  \Comment{Go through all points sampled on $Q$}
	\If{$|\bmx_i| = 0$} $n^z_{\bmx} \gets n^z_{\bmx} + 1$  \Comment{Count $Q$-points in the origin}
	\Else
		\State{$P(m - n^z_{\bmy} + i - n^z_{\bmx})\gets(\alpha = ATAN2(\bmx_i(2),\bmx_i(1)), c = 1)$}  \Comment{Save to all}		
		\If{$ATAN2(\bmx_i(2),\bmx_i(1)) < 0$}
			\State{$n^h_{\bmx}\gets n^h_{\bmx} + 1$}  \Comment{Count $Q$-points in the (lower) halfplane}
		\EndIf
	\EndIf
\EndFor
\State{$k \gets (m + m * n) - (n^z_{\bmy} + n^z_{\bmx})$}
\State{Sort $P$ w.r.t. $\alpha$s in ascending order}
\State{$D \gets \frac{n^h_{\bmx} / (m \cdot n)}{n^h_{\bmy} / m}$}  \Comment{Initialize the depth value}
\State{$j \gets n^h_{\bmy} + n^h_{\bmx} + 1$}

\algstore{hdpart2d1}
\end{algorithmic}
\end{algorithm}
\end{minipage}

\noindent\hspace*{0\parindent}%
\begin{minipage}{0.875\textwidth}
\begin{algorithm}[H]
\ContinuedFloat
\caption{Routine for computing $D(\boldsymbol{0}|\widehat{Q}_{m,n},\widehat{\mu}_{m},\Delta)$ in dimension $2$ (continued)}
\begin{algorithmic}
\algrestore{hdpart2d1}		

\State\Comment{Turn around counter-clockwise from the lower to the upper halfplane}
\For{$i = 1:(n^h_{\bmy} + n^h_{\bmx} + 1)$}
	\While{$j \le k$ and ((($i = n^h_{\bmy} + n^h_{\bmx} + 1$) and ($P(j).\alpha \le \pi$)) or \\
	\hfill($P(j).\alpha - \pi \le P(i).\alpha$))}
		\If{$P(i).c = 0$} $n^h_{\bmy} \gets n^h_{\bmy} + 1$  \Comment{Add the point to the halfplane}
		\Else \,$n^h_{\bmx} \gets n^h_{\bmx} + 1$
		\EndIf
		\If{($j < k$) and ($P(j + 1).\alpha = P(j).\alpha$)}  \Comment{If next point a tie ...}
			\State{$j \gets j + 1$;} \algorithmiccontinue  \Comment{... go directly to it}
		\EndIf
		\If{($i = n^h_{\bmy} + n^h_{\bmx} + 1$) or ($P(j).\alpha - \pi \le P(i).\alpha$)}  \Comment{If last point ...}
			\If{$n^h_{\bmx} = 0$}
				\State\Return{$0$} \Comment{... stop if zero depth achieved ...}
			\EndIf
			\If{$n^h_{\bmy} / m > \Delta$} \Comment{... otherwise still update the depth}
				\State{$D \gets \min\{D, \frac{n^h_{\bmx} / (m \cdot n)}{n^h_{\bmy} / m}\}$}
			\EndIf
		\EndIf
		\State{$j \gets j + 1$}  \Comment{Add point to the halfplane}
	\EndWhile
	\If{$i = n^h_{\bmy} + n^h_{\bmx} + 1$} \algorithmicbreak  \Comment{No more points to remove from the halfplane}
	\EndIf
	\If{$P(i).c = 0$} $n^h_{\bmy} \gets n^h_{\bmy} - 1$  \Comment{Remove the point from the halfplane}
	\Else \, $n^h_{\bmx} \gets n^h_{\bmx} - 1$
	\EndIf
	\If{($i < n^h_{\bmy} + n^h_{\bmx}$) and ($P(i + 1).\alpha = P(i).\alpha$)}  \Comment{If next point a tie ...}
		\State\algorithmiccontinue  \Comment{... go directly to the next iteration}
	\EndIf
	\If{$n^h_{\bmx} = 0$}
		\State\Return{$0$} \Comment{Stop if zero depth achieved}
	\EndIf
	\If{$n^h_{\bmy} / m > \Delta$} \Comment{Update the depth}
		\State{$D \gets \min\{D, \frac{n^h_{\bmx} / (m \cdot n)}{n^h_{\bmy} / m}\}$}
	\EndIf
\EndFor
\State{$j \gets 0$}

\algstore{hdpart2d2}
\end{algorithmic}
\end{algorithm}
\end{minipage}

\noindent\hspace*{0\parindent}%
\begin{minipage}{0.875\textwidth}
\begin{algorithm}[H]
\ContinuedFloat
\caption{Routine for computing $D(\boldsymbol{0}|\widehat{Q}_{m,n},\widehat{\mu}_{m},\Delta)$ in dimension $2$ (continued)}
\begin{algorithmic}
\algrestore{hdpart2d2}

\State\Comment{Turn around counter-clockwise from the upper to the lower halfplane}
\For{$i = (k - (n^h_{\bmy} + n^h_{\bmx}) + 1):(k + 1)$}
	\While{$j \le k < (n^h_{\bmy} + n^h_{\bmx})$ and ((($i = k + 1$) and ($P(j).\alpha \le 0$)) or \\
	\hfill($P(j).\alpha + \pi \le P(i).\alpha$))}
		\If{$P(i).c = 0$} $n^h_{\bmy} \gets n^h_{\bmy} + 1$  \Comment{Add the point to the halfplane}
		\Else \,$n^h_{\bmx} \gets n^h_{\bmx} + 1$
		\EndIf
		\If{($j < k - (n^h_{\bmy} + n^h_{\bmx})$) and ($P(j + 1).\alpha = P(j).\alpha$)}  \Comment{If a tie ...}
			\State{$j \gets j + 1$;} \algorithmiccontinue \Comment{... add it as well}
		\EndIf
		\If{($i = k + 1$) or ($P(j).\alpha + \pi \le P(i).\alpha$)}  \Comment{If last point ...}
			\If{$n^h_{\bmx} = 0$}
				\State\Return{$0$} \Comment{... stop if zero depth achieved ...}
			\EndIf
			\If{$n^h_{\bmy} / m > \Delta$} \Comment{... otherwise still update the depth}
				\State{$D \gets \min\{D, \frac{n^h_{\bmx} / (m \cdot n)}{n^h_{\bmy} / m}\}$}
			\EndIf
		\EndIf
		\State{$j \gets j + 1$}  \Comment{Add point to the halfplane}
	\EndWhile
	\If{$i = k + 1$}  \Comment{If last point ...}
		\State\algorithmicbreak  \Comment{... no points to remove from the halfplane, so stop the outer loop}
	\EndIf
	\If{$P(i).c = 0$} $n^h_{\bmy} \gets n^h_{\bmy} - 1$  \Comment{Remove the point from the halfplane}
	\Else \,$n^h_{\bmx} \gets n^h_{\bmx} - 1$
	\EndIf
	\If{($i < k$) and ($P(i + 1).\alpha = P(i).\alpha$)}  \Comment{If next point a tie ...}
		\State\algorithmiccontinue  \Comment{... go directly to the next iteration}
	\EndIf
	\If{$n^h_{\bmx} = 0$}
		\State\Return{$0$} \Comment{Stop if zero depth achieved}
	\EndIf
	\If{$n^h_{\bmy} / m > \Delta$} \Comment{Update the depth}
		\State{$D \gets \min\{D, \frac{n^h_{\bmx} / (m \cdot n)}{n^h_{\bmy} / m}\}$}
	\EndIf
\EndFor
\State\Return{$D$}
\EndFunction
\end{algorithmic}
\end{algorithm}
\end{minipage}

\noindent\hspace*{0\parindent}%
\begin{minipage}{0.875\textwidth}
\begin{algorithm}[H]
\caption{Routine for computing $D(\boldsymbol{0}|\widehat{Q}_{m,n},\widehat{\mu}_{m},\Delta)$ in dimension $3$}\label{alg:hdpart3d}
\begin{algorithmic}[1]
\Function{pTcd.3d}{$\bmy_1,...,\bmy_m$,$\bmx_1,...,\bmx_{m \cdot n},\Delta$}  \Comment{Point Tukey curve depth of $\boldsymbol{0}$}
\State{$D \gets 1$}
\For{$i=1:(m + m \cdot n)$}  \Comment{For each point of the both samples}
	\If{$i \le m$} $\bmz \gets \bmy_i$
	\Else \,$\bmz \gets \bmx_{i - m}$
	\EndIf
	\State{Compute a basis $\boldsymbol{A}=[\bma_1,\bma_2]$ of the hyperplane with normal $\bmz$}
	\State{$n^a_{\bmy} = 0$}  \Comment{Number of $\mu$-points in the origin \textbf{a}bove halfplane}
	\State{$n^b_{\bmy} = 0$}  \Comment{Number of $\mu$-points in the origin \textbf{b}elow plane}
	\State{$n_{\bmy} = 0$}  \Comment{Number of $\mu$-points not in the origin in the plane}
	\For{$j = 1:m$}  \Comment{Go through all points sampled on $\mu$}
		\If{$\boldsymbol{A}^{\top}\bmy_i = \boldsymbol{0}$}  \Comment{If projected in the origin}
			\If{$\bmz^{\top}\bmy_i > 0$} $n^a_{\bmy} \gets n^a_{\bmy} + 1$  \Comment{$\bmy_i$ above the plane}
			\ElsIf{$\bmz^{\top}\bmy_i < 0$} $n^b_{\bmy} \gets n^b_{\bmy} + 1$  \Comment{$\bmy_i$ below the plane}
			\Else \,$\bmy^{\prime}_{n_{\bmy} + 1} \gets \boldsymbol{A}^{\top}\bmy_i$; $n_{\bmy} \gets n_{\bmy} + 1$  \Comment{Add $\bmy_i$'s projection to the plane}
			\EndIf
		\Else \,$\bmy^{\prime}_{n_{\bmy} + 1} \gets \boldsymbol{A}^{\top}\bmy_i$; $n_{\bmy} \gets n_{\bmy} + 1$  \Comment{Add $\bmy_i$'s projection to the plane}
		\EndIf
	\EndFor
	\State{$n^a_{\bmx} = 0$}  \Comment{Number of $Q$-points in the origin \textbf{a}bove halfplane}
	\State{$n^b_{\bmx} = 0$}  \Comment{Number of $Q$-points in the origin \textbf{b}elow halfplane}
	\State{$n_{\bmx} = 0$}  \Comment{Number of $Q$-points not in the origin in the plane}
	\For{$j = 1:(m \cdot n)$}  \Comment{Go through all points sampled on $Q$}
		\If{$\boldsymbol{A}^{\top}\bmx_i = \boldsymbol{0}$}  \Comment{If projected in the origin}
			\If{$\bmz^{\top}\bmx_i > 0$} $n^a_{\bmx} \gets n^a_{\bmx} + 1$  \Comment{$\bmx_i$ above the plane}
			\ElsIf{$\bmz^{\top}\bmx_i < 0$} $n^b_{\bmx} \gets n^b_{\bmx} + 1$  \Comment{$\bmx_i$ below the plane}
			\Else \,$\bmx^{\prime}_{n_{\bmx} + 1} \gets \boldsymbol{A}^{\top}\bmx_i$; $n_{\bmx} \gets n_{\bmx} + 1$  \Comment{Add $\bmx_i$'s projection to the plane}
			\EndIf
		\Else \,$\bmx^{\prime}_{n_{\bmx} + 1} \gets \boldsymbol{A}^{\top}\bmx_i$; $n_{\bmx} \gets n_{\bmx} + 1$  \Comment{Add $\bmx_i$'s projection to the plane}
		\EndIf
	\EndFor
	\State{$D \gets \min\{D,pTcd.2d(\bmy_1,...,\bmy_{n_{\bmy}},\bmx_1,...,\bmx_{n_{\bmx}},n^a_{\bmy},n^b_{\bmy},n^a_{\bmx},n^b_{\bmx},\Delta)\}$}  \Comment{Update depth}
\EndFor
\State\Return{$D$}
\EndFunction
\end{algorithmic}
\end{algorithm}
\end{minipage}

\subsection{Procedure for Calculating the Distance Between Two Curves}\label{app:ssec:calcdist}

When calculating the metric $d_\Bg(\Cc_1,\Cc_2)$ in \eqref{eq:distance} one searches for two parametrizations that minimize the maximum norm between the two corresponding parametrized curves. Numerically this can be done by looking for a possible relocation of points from one curve to another keeping their order, in such a way that the distance of the longest relocation is minimal. Below we state the formal algorithm (Algorithm~\ref{alg:dist}) and demonstrate it on an example of calculation of distance between two digits. The complexity of the algorithm is $O(m_1m_2\log(m_1m_2))$ with $m_1$ and $m_2$ being the number of points of each of the curves $\Cc_1$ and $\Cc_2$, respectively.

In Figure~\ref{fig:digitsdists}, two curves (digits '1') are given in a pixel form, or more precisely by the coordinates of the corresponding pixel centers ordered from below to above in the image. Their mutual pixel-wise distances can be represented as a distance matrix; see Figure~\ref{fig:digitsmatrix1}. Keeping in mind that curves are (piece-wise) connected curves (in $\mathbb{R}^2$, here), optimal relocation of points will be approximated by a path in the matrix connecting the most upper left and the most bottom right cells in Figure~\ref{fig:digitsmatrix1}, such that the largest cell of this path will have smallest possible value. Algorithm~\ref{alg:dist} starts by eliminating the cells with the highest values and continues until any such path is blocked. The blockage of the path is identified when either at least one row or at least one column does not contain a single cell. Note that unreachable cells (while the path can proceed only right and down) are immediately deleted as well on each iteration of the algorithm.

\begin{figure}[H]
	\begin{center}
	\includegraphics[keepaspectratio=true,scale=0.5,trim=10mm 10mm 10mm 10mm,clip=true,page=1]{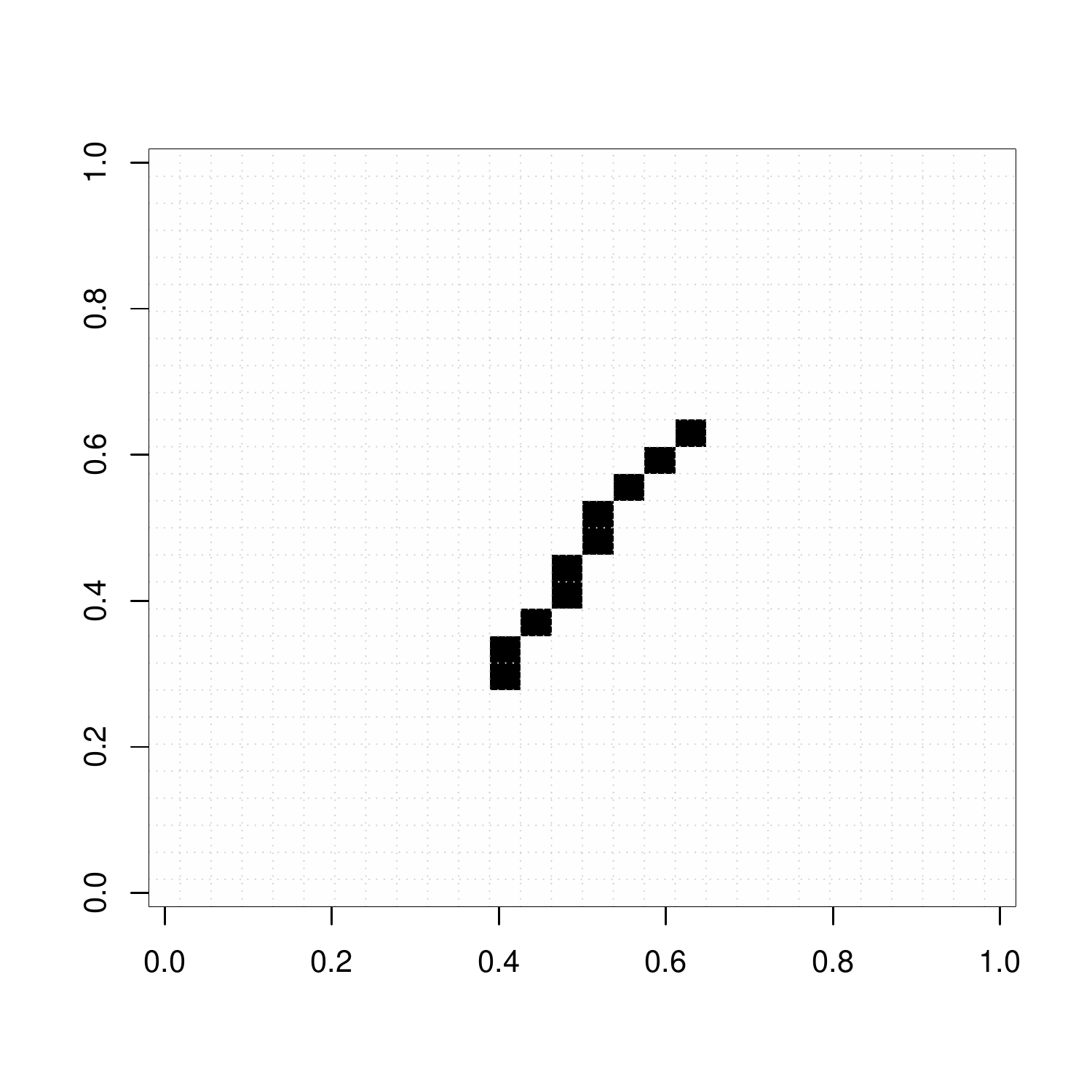}\qquad\includegraphics[keepaspectratio=true,scale=0.5,trim=10mm 10mm 10mm 10mm,clip=true,page=2]{digits_dists.pdf}
	\end{center}
	\caption{Two digits '1' used as an example to demonstrate calculation of the metric $d_\Bg$.}
	\label{fig:digitsdists}
\end{figure}

\begin{figure}[H]
	\begin{center}
		{\scriptsize
\begin{tabular}{R{.5cm}R{.5cm}R{.5cm}R{.5cm}R{.5cm}R{.5cm}R{.5cm}R{.5cm}R{.5cm}R{.5cm}R{.5cm}R{.5cm}}
\cellcolor{blue!25}0.107 & \cellcolor{blue!25}0.113 & \cellcolor{blue!25}0.129 & \cellcolor{blue!25}0.152 & \cellcolor{blue!25}0.179 & \cellcolor{blue!25}0.208 & \cellcolor{blue!25}0.240 & \cellcolor{blue!25}0.272 & \cellcolor{blue!25}0.305 & \cellcolor{blue!25}0.352 & \cellcolor{blue!25}0.385 & \cellcolor{blue!25}0.418 \\
\cellcolor{blue!25}0.113 & \cellcolor{blue!25}0.107 & \cellcolor{blue!25}0.113 & \cellcolor{blue!25}0.129 & \cellcolor{blue!25}0.152 & \cellcolor{blue!25}0.179 & \cellcolor{blue!25}0.208 & \cellcolor{blue!25}0.240 & \cellcolor{blue!25}0.272 & \cellcolor{blue!25}0.319 & \cellcolor{blue!25}0.352 & \cellcolor{blue!25}0.385 \\
\cellcolor{blue!25}0.101 & \cellcolor{blue!25}0.080 & \cellcolor{blue!25}0.071 & \cellcolor{blue!25}0.080 & \cellcolor{blue!25}0.101 & \cellcolor{blue!25}0.129 & \cellcolor{blue!25}0.160 & \cellcolor{blue!25}0.192 & \cellcolor{blue!25}0.226 & \cellcolor{blue!25}0.272 & \cellcolor{blue!25}0.305 & \cellcolor{blue!25}0.339 \\
\cellcolor{blue!25}0.113 & \cellcolor{blue!25}0.080 & \cellcolor{blue!25}0.051 & \cellcolor{blue!25}0.036 & \cellcolor{blue!25}0.051 & \cellcolor{blue!25}0.080 & \cellcolor{blue!25}0.113 & \cellcolor{blue!25}0.147 & \cellcolor{blue!25}0.182 & \cellcolor{blue!25}0.226 & \cellcolor{blue!25}0.260 & \cellcolor{blue!25}0.295 \\
\cellcolor{blue!25}0.147 & \cellcolor{blue!25}0.113 & \cellcolor{blue!25}0.080 & \cellcolor{blue!25}0.051 & \cellcolor{blue!25}0.036 & \cellcolor{blue!25}0.051 & \cellcolor{blue!25}0.080 & \cellcolor{blue!25}0.113 & \cellcolor{blue!25}0.147 & \cellcolor{blue!25}0.192 & \cellcolor{blue!25}0.226 & \cellcolor{blue!25}0.260 \\
\cellcolor{blue!25}0.179 & \cellcolor{blue!25}0.143 & \cellcolor{blue!25}0.107 & \cellcolor{blue!25}0.071 & \cellcolor{blue!25}0.036 & \cellcolor{blue!25}0.000 & \cellcolor{blue!25}0.036 & \cellcolor{blue!25}0.071 & \cellcolor{blue!25}0.107 & \cellcolor{blue!25}0.147 & \cellcolor{blue!25}0.182 & \cellcolor{blue!25}0.217 \\
\cellcolor{blue!25}0.214 & \cellcolor{blue!25}0.179 & \cellcolor{blue!25}0.143 & \cellcolor{blue!25}0.107 & \cellcolor{blue!25}0.071 & \cellcolor{blue!25}0.036 & \cellcolor{blue!25}0.000 & \cellcolor{blue!25}0.036 & \cellcolor{blue!25}0.071 & \cellcolor{blue!25}0.113 & \cellcolor{blue!25}0.147 & \cellcolor{blue!25}0.182 \\
\cellcolor{blue!25}0.253 & \cellcolor{blue!25}0.217 & \cellcolor{blue!25}0.182 & \cellcolor{blue!25}0.147 & \cellcolor{blue!25}0.113 & \cellcolor{blue!25}0.080 & \cellcolor{blue!25}0.051 & \cellcolor{blue!25}0.036 & \cellcolor{blue!25}0.051 & \cellcolor{blue!25}0.071 & \cellcolor{blue!25}0.107 & \cellcolor{blue!25}0.143 \\
\cellcolor{blue!25}0.295 & \cellcolor{blue!25}0.260 & \cellcolor{blue!25}0.226 & \cellcolor{blue!25}0.192 & \cellcolor{blue!25}0.160 & \cellcolor{blue!25}0.129 & \cellcolor{blue!25}0.101 & \cellcolor{blue!25}0.080 & \cellcolor{blue!25}0.071 & \cellcolor{blue!25}0.051 & \cellcolor{blue!25}0.080 & \cellcolor{blue!25}0.113 \\
\cellcolor{blue!25}0.339 & \cellcolor{blue!25}0.305 & \cellcolor{blue!25}0.272 & \cellcolor{blue!25}0.240 & \cellcolor{blue!25}0.208 & \cellcolor{blue!25}0.179 & \cellcolor{blue!25}0.152 & \cellcolor{blue!25}0.129 & \cellcolor{blue!25}0.113 & \cellcolor{blue!25}0.071 & \cellcolor{blue!25}0.080 & \cellcolor{blue!25}0.101
\end{tabular}}
	\end{center}
	\caption{Pixel-wise distance matrix for the two digits '1' from Figure~\ref{fig:digitsdists}.}
	\label{fig:digitsmatrix1}
\end{figure}

\noindent\hspace*{0\parindent}%
\begin{minipage}{0.875\textwidth}
\begin{algorithm}[H]
\caption{Routine for computing $d_\Bg(\Cc_1,\Cc_2)$}\label{alg:dist}
\begin{algorithmic}[1]
\Function{distance}{$\bmx_1,...,\bmx_{m_1}$,$\bmy_1,...,\bmy_{m_2}$}  \Comment{Distance between sampled curves}
\For{$i=1:m_1$}
	\For{$j=1:m_2$}
		\State $cells\bigl(m_2(i - 1) + j\bigr) = (i,j,d_{ij} = \|\bmx_i - \bmy_j\|_2)$ \Comment{Calculate cell-wise distances}
	\EndFor
\EndFor
\State Sort $cells$ w.r.t. $d_{\cdot\cdot}$-s in descending order
\State $M = (0,...,0)_{m_1}\times(0,...,0)_{m_2}$ \Comment{$m_1\times m_2$ matrix filled with $0$}
\State $rowMaxs = (m_2,...,m_2)_{m_1}$ \Comment{Vector of length $m_1$ having all entries equal $m_2$}
\State $rowMins = (0,...,0)_{m_1}$ \Comment{Vector of length $m_1$ having all entries equal $0$}
\State $colMaxs = (m_1,...,m_1)_{m_2}$ \Comment{Vector of length $m_2$ having all entries equal $m_1$}
\State $colMins = (0,...,0)_{m_2}$ \Comment{Vector of length $m_2$ having all entries equal $0$}
\State $k = 1$

	\algstore{dist1}
\end{algorithmic}
\end{algorithm}
\end{minipage}

\noindent\hspace*{0\parindent}%
\begin{minipage}{0.875\textwidth}
\begin{algorithm}[H]
\ContinuedFloat
\caption{Routine for computing $d_\Bg(\Cc_1,\Cc_2)$ (continued)}\label{alg:dist2}
\begin{algorithmic}
	\algrestore{dist1}

\While{$k <= m_1\cdot m_2$}
	\State $d = cells(k).d_{ij}$
	\While{$cells(k).d_{ij} = d$}
		\State $M(i,j) = 1$
		\If{$rowMaxs(i) = j + 1$} \Comment{If blocking cells above, then ...}
			\State $l\gets j$
			\While{$l \ge 1 \text{ and } M(i,l) = 1$}
				\State $M(0:i,l)\gets 1$; $l \gets l - 1$ \Comment{... mark cells above}
			\EndWhile
			\State $rowMaxs(i) = l + 1$ \Comment{Update maximum row's extension}
		\EndIf
		\If{$rowMins(i) = j - 1$} \Comment{If blocking cells below, then ...}
			\State $l\gets j$
			\While{$l \le m_2 \text{ and } M(i,l) = 1$}
				\State $M(i:m_1,l)\gets 1$; $l \gets l + 1$ \Comment{... mark cells below}
			\EndWhile
			\State $rowMins(i) = l - 1$ \Comment{Update minimum row's extension}
		\EndIf
		\If{$colMaxs(j) = i + 1$} \Comment{If blocking cells to the left, then ...}
			\State $l\gets i$
			\While{$l \ge 1 \text{ and } M(l,j) = 1$}
				\State $M(l,0:j)\gets 1$; $l \gets l - 1$ \Comment{... mark cells to the left}
			\EndWhile
			\State $colMaxs(j) = l + 1$ \Comment{Update maximum column's extension}
		\EndIf
		\If{$colMins(j) = i - 1$} \Comment{If blocking cells to the left, then ...}
			\State $l\gets i$
			\While{$l \le m_1 \text{ and } M(l,j) = 1$}
				\State $M(l,j:m_2)\gets 1$; $l \gets l + 1$ \Comment{... mark cells to the left}
			\EndWhile
			\State $colMins(j) = l - 1$ \Comment{Update minimum column's extension}
		\EndIf
		\State $k\gets k + 1$
	\EndWhile	
	\If{$\min{rowMaxs} = 1 \text{ or } \max{rowMins} = m_2 \text{ or } \min{colMaxs} = 1 \text{ or }$ \\
	\hfill $\max{colMins} = m_1$} \Comment{If the route through the matrix is blocked}
		\State\algorithmicbreak
	\EndIf
\EndWhile
\State\Return{d}
\EndFunction
\end{algorithmic}
\end{algorithm}
\end{minipage}

\subsection{Procedure for the Clustering of Curves}\label{app:ssec:clustcurv}

\noindent\hspace*{0\parindent}%
\begin{minipage}{0.875\textwidth}
\begin{algorithm}[H]
\caption{Clustering of curves \citep[following][]{jornsten2004clustering}}\label{alg:clust}
\begin{algorithmic}[1]
\Function{DDclustCurve}{$\Cc_1,\dots,\Cc_n$} \Comment{Input unlabeled curves}
\State{Initialize $\{I_k\}_1^K$ randomly}
\State{$m\gets 0$; $\beta\gets -1$; $i\gets 0$} \Comment{Initialize iterated variables}
\While{$m < M$ or $j < maxIter$} \Comment{Termination criterion}
	\State{Calculate $C_i(\{I_k\}_1^K)$, $i=1,...,n$}
	\State{Identify a set $S=\{i\,:\,C_i(\{I_k\}_1^K)\le T\}$} \Comment{Candidates for reallocation}
	\State{$f\gets \boldsymbol{false}$}
	\While{$S\ne \emptyset$}
		\State{For a random subset $E\subset S$, reallocate observations to get partitioning $\{\tilde I_k\}_1^K$}
		\If{$C(\{\tilde I_k\}_1^K) > C(\{I_k\}_1^K)$ or $\mathcal{B}\Bigl(P\bigl(C(\{I_k\}_1^K) - C(\{\tilde I_k\}_1^K),\beta\bigr)\sim b = 1\Bigr)$}
			\State{$I_k\gets\tilde I_k$, $k=1,\dots,K$; $f\gets \boldsymbol{true}$} \Comment{Accept reallocation}
		\EndIf
		\State{$S\gets S \setminus E$}
	\EndWhile
	\State{$\beta\gets 2\beta$; $j\gets j + 1$} \Comment{Increase simulated-annealing temperature}
	\If{$f = \boldsymbol{true}$}
		\State{$m\gets 0$}
	\Else
		\State{$m\gets m + 1$} \Comment{No changes on current iteration}
	\EndIf
\EndWhile
\State\Return{$\{I_k\}_1^K$} \Comment{Output the final clustering}
\EndFunction
\end{algorithmic}
\end{algorithm}
\end{minipage}

\section{Numerical experiments and Applications}\label{app:sec:simulation}

In this section, some additional materials are proposed in order to illustrate the convergence of the Monte Carlo approximation of our \curvedepth{} with the size $n$ of the sample of curves and the size $m$ of the Monte Carlo sample. We conduct the Monte Carlo study with three models of two-dimensional curves.

\subsection{Simple examples and their explicit depth expression}\label{app:ssec:examples}

\paragraph{Segments on a line.}
We observe $n$ non-overlapping segments on a line.
Without loss of generality, we denote by $\Xx_k$ the $k^{\text{th}}$ segment ($k=1,\ldots,n$), from left to right.
For $t\in[0,1]$, we have
\begin{align*}
\text{for } k\in\{1,n\},\ D(\beta_{\Xx_k}(t)|Q_n,\mu_{\Xx_k})&=1/n,\\ 
\text{for } k\notin\{1,n\},\ D(\beta_{\Xx_k}(t)|Q_n,\mu_{\Xx_k})&= \frac{t+(n-1)t_k}{nt}\mathds{1}_{t\geq t_k} + \frac{1-t+(n-1)(1-t_k)}{n(1-t)}\mathds{1}_{t< t_k}, 
\end{align*}
where $t_k=(k-1)/(n-1).$

\paragraph{Star segments.}
Let $\Cc_\theta$ be the segment in $\RR^2$ of starting point $(0,0)$ and ending point $(\cos(\theta),\sin(\theta)),$ for $\theta\in[0,2\pi)$.
We define $\Xx\sim P$ as the random curve generated from the following scheme :
\begin{align*}
\theta\sim \Uu[0,2\pi],\ \Xx=\Cc_{\theta} .
\end{align*}
For $t\in[0,1]$, we have
\begin{align*}
D(\beta_{\Cc_\theta}(t)|Q,\mu_{\Xx})=\left\{ 
\begin{array}{ll}
1/2, & \text{if } t=0,\\
g(t), & \text{if } t\in(0,1),\\
0, & \text{if } t=1,
\end{array}
\right.
\end{align*}
where $g:(0,1)\to(0,1)$ is a function defined by
\begin{align*}
g(t) &= \min\left\{\frac{1}{2}, \inf_{\alpha\in(0,\pi/2]} \frac{q_{t,\alpha}}{1-t}, \inf_{\alpha\in(-\pi/2,0)} \frac{1-q_{t,-\alpha}}{t}, \inf_{\alpha\in(-\pi,-pi/2)} \frac{1-q_{t,\pi+\alpha}}{t}, \inf_{\alpha\in(\pi/2,\pi)} \frac{q_{t,\pi-\alpha}}{1-t}\right\},\\
q_{t_\alpha} &= \frac{\pi - \sin^{-1}(t \sin\alpha)}{2\pi} -\frac{t\sin(\alpha)}{2\pi}\log\frac{1+\cos\left[\sin^{-1}(t\sin(\alpha))\right]}{1-\cos\left[\sin^{-1}(t\sin(\alpha))\right]}.
\end{align*}
The population version of the curve depth is defined as
\begin{align*}
D(\Cc_\theta|P) &= \int_0^1 g(r) dr .
\end{align*}

%

\paragraph{Concentric circles.}
Let $\Cc_r$ be the circle of center $0$ and radius $r>0$ in $\RR^2$.
We define $\Xx\sim P$ as the random curve generated from the following scheme :
\begin{align*}
R\sim \Uu[0,1],\ \Xx=\Cc_{R} .
\end{align*}
The population version of the curve depth is defined as
\begin{align*}
D(\Cc_r|P) &= 0 \text{ if } r=1;\\
D(\Cc_r|P) &= \min\left\{1,\frac{\cos^{-1}(r)}{\pi} - \frac{r}{\pi}\log\left(\frac{1+\sin\cos^{-1}(r)}{r}\right),\right. \\ 
&\qquad\quad \left. \inf_{\alpha\in(0,\pi/2)} \frac{2\pi - 2\cos^{-1}(r \sin\alpha)}{\pi+2\alpha} + \frac{2r\sin(\alpha)}{\pi+2\alpha}\log\frac{1+\sin\left[\cos^{-1}(r\sin(\alpha))\right]}{r\sin(\alpha)} , \right.\\
&\qquad\quad \left. \inf_{\alpha\in(0,\pi/2)} \frac{2\cos^{-1}(r \sin\alpha)}{\pi-2\alpha} - \frac{2r\sin(\alpha)}{\pi-2\alpha}\log\frac{1+\sin\left[\cos^{-1}(r\sin(\alpha))\right]}{r\sin(\alpha)} \right\}.
\end{align*}
For $n\geq 1,$ we denote by $R_i$ the radius of the circle $\Xx_i$ for all $i=1,\ldots,n.$ The sample depth of $\Cc_r$ with respect to $\{\Xx_1,\ldots,\Xx_n\}$ is
\begin{align}
D({\Cc_r}|\Xx_1,\ldots, \Xx_n) &= 0 \text{ if } r\geq \max_{i=1\ldots n} R_i;\label{app:eq:concetric:depth:sample} \\
D(\Cc_r|\Xx_1,\ldots, \Xx_n) &= \min\left\{1, \inf_{\alpha\in(0,\pi/2)} \frac{2\pi}{n(\pi-2\alpha)} \sum_{i=1}^n \frac{\cos^{-1}(r\sin(\alpha))}{\pi}\mathds{1}_{R_i>r\sin(\alpha)}, \right.\nonumber\\ 
&\qquad\quad \left. \inf_{\alpha\in(0,\pi/2]} \frac{2\pi}{n(\pi+2\alpha)}\sum_{i=1}^n\left(\mathds{1}_{R_i\leq r\sin(\alpha)} + \frac{\cos^{-1}(r\sin(\alpha))}{\pi}\mathds{1}_{R_i>r\sin(\alpha)} \right) \right\}.\nonumber
\end{align}

\subsection{Monte Carlo Approximation of the \CurveDepth}\label{app:ssec:simulation}
We illustrate the convergence of the Monte Carlo approximation of our \curvedepth{} according to the size $n$ of the sample and the Monte Carlo size $m$ on three simulation schemes.

\paragraph{Scheme 1 : Concentric circles.}
First, we consider the population of concentric circles with radius lying in the interval $(0,1)$ described in Section \ref{app:ssec:examples}.
We fix the sample size $n$ of $\{\Xx_1,\ldots, \Xx_n\},$ and we compute the depth of circles $\Cc_r$ of radius $r\in\{0.1,0.4,0.5,0.6,0.9\}$.

In the companion package of the paper, we propose an algorithm to approximate the depth $D(\Cc|\Xx_1,\ldots,\Xx_n)$ using a Monte Carlo estimate (see line 11 of Algorithm \ref{alg:montecarlo}):
$$
\frac{1}{m}\sum_{k=1}^{m} \widehat{D}(Z_{k}|\widehat{Q}_{m,n},\widehat{\mu}_{m},\HH_{\Delta}^{n,m}),
$$
where  $\widehat{Q}_{m,n}$ is the empirical measure associated to an i.i.d.\ sample $\XX_{n,m}=(X_{i,j})_{i=1,\ldots,n;j=1,\ldots,m}$ from $Q_{P_n}$, where $\widehat{\mu}_{m}$ is the empirical measure associated to an i.i.d.\ sample $\YY_{m}=(Y_{j})_{j=1,\ldots,m}$ from $\mu_\Cc$ and $(Z_{j})_{j=1,\ldots,m}$ is an i.i.d.\ sample from $\mu_\Cc$. 
Theorem~\ref{thm:consistency0} states that this Monte Carlo estimator is consistent as $m$ goes to $+\infty$.
We emphasize that the threshold $\Delta$ is here useless due to the geometry of circles :  the halfspaces $H_{x,u}$ such that $\mu_\Cc(H_{x,u})$ is small are those for which the ratio $\widehat{Q}_{n,m}/\widehat{\mu}_m(H_{x,u})$ is larger than $1.$
 
First, we fix the sample $\Xx_1,\ldots,\Xx_n$ and we compute $100$ replications of the Monte Carlo estimate of $D(\Cc_r|\Xx_1,\ldots,\Xx_n)$; see Figure~\ref{app:fig:circles:m:pack}.
The Monte Carlo estimates tend to underestimate the sample depth. Moreover the variability of the Monte Carlo estimates depends on the position of the circle $\Cc$ with respect to the sample of curves. 
Nevertheless both the bias and the variance decrease as $m$ increases.

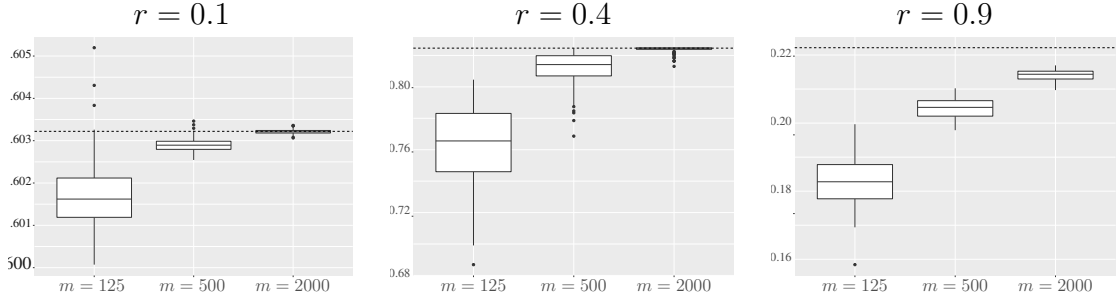
\begin{figure}[H]
\begin{center}
\begin{tabular}{ccc}
		$r=0.1$ & $r=0.4$ & $r=0.9$\\
		\resizebox{0.275\linewidth}{!}{    
		\begin{tikzpicture}[x=1pt,y=1pt]
\definecolor{fillColor}{RGB}{255,255,255}
\path[use as bounding box,fill=fillColor,fill opacity=0.00] (0,0) rectangle (573,430);
\begin{scope}
\path[clip] (  0.00,  0.00) rectangle (573,430);
\definecolor{drawColor}{RGB}{255,255,255}
\definecolor{fillColor}{RGB}{255,255,255}

\path[draw=drawColor,line width= 0.6pt,line join=round,line cap=round,fill=fillColor] (  0.00,  0.00) rectangle (573,430);
\end{scope}

\begin{scope}
\path[clip] ( 42.95, 30.69) rectangle (572.66,428.12);
\definecolor{fillColor}{gray}{0.92}

\path[fill=fillColor] ( 42.95, 30.69) rectangle (572.66,428.12);
\definecolor{drawColor}{RGB}{255,255,255}

\path[draw=drawColor,line width= 0.3pt,line join=round] ( 42.95, 78.96) --
	(572.66, 78.96);

\path[draw=drawColor,line width= 0.3pt,line join=round] ( 42.95,149.46) --
	(572.66,149.46);

\path[draw=drawColor,line width= 0.3pt,line join=round] ( 42.95,219.95) --
	(572.66,219.95);

\path[draw=drawColor,line width= 0.3pt,line join=round] ( 42.95,290.45) --
	(572.66,290.45);

\path[draw=drawColor,line width= 0.3pt,line join=round] ( 42.95,360.94) --
	(572.66,360.94);

\path[draw=drawColor,line width= 0.6pt,line join=round] ( 42.95, 43.72) --
	(572.66, 43.72);

\path[draw=drawColor,line width= 0.6pt,line join=round] ( 42.95,114.21) --
	(572.66,114.21);

\path[draw=drawColor,line width= 0.6pt,line join=round] ( 42.95,184.71) --
	(572.66,184.71);

\path[draw=drawColor,line width= 0.6pt,line join=round] ( 42.95,255.20) --
	(572.66,255.20);

\path[draw=drawColor,line width= 0.6pt,line join=round] ( 42.95,325.70) --
	(572.66,325.70);

\path[draw=drawColor,line width= 0.6pt,line join=round] ( 42.95,396.19) --
	(572.66,396.19);

\path[draw=drawColor,line width= 0.6pt,line join=round] (142.27, 30.69) --
	(142.27,428.12);

\path[draw=drawColor,line width= 0.6pt,line join=round] (307.81, 30.69) --
	(307.81,428.12);

\path[draw=drawColor,line width= 0.6pt,line join=round] (473.34, 30.69) --
	(473.34,428.12);
\definecolor{drawColor}{gray}{0.20}
\definecolor{fillColor}{gray}{0.20}

\path[draw=drawColor,line width= 0.4pt,line join=round,line cap=round,fill=fillColor] (142.27,410.05) circle (  1.96);

\path[draw=drawColor,line width= 0.4pt,line join=round,line cap=round,fill=fillColor] (142.27,347.34) circle (  1.96);

\path[draw=drawColor,line width= 0.4pt,line join=round,line cap=round,fill=fillColor] (142.27,313.99) circle (  1.96);

\path[draw=drawColor,line width= 0.6pt,line join=round] (142.27,193.02) -- (142.27,273.53);

\path[draw=drawColor,line width= 0.6pt,line join=round] (142.27,127.54) -- (142.27, 48.75);
\definecolor{fillColor}{RGB}{255,255,255}

\path[draw=drawColor,line width= 0.6pt,line join=round,line cap=round,fill=fillColor] ( 80.20,193.02) --
	( 80.20,127.54) --
	(204.35,127.54) --
	(204.35,193.02) --
	( 80.20,193.02) --
	cycle;

\path[draw=drawColor,line width= 1.1pt,line join=round] ( 80.20,157.95) -- (204.35,157.95);
\definecolor{fillColor}{gray}{0.20}

\path[draw=drawColor,line width= 0.4pt,line join=round,line cap=round,fill=fillColor] (307.81,281.70) circle (  1.96);

\path[draw=drawColor,line width= 0.4pt,line join=round,line cap=round,fill=fillColor] (307.81,275.89) circle (  1.96);

\path[draw=drawColor,line width= 0.4pt,line join=round,line cap=round,fill=fillColor] (307.81,287.85) circle (  1.96);

\path[draw=drawColor,line width= 0.6pt,line join=round] (307.81,254.28) -- (307.81,271.99);

\path[draw=drawColor,line width= 0.6pt,line join=round] (307.81,240.76) -- (307.81,223.08);
\definecolor{fillColor}{RGB}{255,255,255}

\path[draw=drawColor,line width= 0.6pt,line join=round,line cap=round,fill=fillColor] (245.73,254.28) --
	(245.73,240.76) --
	(369.88,240.76) --
	(369.88,254.28) --
	(245.73,254.28) --
	cycle;

\path[draw=drawColor,line width= 1.1pt,line join=round] (245.73,247.73) -- (369.88,247.73);
\definecolor{fillColor}{gray}{0.20}

\path[draw=drawColor,line width= 0.4pt,line join=round,line cap=round,fill=fillColor] (473.34,280.65) circle (  1.96);

\path[draw=drawColor,line width= 0.4pt,line join=round,line cap=round,fill=fillColor] (473.34,279.87) circle (  1.96);

\path[draw=drawColor,line width= 0.4pt,line join=round,line cap=round,fill=fillColor] (473.34,260.83) circle (  1.96);

\path[draw=drawColor,line width= 0.4pt,line join=round,line cap=round,fill=fillColor] (473.34,259.78) circle (  1.96);

\path[draw=drawColor,line width= 0.4pt,line join=round,line cap=round,fill=fillColor] (473.34,279.94) circle (  1.96);

\path[draw=drawColor,line width= 0.6pt,line join=round] (473.34,272.22) -- (473.34,278.00);

\path[draw=drawColor,line width= 0.6pt,line join=round] (473.34,267.97) -- (473.34,262.36);
\definecolor{fillColor}{RGB}{255,255,255}

\path[draw=drawColor,line width= 0.6pt,line join=round,line cap=round,fill=fillColor] (411.27,272.22) --
	(411.27,267.97) --
	(535.42,267.97) --
	(535.42,272.22) --
	(411.27,272.22) --
	cycle;

\path[draw=drawColor,line width= 1.1pt,line join=round] (411.27,269.96) -- (535.42,269.96);
\definecolor{drawColor}{RGB}{0,0,0}

\path[draw=drawColor,line width= 0.6pt,dash pattern=on 4pt off 4pt ,line join=round] ( 42.95,270.79) -- (572.66,270.79);
\end{scope}
\begin{scope}
\path[clip] (  0.00,  0.00) rectangle (578.16,433.62);
\definecolor{drawColor}{gray}{0.10}

\node[text=drawColor,anchor=base east,inner sep=0pt, outer sep=0pt, scale=  2.5] at ( 38.00, 40.68) {.600};

\node[text=drawColor,anchor=base east,inner sep=0pt, outer sep=0pt, scale=  1.7] at ( 38.00,111.18) {.601};

\node[text=drawColor,anchor=base east,inner sep=0pt, outer sep=0pt, scale=  1.7] at ( 38.00,181.68) {.602};

\node[text=drawColor,anchor=base east,inner sep=0pt, outer sep=0pt, scale=  1.7] at ( 38.00,252.17) {.603};

\node[text=drawColor,anchor=base east,inner sep=0pt, outer sep=0pt, scale=  1.7] at ( 38.00,322.67) {.604};

\node[text=drawColor,anchor=base east,inner sep=0pt, outer sep=0pt, scale=  1.7] at ( 38.00,393.16) {.605};
\end{scope}

\begin{scope}
\path[clip] (  0.00,  0.00) rectangle (578.16,433.62);
\definecolor{drawColor}{gray}{0.20}

\path[draw=drawColor,line width= 0.6pt,line join=round] ( 40.20, 43.72) --
	( 42.95, 43.72);

\path[draw=drawColor,line width= 0.6pt,line join=round] ( 40.20,114.21) --
	( 42.95,114.21);

\path[draw=drawColor,line width= 0.6pt,line join=round] ( 40.20,184.71) --
	( 42.95,184.71);

\path[draw=drawColor,line width= 0.6pt,line join=round] ( 40.20,255.20) --
	( 42.95,255.20);

\path[draw=drawColor,line width= 0.6pt,line join=round] ( 40.20,325.70) --
	( 42.95,325.70);

\path[draw=drawColor,line width= 0.6pt,line join=round] ( 40.20,396.19) --
	( 42.95,396.19);
\end{scope}
\begin{scope}
\path[clip] (  0.00,  0.00) rectangle (578.16,433.62);
\definecolor{drawColor}{gray}{0.20}

\path[draw=drawColor,line width= 0.6pt,line join=round] (142.27, 27.94) --
	(142.27, 30.69);

\path[draw=drawColor,line width= 0.6pt,line join=round] (307.81, 27.94) --
	(307.81, 30.69);

\path[draw=drawColor,line width= 0.6pt,line join=round] (473.34, 27.94) --
	(473.34, 30.69);
\end{scope}

\begin{scope}
\path[clip] (  0.00,  0.00) rectangle (578.16,433.62);
\definecolor{drawColor}{gray}{0.30}

\node[text=drawColor,anchor=base,inner sep=0pt, outer sep=0pt, scale=  2.5] at (138.70, 5) {$m=125$};

\node[text=drawColor,anchor=base,inner sep=0pt, outer sep=0pt, scale=  2.5] at (305.61, 5) {$m=500$};

\node[text=drawColor,anchor=base,inner sep=0pt, outer sep=0pt, scale=  2.5] at (472.52, 5) {$m=2000$};
\end{scope}

\end{tikzpicture}
		}
		&
		\resizebox{0.275\linewidth}{!}{    
		\begin{tikzpicture}[x=1pt,y=1pt]
\definecolor{fillColor}{RGB}{255,255,255}
\path[use as bounding box,fill=fillColor,fill opacity=0.00] (0,0) rectangle (573,430);
\begin{scope}
\path[clip] (  0.00,  0.00) rectangle (573,430);
\definecolor{drawColor}{RGB}{255,255,255}
\definecolor{fillColor}{RGB}{255,255,255}

\path[draw=drawColor,line width= 0.6pt,line join=round,line cap=round,fill=fillColor] (  0.00,  0.00) rectangle (573,430);
\end{scope}

\begin{scope}
\path[clip] ( 38.56, 30.69) rectangle (572.66,428.12);
\definecolor{fillColor}{gray}{0.92}

\path[fill=fillColor] ( 38.56, 30.69) rectangle (572.66,428.12);
\definecolor{drawColor}{RGB}{255,255,255}

\path[draw=drawColor,line width= 0.3pt,line join=round] ( 38.56, 83.13) --
	(572.66, 83.13);

\path[draw=drawColor,line width= 0.3pt,line join=round] ( 38.56,187.83) --
	(572.66,187.83);

\path[draw=drawColor,line width= 0.3pt,line join=round] ( 38.56,292.52) --
	(572.66,292.52);

\path[draw=drawColor,line width= 0.3pt,line join=round] ( 38.56,397.22) --
	(572.66,397.22);

\path[draw=drawColor,line width= 0.6pt,line join=round] ( 38.56, 30.78) --
	(572.66, 30.78);

\path[draw=drawColor,line width= 0.6pt,line join=round] ( 38.56,135.48) --
	(572.66,135.48);

\path[draw=drawColor,line width= 0.6pt,line join=round] ( 38.56,240.18) --
	(572.66,240.18);

\path[draw=drawColor,line width= 0.6pt,line join=round] ( 38.56,344.87) --
	(572.66,344.87);

\path[draw=drawColor,line width= 0.6pt,line join=round] (138.70, 30.69) --
	(138.70,428.12);

\path[draw=drawColor,line width= 0.6pt,line join=round] (305.61, 30.69) --
	(305.61,428.12);

\path[draw=drawColor,line width= 0.6pt,line join=round] (472.52, 30.69) --
	(472.52,428.12);
\definecolor{drawColor}{gray}{0.20}
\definecolor{fillColor}{gray}{0.20}

\path[draw=drawColor,line width= 0.4pt,line join=round,line cap=round,fill=fillColor] (138.70, 48.75) circle (  1.96);

\path[draw=drawColor,line width= 0.6pt,line join=round] (138.70,300.62) -- (138.70,356.83);

\path[draw=drawColor,line width= 0.6pt,line join=round] (138.70,203.54) -- (138.70, 80.77);
\definecolor{fillColor}{RGB}{255,255,255}

\path[draw=drawColor,line width= 0.6pt,line join=round,line cap=round,fill=fillColor] ( 76.11,300.62) --
	( 76.11,203.54) --
	(201.29,203.54) --
	(201.29,300.62) --
	( 76.11,300.62) --
	cycle;

\path[draw=drawColor,line width= 1.1pt,line join=round] ( 76.11,254.79) -- (201.29,254.79);
\definecolor{fillColor}{gray}{0.20}

\path[draw=drawColor,line width= 0.4pt,line join=round,line cap=round,fill=fillColor] (305.61,304.76) circle (  1.96);

\path[draw=drawColor,line width= 0.4pt,line join=round,line cap=round,fill=fillColor] (305.61,262.79) circle (  1.96);

\path[draw=drawColor,line width= 0.4pt,line join=round,line cap=round,fill=fillColor] (305.61,288.73) circle (  1.96);

\path[draw=drawColor,line width= 0.4pt,line join=round,line cap=round,fill=fillColor] (305.61,311.90) circle (  1.96);

\path[draw=drawColor,line width= 0.4pt,line join=round,line cap=round,fill=fillColor] (305.61,301.56) circle (  1.96);

\path[draw=drawColor,line width= 0.6pt,line join=round] (305.61,396.75) -- (305.61,409.32);

\path[draw=drawColor,line width= 0.6pt,line join=round] (305.61,363.09) -- (305.61,313.59);
\definecolor{fillColor}{RGB}{255,255,255}

\path[draw=drawColor,line width= 0.6pt,line join=round,line cap=round,fill=fillColor] (243.02,396.75) --
	(243.02,363.09) --
	(368.20,363.09) --
	(368.20,396.75) --
	(243.02,396.75) --
	cycle;

\path[draw=drawColor,line width= 1.1pt,line join=round] (243.02,381.95) -- (368.20,381.95);
\definecolor{fillColor}{gray}{0.20}

\path[draw=drawColor,line width= 0.4pt,line join=round,line cap=round,fill=fillColor] (472.52,379.01) circle (  1.96);

\path[draw=drawColor,line width= 0.4pt,line join=round,line cap=round,fill=fillColor] (472.52,403.35) circle (  1.96);

\path[draw=drawColor,line width= 0.4pt,line join=round,line cap=round,fill=fillColor] (472.52,388.17) circle (  1.96);

\path[draw=drawColor,line width= 0.4pt,line join=round,line cap=round,fill=fillColor] (472.52,392.76) circle (  1.96);

\path[draw=drawColor,line width= 0.4pt,line join=round,line cap=round,fill=fillColor] (472.52,401.55) circle (  1.96);

\path[draw=drawColor,line width= 0.4pt,line join=round,line cap=round,fill=fillColor] (472.52,387.70) circle (  1.96);

\path[draw=drawColor,line width= 0.4pt,line join=round,line cap=round,fill=fillColor] (472.52,398.81) circle (  1.96);

\path[draw=drawColor,line width= 0.4pt,line join=round,line cap=round,fill=fillColor] (472.52,401.62) circle (  1.96);

\path[draw=drawColor,line width= 0.4pt,line join=round,line cap=round,fill=fillColor] (472.52,392.95) circle (  1.96);

\path[draw=drawColor,line width= 0.4pt,line join=round,line cap=round,fill=fillColor] (472.52,399.54) circle (  1.96);

\path[draw=drawColor,line width= 0.4pt,line join=round,line cap=round,fill=fillColor] (472.52,400.12) circle (  1.96);

\path[draw=drawColor,line width= 0.4pt,line join=round,line cap=round,fill=fillColor] (472.52,394.66) circle (  1.96);

\path[draw=drawColor,line width= 0.4pt,line join=round,line cap=round,fill=fillColor] (472.52,396.83) circle (  1.96);

\path[draw=drawColor,line width= 0.4pt,line join=round,line cap=round,fill=fillColor] (472.52,403.74) circle (  1.96);

\path[draw=drawColor,line width= 0.4pt,line join=round,line cap=round,fill=fillColor] (472.52,387.68) circle (  1.96);

\path[draw=drawColor,line width= 0.4pt,line join=round,line cap=round,fill=fillColor] (472.52,400.72) circle (  1.96);

\path[draw=drawColor,line width= 0.4pt,line join=round,line cap=round,fill=fillColor] (472.52,400.19) circle (  1.96);

\path[draw=drawColor,line width= 0.6pt,line join=round] (472.52,409.63) -- (472.52,410.05);

\path[draw=drawColor,line width= 0.6pt,line join=round] (472.52,407.38) -- (472.52,404.48);
\definecolor{fillColor}{RGB}{255,255,255}

\path[draw=drawColor,line width= 0.6pt,line join=round,line cap=round,fill=fillColor] (409.92,409.63) --
	(409.92,407.38) --
	(535.11,407.38) --
	(535.11,409.63) --
	(409.92,409.63) --
	cycle;

\path[draw=drawColor,line width= 1.1pt,line join=round] (409.92,409.38) -- (535.11,409.38);
\definecolor{drawColor}{RGB}{0,0,0}

\path[draw=drawColor,line width= 0.6pt,dash pattern=on 4pt off 4pt ,line join=round] ( 38.56,409.36) -- (572.66,409.36);
\end{scope}
\begin{scope}
\path[clip] (  0.00,  0.00) rectangle (578.16,433.62);
\definecolor{drawColor}{gray}{0.30}

\node[text=drawColor,anchor=base east,inner sep=0pt, outer sep=0pt, scale=  1.7] at ( 33.61, 27.75) {0.68};

\node[text=drawColor,anchor=base east,inner sep=0pt, outer sep=0pt, scale=  1.7] at ( 33.61,132.45) {0.72};

\node[text=drawColor,anchor=base east,inner sep=0pt, outer sep=0pt, scale=  1.7] at ( 33.61,237.15) {0.76};

\node[text=drawColor,anchor=base east,inner sep=0pt, outer sep=0pt, scale=  1.7] at ( 33.61,341.84) {0.80};
\end{scope}

\begin{scope}
\path[clip] (  0.00,  0.00) rectangle (578.16,433.62);
\definecolor{drawColor}{gray}{0.20}

\path[draw=drawColor,line width= 0.6pt,line join=round] ( 35.81,129.38) --
	( 38.56,129.38);

\path[draw=drawColor,line width= 0.6pt,line join=round] ( 35.81,236.45) --
	( 38.56,236.45);

\path[draw=drawColor,line width= 0.6pt,line join=round] ( 35.81,343.51) --
	( 38.56,343.51);
\end{scope}
\begin{scope}
\path[clip] (  0.00,  0.00) rectangle (578.16,433.62);
\definecolor{drawColor}{gray}{0.20}

\path[draw=drawColor,line width= 0.6pt,line join=round] (138.70, 27.94) --
	(138.70, 30.69);

\path[draw=drawColor,line width= 0.6pt,line join=round] (305.61, 27.94) --
	(305.61, 30.69);

\path[draw=drawColor,line width= 0.6pt,line join=round] (472.52, 27.94) --
	(472.52, 30.69);
\end{scope}
\begin{scope}
\path[clip] (  0.00,  0.00) rectangle (578.16,433.62);
\definecolor{drawColor}{gray}{0.30}

\node[text=drawColor,anchor=base,inner sep=0pt, outer sep=0pt, scale=  2.5] at (138.70, 5) {$m=125$};

\node[text=drawColor,anchor=base,inner sep=0pt, outer sep=0pt, scale=  2.5] at (305.61, 5) {$m=500$};

\node[text=drawColor,anchor=base,inner sep=0pt, outer sep=0pt, scale=  2.5] at (472.52, 5) {$m=2000$};
\end{scope}

\end{tikzpicture}
		}
		&
		\resizebox{0.275\linewidth}{!}{    
		\begin{tikzpicture}[x=1pt,y=1pt]
\definecolor{fillColor}{RGB}{255,255,255}
\path[use as bounding box,fill=fillColor,fill opacity=0.00] (0,0) rectangle (573,430);
\begin{scope}
\path[clip] (  0.00,  0.00) rectangle (573,430);
\definecolor{drawColor}{RGB}{255,255,255}
\definecolor{fillColor}{RGB}{255,255,255}

\path[draw=drawColor,line width= 0.6pt,line join=round,line cap=round,fill=fillColor] (  0.00,  0.00) rectangle (573,430);
\end{scope}

\begin{scope}
\path[clip] ( 38.56, 30.69) rectangle (572.66,428.12);
\definecolor{fillColor}{gray}{0.92}

\path[fill=fillColor] ( 38.56, 30.69) rectangle (572.66,428.12);
\definecolor{drawColor}{RGB}{255,255,255}

\path[draw=drawColor,line width= 0.3pt,line join=round] ( 38.56,114.38) --
	(572.66,114.38);

\path[draw=drawColor,line width= 0.3pt,line join=round] ( 38.56,227.79) --
	(572.66,227.79);

\path[draw=drawColor,line width= 0.3pt,line join=round] ( 38.56,341.21) --
	(572.66,341.21);

\path[draw=drawColor,line width= 0.6pt,line join=round] ( 38.56, 57.67) --
	(572.66, 57.67);

\path[draw=drawColor,line width= 0.6pt,line join=round] ( 38.56,171.09) --
	(572.66,171.09);

\path[draw=drawColor,line width= 0.6pt,line join=round] ( 38.56,284.50) --
	(572.66,284.50);

\path[draw=drawColor,line width= 0.6pt,line join=round] ( 38.56,397.92) --
	(572.66,397.92);

\path[draw=drawColor,line width= 0.6pt,line join=round] (138.70, 30.69) --
	(138.70,428.12);

\path[draw=drawColor,line width= 0.6pt,line join=round] (305.61, 30.69) --
	(305.61,428.12);

\path[draw=drawColor,line width= 0.6pt,line join=round] (472.52, 30.69) --
	(472.52,428.12);
\definecolor{drawColor}{gray}{0.20}
\definecolor{fillColor}{gray}{0.20}

\path[draw=drawColor,line width= 0.4pt,line join=round,line cap=round,fill=fillColor] (138.70, 48.75) circle (  1.96);

\path[draw=drawColor,line width= 0.6pt,line join=round] (138.70,215.40) -- (138.70,282.63);

\path[draw=drawColor,line width= 0.6pt,line join=round] (138.70,158.42) -- (138.70,110.91);
\definecolor{fillColor}{RGB}{255,255,255}

\path[draw=drawColor,line width= 0.6pt,line join=round,line cap=round,fill=fillColor] ( 76.11,215.40) --
	( 76.11,158.42) --
	(201.29,158.42) --
	(201.29,215.40) --
	( 76.11,215.40) --
	cycle;

\path[draw=drawColor,line width= 1.1pt,line join=round] ( 76.11,186.77) -- (201.29,186.77);

\path[draw=drawColor,line width= 0.6pt,line join=round] (305.61,321.87) -- (305.61,342.45);

\path[draw=drawColor,line width= 0.6pt,line join=round] (305.61,296.08) -- (305.61,272.69);

\path[draw=drawColor,line width= 0.6pt,line join=round,line cap=round,fill=fillColor] (243.02,321.87) --
	(243.02,296.08) --
	(368.20,296.08) --
	(368.20,321.87) --
	(243.02,321.87) --
	cycle;

\path[draw=drawColor,line width= 1.1pt,line join=round] (243.02,310.70) -- (368.20,310.70);

\path[draw=drawColor,line width= 0.6pt,line join=round] (472.52,370.93) -- (472.52,380.62);

\path[draw=drawColor,line width= 0.6pt,line join=round] (472.52,357.95) -- (472.52,339.49);

\path[draw=drawColor,line width= 0.6pt,line join=round,line cap=round,fill=fillColor] (409.92,370.93) --
	(409.92,357.95) --
	(535.11,357.95) --
	(535.11,370.93) --
	(409.92,370.93) --
	cycle;

\path[draw=drawColor,line width= 1.1pt,line join=round] (409.92,365.86) -- (535.11,365.86);
\definecolor{drawColor}{RGB}{0,0,0}

\path[draw=drawColor,line width= 0.6pt,dash pattern=on 4pt off 4pt ,line join=round] ( 38.56,410.05) -- (572.66,410.05);
\end{scope}
\begin{scope}
\path[clip] (  0.00,  0.00) rectangle (578.16,433.62);
\definecolor{drawColor}{gray}{0.30}

\node[text=drawColor,anchor=base east,inner sep=0pt, outer sep=0pt, scale=  1.7] at ( 33.61, 54.64) {0.16};

\node[text=drawColor,anchor=base east,inner sep=0pt, outer sep=0pt, scale=  1.7] at ( 33.61,168.06) {0.18};

\node[text=drawColor,anchor=base east,inner sep=0pt, outer sep=0pt, scale=  1.7] at ( 33.61,281.47) {0.20};

\node[text=drawColor,anchor=base east,inner sep=0pt, outer sep=0pt, scale=  1.7] at ( 33.61,394.89) {0.22};
\end{scope}

\begin{scope}
\path[clip] (  0.00,  0.00) rectangle (578.16,433.62);
\definecolor{drawColor}{gray}{0.20}

\path[draw=drawColor,line width= 0.6pt,line join=round] ( 35.81,133.79) --
	( 38.56,133.79);

\path[draw=drawColor,line width= 0.6pt,line join=round] ( 35.81,264.91) --
	( 38.56,264.91);

\path[draw=drawColor,line width= 0.6pt,line join=round] ( 35.81,396.02) --
	( 38.56,396.02);
\end{scope}
\begin{scope}
\path[clip] (  0.00,  0.00) rectangle (578.16,433.62);
\definecolor{drawColor}{gray}{0.20}

\path[draw=drawColor,line width= 0.6pt,line join=round] (138.70, 27.94) --
	(138.70, 30.69);

\path[draw=drawColor,line width= 0.6pt,line join=round] (305.61, 27.94) --
	(305.61, 30.69);

\path[draw=drawColor,line width= 0.6pt,line join=round] (472.52, 27.94) --
	(472.52, 30.69);
\end{scope}

\begin{scope}
\path[clip] (  0.00,  0.00) rectangle (578.16,433.62);
\definecolor{drawColor}{gray}{0.30}

\node[text=drawColor,anchor=base,inner sep=0pt, outer sep=0pt, scale=  2.5] at (138.70, 5) {$m=125$};

\node[text=drawColor,anchor=base,inner sep=0pt, outer sep=0pt, scale=  2.5] at (305.61, 5) {$m=500$};

\node[text=drawColor,anchor=base,inner sep=0pt, outer sep=0pt, scale=  2.5] at (472.52, 5) {$m=2000$};
\end{scope}

\end{tikzpicture}
		}

\end{tabular}
\end{center}
	\caption{Boxplot of Monte Carlo estimates of the sample depths (Algorithm \ref{alg:montecarlo}) for curves $\Cc_r$ with respect to a sample of $n=25$ concentric circles over $100$ replications : (left) $r=0.1$, (middle) $r=0.4$, (right) $r=0.9$.  The Monte Carlo sample sizes considered are $m\in\{125,500,2000\}.$ The horizontal dotted line is at the value $D(\Cc_r|\Xx_1,\ldots,\Xx_n)$.}\label{app:fig:circles:m:pack}
\end{figure}

Table~\ref{app:tab:circles:n} illustrates the convergence in probability of the sample depth $D(\Cc|\Xx_1,\ldots,\Xx_n)$ and of its Monte Carlo approximation (see Algorithm~\ref{alg:montecarlo}) to $D(\Cc|P)$ as $n\to \infty$ over $5,000$ replications. 
For every replicated sample, we compute the depth of $\Cc_r$ using Equation \eqref{app:eq:concetric:depth:sample} and its approximation using Algorithm \ref{alg:montecarlo}.
The sample size $m$ in Algorithm \ref{alg:montecarlo} is set to~$500.$
As expected, both the empirical bias of the sample \curvedepth{} and its empirical standard deviation converge to zero for every value of $r$. The sample \curvedepth{} is on average smaller than the population depth.
Moreover the standard deviation of $D(\Cc_r|\Xx_1,\ldots,\Xx_n)$ is a function of the radius.

Since the Monte Carlo approximation tends to underestimate the sample depth, its average is expected to be smaller than the population depth : the bias of the Monte Carlo approximation is then larger than that of the sample \curvedepth{}.
We may expect a larger variance for the Monte Carlo estimate. 
In this example the variability due to the Monte Carlo estimate is quite weak, and tends to be equivalent to the variability of the sample \curvedepth{} for a large enough sample size.

\begin{table}[H]
\small\centering
\begin{tabular}{|l|l|ccccc|}
\hline
	&         $r$    & 0.1  & 0.4  & 0.5  & 0.6  & 0.9 \\ \hline
	&$D(\Cc_r|P)$    &\bf 0.627 &\bf 0.830 &\bf 0.758 &\bf 0.629 &\bf 0.169\\
\hline
$n=50$ 	& $D(\Cc_r|\Xx_1,\ldots,\Xx_n)$  & 0.627      & 0.808      & 0.744      & 0.617      & 0.161 \\
        &                                &\it (0.020) &\it (0.031) &\it (0.071) &\it (0.083) &\it (0.062)\\[0.2cm]
    	& MC-estimate                    & 0.627      & 0.782      & 0.697      & 0.574      & 0.144\\ 
		&                                &\it (0.020) &\it (0.041) &\it (0.071) &\it (0.078) &\it (0.056)\\ \hline
$n=200$ & $D(\Cc_r|\Xx_1,\ldots,\Xx_n)$  & 0.627      & 0.825      & 0.756      & 0.628      & 0.167 \\
		&                                &\it (0.010) &\it (0.011) &\it (0.037) &\it (0.042) &\it (0.031)\\[0.2cm]
    	& MC-estimate                    & 0.627      & 0.799      & 0.705      & 0.581      & 0.149\\ 
		&                                &\it (0.010) &\it (0.022) &\it (0.038) &\it (0.040) &\it (0.028)\\ \hline	
$n=800$ & $D(\Cc_r|\Xx_1,\ldots,\Xx_n)$  & 0.627      & 0.830      & 0.758      & 0.629      & 0.169\\
        &                                &\it (0.005) &\it (0.005) &\it (0.019) &\it (0.021) &\it (0.015)\\[0.2cm]
    	& MC-estimate                    & 0.627      & 0.806      & 0.707      & 0.582      & 0.151\\ 
		&                                &\it (0.005) &\it (0.015) &\it (0.021) &\it (0.022) &\it (0.014)\\ \hline
\end{tabular}
\caption{Average of sample depths and their Monte Carlo approximations (with $m=500$) for the curves $\Cc_r$ and their corresponding standard deviations (in parentheses) with respect to increasing sample sizes $n$ over $5,000$ replications.}\label{app:tab:circles:n}
\end{table}

\paragraph{Scheme 2 : \citep[see paragraph 4.2.1 in][]{claeskens2014multivariate}.}
We consider an i.i.d. sample $\Xx_1,\ldots \Xx_n,$  from the process $\Xx$ proposed by \citep[see paragraph 4.2.1 in][]{claeskens2014multivariate},
$$
\Xx = \left\{ \left(x, A_1\sin(2\pi x) + A_2\cos(2\pi x)\right);\, x\in[L,U]\right\}\,,
$$
where $A_1,A_2\sim\Uu[0,0.05]$, $L\sim\Uu[0,\frac{2\pi}{3}]$, and $U\sim\Uu[\frac{4\pi}{3},2\pi]$, all independent. 
The mean of the process is denoted as $\Cc_\Xx$ with the value $a_1 = a_2= 0.025$ of $A_1$ and $A_2$. 
Using the Monte Carlo approximation for a fixed values of $m$ and $\Delta$, we compute the depth of $\Cc_{\Xx}$.
We repeat the experiment $1,000$ times for different values of $n,$ $m$ and $\Delta$. 

First we fix the $n=50$ curves of the sample, and we aim to measure the effect of $m$ and $\Delta$ on the computation of the depth of $\Cc_\Xx$ given the sample $\Xx_1,\ldots, \Xx_n.$
Table~\ref{tab:sim1} indicates the average depths and their standard deviations (in parentheses) for the curve $\Cc_\Xx$ for different choices of $m$ and $\Delta.$
From the simulations, the threshold $\Delta$ in the chosen range seems to have very limited influence on the estimated depth value.
Further simulations indicate that averages of the depths converge (the consecutive differences decrease) and their standard deviations decrease towards zero as $m$ increases as expected from Theorem~3.1. 
Figure~\ref{fig:sim1} (right) indicates that a subsample of deepest curves is located nearby the center of the stochastic process.

\begin{figure}[H]
	\begin{center}
		\includegraphics[keepaspectratio=true,width=0.425\linewidth, trim  =10mm 15mm 10mm 10mm, clip]{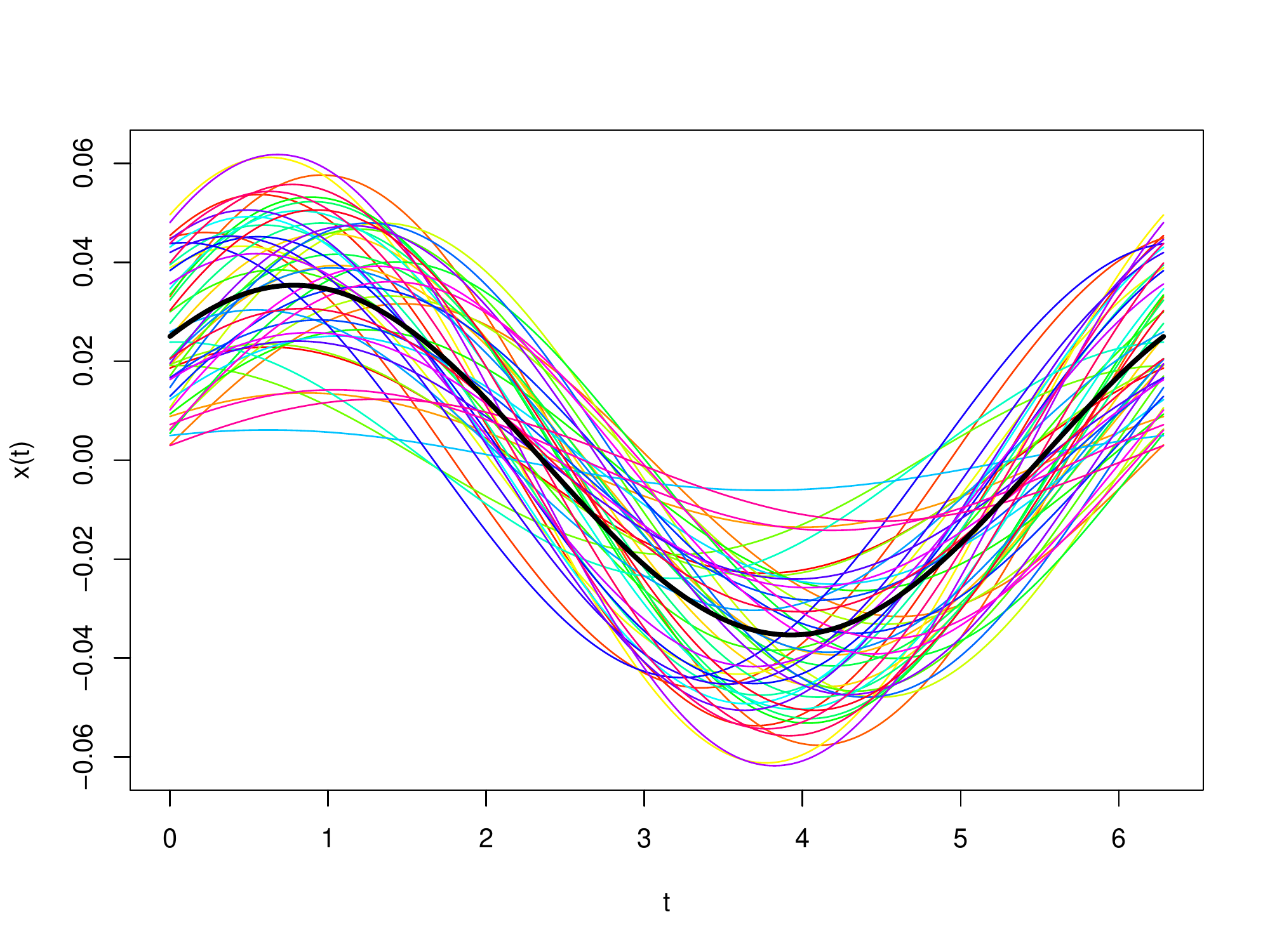}\quad
		\includegraphics[keepaspectratio=true,width=0.425\linewidth, trim  =10mm 15mm 10mm 10mm, clip]{pic-fig3b.pdf}
	\end{center}
	\caption{Illustration of the sample of $n=50$ curves for Simulation 1. 
	In the left panel, the curves $\Xx_i$ are plotted in different colors and the mean curve $\Cc_\Xx$ in black.
	In the right panel, the curves with depth larger than $0.727-2\times 0.014$ (with $m=500$ and $\alpha=1/8$, see Table \ref{tab:sim1}) are plotted in orange, where $0.727$ is the depth of the mean curve, the deepest curve having depth $0.744$ in red.}
	\label{fig:sim1}
\end{figure}

\begin{table}
	 \caption{Average depths for the mean curve $\Cc_\Xx$ and their corresponding standard deviations (in parentheses) with respect to a sample of size $n=50$ for Simulation~1 over $N=1,000$ Monte Carlo repetitions with differing $m$ and $\Delta = 1/(10m^\alpha).$}
	 \label{tab:sim1}
	 \begin{center}
	 \footnotesize
	 	\begin{tabular}{c|ccccccccc}
			$\alpha \backslash m$ & 20 & 50 & 100 & 200 & 500 & 1000 & 2000 & 5000 & 10000 \\ \hline\hline
	 		0 & 0.616 & 0.685 & 0.721 & 0.745 & 0.764 & 0.772 & 0.777 & 0.781 & 0.782 \\
			& (0.055) & (0.040) & (0.030) & (0.023) & (0.015) & (0.010) & (0.007) & (0.004) & (0.003) \\ \hline
			1/8 & 0.62 & 0.69 & 0.726 & 0.749 & 0.767 & 0.774 & 0.778 & 0.781 & 0.782 \\
			& (0.056) & (0.040) & (0.029) & (0.022) & (0.014) & (0.010) & (0.006) & (0.004) & (0.003) \\ \hline
			1/4 & 0.615 & 0.686 & 0.723 & 0.747 & 0.766 & 0.773 & 0.778 & 0.781 & 0.782 \\
			& (0.054) & (0.039) & (0.031) & (0.022) & (0.014) & (0.010) & (0.007) & (0.004) & (0.003) \\
		\end{tabular}
	 \end{center}
       \end{table}

Lastly, we aim to measure the effect of the size $n$ of the sample of curves on the computation of the depth.
Here, we sample $m=1,000$ points on each curve.
For every Monte Carlo replication, we generate a new i.i.d.\ sample $\Xx_1,\ldots, \Xx_n$ from the process $\Xx$ defined above. 
Table~\ref{tab:sim3} shows the average depths and their standard deviations (in parentheses) of the curve $\Cc_\Xx.$
We can see that the depth of $\Cc_\Xx$ converges as expected in Theorem~3.1.
Note that, compared to Table~\ref{tab:sim1}, the standard deviations take into account additionally the variation of the curves' sample, \textit{cf.} $0.010$ for $\alpha=1/8,$ $m=1000$ in Table~\ref{tab:sim1}  and $0.023$ for $n=50$ in Table~\ref{tab:sim3}.

\begin{table}
	 \caption{Average depths for the mean curve and their corresponding standard deviations (in parentheses) for Simulation 1 over $N=1,000$ Monte Carlo repetitions, $m=1,000$, and $\alpha = 1/8$ with growing $n$.}
	 \label{tab:sim3}
	 \begin{center}
	 	\footnotesize
	 	\begin{tabular}{c|ccccccccc}
	 		$n$ & 20 & 50 & 100 & 200 & 500 & 1000 & 2000 & 5000 & 10000 \\ \hline\hline
			$\Cc_\Xx$  & 0.730 & 0.751 & 0.759 & 0.762 & 0.764 & 0.765 & 0.766 & 0.765 & 0.766 \\
			  & (0.044) & (0.023) & (0.016) & (0.012) & (0.009) & (0.009) & (0.008) & (0.008) & (0.007) \\
		\end{tabular}
	 \end{center}
       \end{table}

\paragraph{Simulation 2 : \cite{cuevas2007models}.}
We consider an i.i.d. sample $\Yy_1,\ldots, \Yy_n$ from the process $\Yy$ proposed by~\cite{cuevas2007models}:
$$
\Yy = \left\{ \left(x, 30(1 - x)^{1 + W}x^{1.5 - W} + U_x\right);\,  x\in [L,U]\right\}
$$
where $\{U_t;\, t\in[0,1]\}$ is a zero mean stationary Gaussian process with covariance function $t \mapsto  0.2 e^{-\frac{1}{0.3}|t|}$, $W\sim\Uu[0,0.5]$, $L\sim\Uu[0,0.1]$, $U\sim\Uu[0.9,1]$, all independent. 
The mean of the process is denoted as $\Cc_\Yy$ with the parametrization $\bmy(t)=15(1-t)t(\sqrt{1-t}-\sqrt{t})/(\log(1-t)-\log(t)).$
Notice that since the curves $\Yy_i$ are noisy, they may not be in fact rectifiable.
While in practice such curves are discretely observed, they can be approximated by affine functions with a finite length.

Similarly to the last scheme, we fix the $n=50$ curves of the samples. 
In Figure~\ref{fig:sim2} (right), we depict the first Monte Carlo replication with the deepest curve $\Yy_{1}$ and a subsample of curves with a depth closest to it in depth.
Note that the depth of $\Cc_\Yy$ is  around $0.6$ while the depth of the deepest curve in the sample is about $0.8.$ 
Although the mean curve $\Cc_\Yy$ is fairly central, the deepest curve $\Yy_{1}$ is a better representative of the sample because of the smoothness of $\Cc_\Yy.$
Table~\ref{tab:sim2} indicates the average depths and their standard deviations (in parentheses) for the curves $\Cc_\Yy$ and $\Yy_1$ (notice that due to the variation of the points on the curves, $\Yy_1$ is not always the deepest curve for all Monte Carlo replications).
Even though the standard deviations are of the same order as in the previous simulation, we remark that the depths of $\Yy_1$ are twice more dispersed than those of $\Cc_\Yy.$

\begin{figure}[H]
	\begin{center}
		\includegraphics[keepaspectratio=true,width=0.425\linewidth, trim  =10mm 15mm 10mm 10mm, clip]{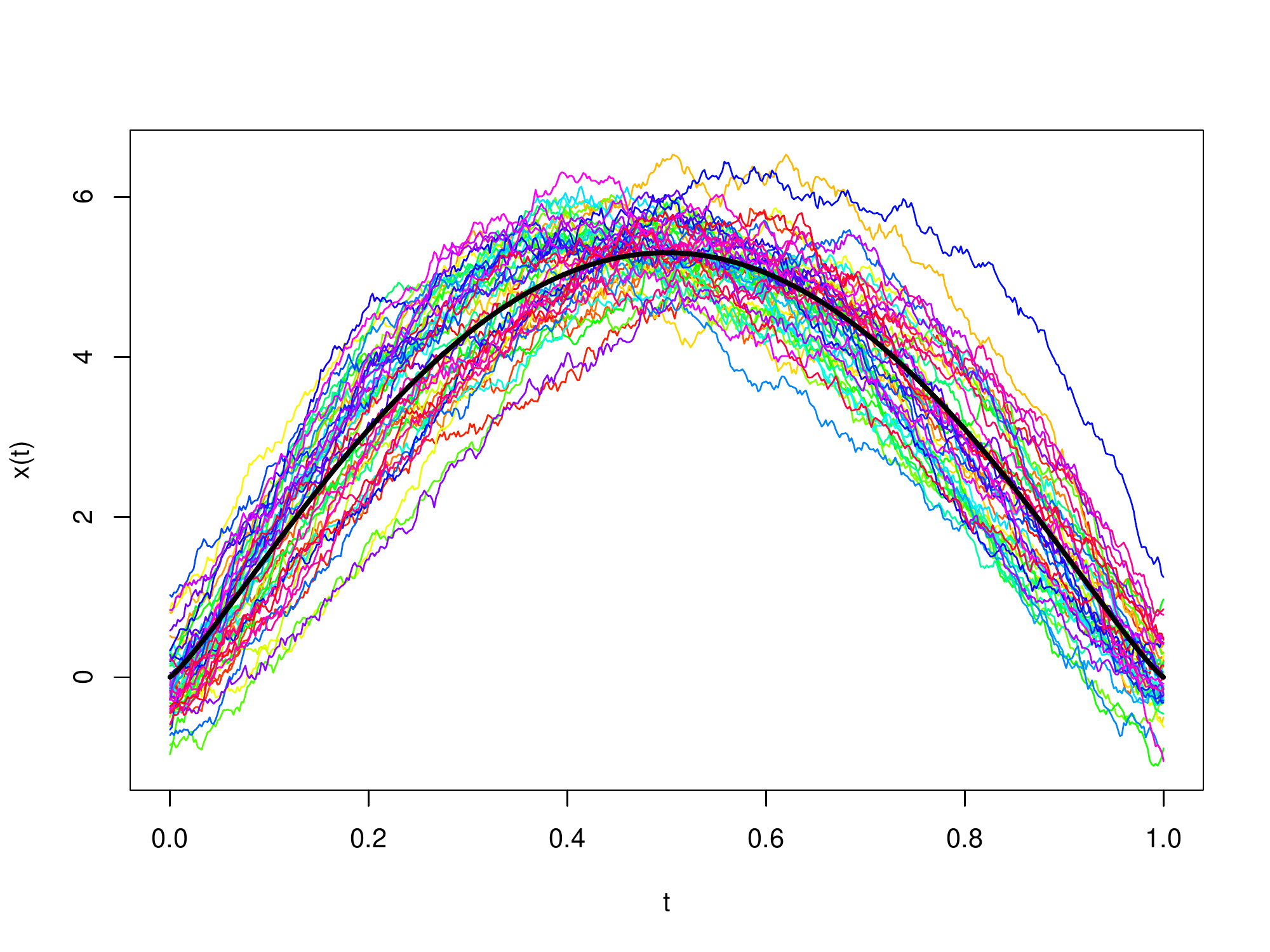}\quad
		\includegraphics[keepaspectratio=true,width=0.425\linewidth, trim  =10mm 15mm 10mm 10mm, clip]{pic-fig4b.pdf}
	\end{center}
	\caption{Illustration of the sample of $n=50$ curves for Simulation 2. 
	In the left panel, the curves $\Yy_i$ are plotted in different colors and the mean curve $\Cc_\Yy$ in black.
	In the right panel, the curves with depth larger than $0.752-2\times 0.026$ (with $m=500$ and $\alpha=1/8$, see Table \ref{tab:sim2}) are plotted in orange, where the deepest curve having depth $0.752$ is plotted in red,
the depth of the mean curve is $0.571.$
}
	\label{fig:sim2}
\end{figure}

\begin{table}
	 \caption{Average depths for the mean curve $\Cc_\Yy$ and the deepest curve $\Yy_{1}$ and their corresponding standard deviations (in parentheses) for Simulation 2 over $N=1,000$ Monte Carlo repetitions, $n=50$ curves and $\alpha=1/8$.}
	 \label{tab:sim2}
 	 \begin{center}
 	 \footnotesize
	 \begin{tabular}{c|ccccccccc}
			$m$ & 20 & 50 & 100 & 200 & 500 & 1000 & 2000 & 5000 & 10000 \\ \hline\hline
	 		$\Cc_\Yy$ & 0.497 & 0.54 & 0.557 & 0.568 & 0.576 & 0.58 & 0.583 & 0.584 & 0.585 \\
			& (0.055) & (0.039) & (0.031) & (0.021) & (0.014) & (0.010) & (0.007) & (0.004) & (0.003) \\ \hline
			$\Yy_{1}$ & 0.633 & 0.694 & 0.726 & 0.746 & 0.761 & 0.768 & 0.772 & 0.774 & 0.774 \\
			& (0.069) & (0.061) & (0.049) & (0.038) & (0.026) & (0.018) & (0.013) & (0.008) & (0.006) \\
		\end{tabular}
	 \end{center}
\end{table}

\section{Pre-processing of DTI Scans}\label{sec:preprocess.DTI}

DTI scans were acquired from all 34 twin pairs on a Philips 3T Achieva Quasar Dual MRI scanner (Philips Medical System, Best, The Netherlands), using a single-shot echo-planar imaging (EPI) sequence (TR = 7115 ms, TE = 70 ms). For each diffusion scan, 32 gradient directions (b = 1000 s/$\text{mm}^2$) and a non-diffusion-weighted acquisition (b = 0 s/$\text{mm}^2$) were acquired over a 96$\text{mm}^2$ image matrix (FOV 240 mm $\times$ 240 $\text{mm}^2$); with a slice thickness of 2.5 mm and no gap, yielding 2.5 mm isotropic voxels.

We used the MRtrix software \citep{tournier2012mrtrix} to extract fiber tracts from the DTI scans and we chose corticospinal tract (both left and right) here because corticospinal tracts are long and could be identified and extracted relatively accurately and reliably in comparison to other shorter and more ambiguous fiber tracts of the human brain. We used the seed region of interest on an axial slice on which the cerebral peduncle was visible. The resulting data sets were two bundles of around $1,000$ fibers each per subject. Each fiber was described by a set of around $400$ successive 3D locations. The size of a single file containing only one bundle of fibers was around 12MB, so that altogether the $34\times2\times2$ files weigh around 1.6GB.

\end{document}